\documentclass[11pt]{article} 
	\usepackage[margin=1 in, bottom=1 in]{geometry} 
	\usepackage[english]{babel} 
	\usepackage{authblk}
	\usepackage{amsmath, amssymb, textcomp, amsthm, mathtools,mathrsfs}  
	\usepackage{tensor}
	\usepackage{blkarray}
	\usepackage{wasysym} 
	\usepackage{stmaryrd}
	\usepackage{dsfont}
	\usepackage{hyperref}
   	\usepackage{natbib}
	\usepackage[utf8]{inputenc} 
	\usepackage{csquotes}
	\usepackage{footnote,footmisc} 
	\usepackage{setspace} 
	\usepackage{threeparttable}
	\usepackage{multirow}
	\usepackage{tikz}
	\usepackage{pgfplots}
	\usetikzlibrary{shapes.misc}
	\usetikzlibrary{patterns}
	\usepackage{xcolor}
	\hypersetup{
    	colorlinks,
    	linkcolor={red!50!black},
    	citecolor={blue!50!black},
    	urlcolor={blue!80!black}
	}
	\usetikzlibrary{decorations.pathreplacing}
	\usetikzlibrary{arrows}
	\usepackage{float} 
	\newcommand{\GG}[1]{}
   	 \allowdisplaybreaks
	\usepackage[toc,page]{appendix}

	\newtheorem{theorem}{Theorem}
	\newtheorem{lemma}{Lemma} 
	\newtheorem{proposition}{Proposition} 
	
	\newtheorem*{conjecture*}{Conjecture}

	\newtheorem{assumption}{Assumption}

	\newtheorem*{example*}{Running Example}
	\newtheorem*{algorithm*}{Algorithm}
	\usepackage{graphicx}  
	\usepackage{marvosym} 
	\makeatletter
\def\moverlay{\mathpalette\mov@rlay}
\def\mov@rlay#1#2{\leavevmode\vtop{%
   \baselineskip\z@skip \lineskiplimit-\maxdimen
   \ialign{\hfil$\m@th#1##$\hfil\cr#2\crcr}}}
\newcommand{\charfusion}[3][\mathord]{
    #1{\ifx#1\mathop\vphantom{#2}\fi
        \mathpalette\mov@rlay{#2\cr#3}
      }
    \ifx#1\mathop\expandafter\displaylimits\fi}
\makeatother

\newcommand\independent{\protect\mathpalette{\protect\independenT}{\perp}}
\def\independenT#1#2{\mathrel{\rlap{$#1#2$}\mkern2mu{#1#2}}}
    
    \DeclareMathOperator*{\argmax}{\arg\!\max}

    \newtheorem*{assumptions*}{\assumptionnumber}
\providecommand{\assumptionnumber}{}


\title{A path-sampling method to partially identify causal effects in instrumental variable models}
\author{Florian F Gunsilius}
\affil{MIT}
\date{\today}
\begin{document}
\maketitle
\begin{abstract}
Partial identification approaches are a flexible and robust alternative to standard point-identification approaches in general instrumental variable models. However, this flexibility comes at the cost of a ``curse of cardinality'':  the number of restrictions on the identified set grows exponentially with the number of points in the support of the endogenous treatment. This article proposes a novel path-sampling approach to this challenge. It is designed for partially identifying causal effects of interest in the most complex models with continuous endogenous treatments. A stochastic process representation allows to seamlessly incorporate assumptions on individual behavior into the model. Some potential applications include dose-response estimation in randomized trials with imperfect compliance, the evaluation of social programs, welfare estimation in demand models, and continuous choice models. As a demonstration, the method provides informative nonparametric bounds on household expenditures under the assumption that expenditure is continuous. The mathematical contribution is an approach to approximately solving infinite dimensional linear programs on path spaces via sampling.
\end{abstract}

\section{Introduction}
In recent years, a trend in the literature in econometrics and in particular causal inference has been to obtain bounds on quantities of interest in general instrumental variable models via linear programming approaches, often in connection with capacity- and random set theory (e.g.~\citet*{balke1994counterfactual}, \citet*{balke1997bounds}, \citet*{manski2007partial},  \citet*{kitamura2018nonparametric}, \citeauthor*{molchanov2005theory} \citeyear{molchanov2005theory}, \citeauthor*{beresteanu2011sharp} \citeyear{beresteanu2011sharp}, \citeauthor*{galichon2011set} \citeyear{galichon2011set}, \citeauthor*{beresteanu2012partial} \citeyear{beresteanu2012partial}, \citeauthor*{molchanov2014applications} \citeyear{molchanov2014applications}, \citeauthor*{chesher2017generalized} \citeyear{chesher2017generalized}, \citeauthor*{russell2019sharp} \citeyear{russell2019sharp}). The arguments put forward in favor of these partial identification approaches are higher flexibility and robustness compared to point-identification approaches (\citet*{chesher2017generalized}, \citet*{manski2003partial}). However, these general partial identification approaches suffer from a severe ``curse of cardinality'' which limits their broader use in general models, as noted in \citet*{beresteanu2012partial}. 

The issue is that the number of restrictions on the identified set grows at least exponentially with the number of points in the support of the endogenous variables. Existing methods are intractable in practical settings with high-cardinal endogenous variables\footnote{We say a random variable is ``of high cardinality'' if it has significantly more than $2$ points in its support, or is continuous. As we show below, even as few as $5$ support points satisfy this definition in the most general instrumental variable setting already.} and have therefore almost exclusively focused on the case of a binary endogenous variable or have made use of the given particular structure of the problem (e.g.~\citet*{aguiar2019prices}, \citet*{chesher2017generalized}, \citet*{chiburis2010semiparametric}, \citet*{cheng2006bounds}, \citet*{demuynck2015bounding}, \citet*{hansen1995econometric}, \citet*{honore2006bounds}, \citet*{honore2006bounds2},  \citet*{manski2007partial}, \citet*{manski2014identification}, \citet*{molinari2008partial}, \citet*{norets2013semiparametric}, \citet*{laffers2015bounding}, \citet*{kamat2017identification}, \citet*{torgovitsky2016nonparametric}, \citet*{kitamura2018nonparametric}, \citet*{mogstad2018using}, \citet*{russell2019sharp}, \citet*{tebaldi2019nonparametric}). A practical and generally applicable method that deals with this curse of cardinality has so far been unavailable. 

This paper introduces such a method. It extends the linear programming approach to partial identification of functionals of interest to high-cardinality settings, even allowing for continuous variables. The basic idea is to represent the instrumental variable model as a system of stochastic processes indexed by the unobservable heterogeneity, based on the potential outcome notation for causal models \citep*{rubin1974estimating}. This representation as stochastic processes is useful, as it opens the door for the application of all the tools from time-series analysis and stochastic calculus to this setting. A general, in the case of continuous endogenous variables infinite dimensional, linear program is then constructed on the paths of these processes, generalizing the linear programming approach for binary causal models in \citet*{balke1997bounds}. This approach does not require structural assumptions, but allows to incorporate them seamlessly into the model by ruling out paths of the stochastic processes: the idea is that each path denotes a hypothetical participant in the model, so that the stochastic framework approach is a simple way to explicitly model the unobserved heterogeneity in terms of human behavior. 

The main contribution of this article is a computational approach to approximately solve these potentially infinite-dimensional linear programs on path spaces. The idea is to sample paths of the stochastic processes and to solve the linear programs on this sample of paths. This introduced randomness is crucial, because it permits the derivation of probabilistic approximation guarantees of theoretical linear program by its sampled counterpart using standard concentration results (\citeauthor*{vapnik1998statistical} \citeyear{vapnik1998statistical}, \citeauthor*{wellner2013weak} \citeyear{wellner2013weak}). This approach could be helpful in many other settings with optimization problems on an infinite dimensional state space. The key is that one can straightforwardly define a measure on the paths of these processes and use the theory of stochastic processes to analyze the properties, especially in continuous models. In a nutshell, we use randomization to alleviate the ``curse of cardinality'' in linear programs on path spaces, which is similar in spirit to using randomization to ``break the curse of dimensionality'' in classical dynamic optimization problems \citep*{rust1997using}. In this article we argue that our approach is an attractive alternative to other sampling approaches such as \citet*{de2004constraint}, because it works on paths of processes instead of sets of inequalities and is therefore more flexible.

As a demonstration of its capabilities, the method estimates bounds on expenditure differences using the 1995/1996 UK family expenditure survey. This problem is well suited for demonstrating the method's capabilities in practice as it is nonlinear with continuous variables and allows to gauge if the program actually obtains reasonable results.\footnote{In particular, the method should produce results which show that food is a necessity- and leisure is a luxury good, as these are well-established economic facts. This application is actually more challenging than a Monte-Carlo approach, as the method needs to replicate known facts on real data under minimal assumptions (see \citeauthor*{advani2019mostly} \citeyear{advani2019mostly} for a recent discussion). A priori, it is not even clear that \emph{any} approach can deliver informative enough bounds to check its validity. The fact that this method does provide informative bounds is a testament to its potential usefulness.} Surprisingly, the method already seems to provide informative nonparametric bounds on household expenditures under the sole assumption that expenditure on goods is continuous with respect to the budget set, corroborating the nonparametric and semi-nonparametric approaches in \citet*{imbens2009identification}, \citet*{blundell2007semi}, \citet*{de2016nonparametric}, and \citet*{song2018nonseparable}, which assume monotonicity or additive separability in the unobservables in the first- or second stage. 

The outline of this article is as follows. Section \ref{intuitionsec} introduces the setting and the optimization problem at an intuitive level, highlighting three potential settings for applications: randomized controlled trials with imperfect compliance, time-varying treatment effects in program evaluation, and welfare estimation. Section \ref{sec:mainsection} is the main section of this article and introduces all theoretical results. It contains introduces the stochastic process representation of the causal model formally (Proposition \ref{myprop} in section \ref{stochasticpathsec}), states the main result of this article which quantifies the probabilistic approximation of the optimization problems by path sampling (Theorem \ref{maintheorem2} in section \ref{mainsubsec}), and derives the large sample asymptotics of the optimization programs (Proposition \ref{largesampleprop} in section \ref{inferencesec}). Section \ref{practicalimp} deals with the practical implementation and introduces the practical algorithm. Section \ref{empiricalsec} contains empirical results: section \ref{simulationsubsec} contains a brief simulation exercise and section \ref{applicationsubsec} contains the application to demand estimation. Section \ref{conclusion} concludes. The appendix contains all proofs.

\section{Intuitive setup of the linear programming approach and examples}\label{intuitionsec}
\subsection{The model considered}
The basic model in this article is the most general form of an instrumental variable model. This means it makes no a priori assumptions on the functional relations between the variables and leaves the unobserved heterogeneity unrestricted. It is
\begin{equation}\label{mainmodel}
\begin{aligned}
Y &= h(X,W)\\
X &=g(Z,W)\qquad Z\independent W,
\end{aligned}
\end{equation}
where, $Y$ is the outcome variable of interest, $X$ is the treatment variable of interest. $X$ is endogenous in the sense that it depends on the unobservable confounder $W$, which also has an influence on $Y$, in addition to the direct effect of $X$ on $Y$. This additional pathway induces a bias when trying to estimate a causal effect of $X$ on $Y$, the classical problem in causal inference and econometrics. A general solution for this is to use an instrumental variable $Z$, which affects the treatment $X$, but is itself exogenous, i.e.~independent of the model. $Z$ satisfies full independence of the unrestricted unobserved heterogeneity $W$, denoted by $Z\independent W$, as the production functions $g$ and $h$ are unrestricted.\footnote{$Z\independent W$ means that $Z$ is independent of $W$, i.e.~$P_{Z,W} = P_ZP_{W}$. This model is often written with two separate unobservable variables $V$ and $U$ in the second- and first stage (e.g.~\citeauthor*{imbens2009identification} \citeyear{imbens2009identification}). This is an equivalent model to \eqref{mainmodel} as one can define the two dependent variables $U$ and $V$ on the same probability space and combine them to one variable $W$.} 

In the causal inference literature, it is common to represent model \eqref{mainmodel} via a directed acyclic graph (\citeauthor*{balke1994counterfactual} \citeyear{balke1994counterfactual}, \citeauthor*{pearl1995causal} \citeyear{pearl1995causal}, \citeauthor*{balke1997bounds} \citeyear{balke1997bounds}) as in Figure \ref{dagmodel}.
\begin{figure}[h!t]
\centering
\begin{tikzpicture}
\draw node at (-3,0) {$\boxed{Z}$};
\draw node at (0,0) {$\boxed{X}$};
\draw node at (3,0) {$\boxed{Y}$};
\draw node[draw=black, circle, anchor=center] at (1.5,2) {$W$};
\draw[->,thick] (-2.7,0) -- (-0.3,0);
\draw[->,thick] (0.3,0) -- (2.7,0);
\draw[->,thick] (1.3,1.6) -- (0,0.3);
\draw[->,thick] (1.7,1.6) -- (3,0.3);
\end{tikzpicture}
\caption{DAG representation of model \eqref{mainmodel}}
\label{dagmodel}
\end{figure}
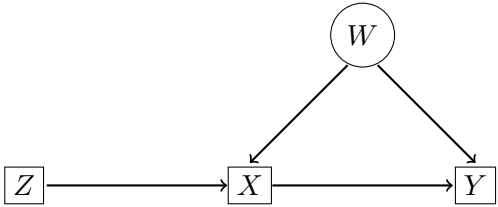

A third way to represent model \eqref{mainmodel} is via the counterfactual notation \citep*{rubin1974estimating}. In this case, the counterfactual (i.e.~unobservable) distribution $P_{Y(x)}$ denotes the distribution of the outcome given that $X$ is exogenous, i.e.~that $X\independent W$. Due to the endogeneity problem $X\not\independent W$ depicted in Figure \ref{dagmodel}, the observable distribution $P_{Y|X=x}$ does not coincide with $P_{Y(x)}$, so that one cannot simply use the observable distribution to obtain causal effects of interest. However, using the information of the instrument $Z$, a linear programming approach allows us to partially identify functionals of interest such as average treatment effects (ATE) or probabilities of counterfactual events without making structural assumptions on the relations between the variables in question. 

\subsection{Outline of the linear programming approach}
The linear programming approach proposed in this article is the natural generalization of \citet*{balke1994counterfactual} and \citet*{balke1997bounds}, which deal with obtaining sharp bounds on a causal effect (such as the average treatment effect) of $X$ on $Y$ when $Y,X,Z$ are binary, to the settings where all $Y$, $X$, and $Z$ are allowed to be continuous.\footnote{See \citet*{russell2019sharp} for an overview of linear programming approaches to identification in instrumental variable models.} It relies on yet another representation of model \eqref{mainmodel}, which we now introduce in an intuitive way and will address formally in Section \ref{sec:mainsection}. 

This representation is based on Rubin's counterfactual notation in the sense that it considers the counterfactual distributions $P_{Y(x)}$ and $P_{X(z)}$ to be the laws of corresponding stochastic processes $Y_x(w)$ and $X_z(w)$ of the first and second stage. Each element $w\in\mathcal{W}$ indexes one path $Y_x(w)$ and $X_z(w)$ of the processes, respectively. Formally, the process $Y_x$ corresponds to $h(x,W)$, and analogous for $X_z$ and $g(z,W)$.
This representation allows one to set up a linear program for obtaining bounds on functionals of interest in model \eqref{mainmodel}: 
\begin{equation}\label{minimizereq}
\begin{aligned}
&\underset{\substack{P_W\in\mathscr{P}^*(\mathcal{W})}}{\min/\max}\quad  E_{P_W}[f(Y_{x},x_0)]\\
\text{s.t.}\thickspace\thickspace &F_{Y,X|Z=z}(y,x) = P_W(Y_{X_z}\leq y,X_z\leq x)
\end{aligned}
\end{equation}
The programs \eqref{minimizereq} are phrased for the general case of potentially continuous $Y$, $X$, and $Z$, anticipating the formal results in Section \ref{sec:mainsection}, but are otherwise perfectly analogous to the optimization program in \citet*{balke1994counterfactual}. They proceed by finding the optimal distribution $P_W$ which maximizes (for an upper bound) or minimizes (for a lower bound) the functional of interest $E_{P_W}[f(Y_{x},x_0)]$ under the restriction that $P_W$ induces processes $Y_x$ and $X_z$ whose induced joint distribution (defined below)
\[F_{[Y,X]^*_z}\coloneqq P_W(Y_{X_z}\leq y,X_z\leq x)\] coincides with the joint observable distribution $F_{Y,X|Z=z}$. The optimization is over some set of probability distributions $\mathscr{P}^*(\mathcal{W})$ over the support $\mathcal{W}$ of $W$. The main contribution of this article is a method to probabilistically approximate the solution to optimization problems of the form \eqref{minimizereq}. We relegate all formal questions about $\mathscr{P}^*(\mathcal{W})$ and the optimization problem more generally to Section \ref{sec:mainsection}. 

The functionals in the objective function of \eqref{minimizereq} in general can be any map $\phi:\mathscr{P}^*(\mathcal{W})\to\mathbb{R}$ from probability distributions on $\mathcal{W}$ to the real line. For the sake of conciseness, and because they are ubiquitous in applied research, we focus on linear functionals such as the ATE, probabilities of counterfactual events, and the average derivative, which make \eqref{minimizereq} linear programs. In this case the objective functions all take the form $E_{P_W}[f(Y_x,x)]$, where $E_{P_W}[\cdot]$ denotes the expectation with respect to $W$ and where $f(Y_x, x)$ is a function that takes in the position $y$ of the stochastic process $Y_x$ at some value $x$ and provides a real number. Some examples are the following.
\begin{enumerate}
\item For the ATE of an exogenous change of $X$ from $x_0$ to $x_1$, $E[Y|X=x_1] - E[Y|X=x_0]$, the objective function takes the form $E_{P_W}[f(Y_x,x_0,x_1)] = E_{P_W}[Y_{x_1} - Y_{x_0}]$. 
\item For the probability of some counterfactual event $A_y\subset\mathcal{Y}$ happening given that $X$ exogenously takes the value $x_0$, $P_{Y(x_0)}(A_y)$, the objective function takes the form $E_{P_W}[f(Y_x,x_0, A_y)]=E_{P_W}[\mathds{1}_{A_y}(Y_{x_0})]$, where $\mathds{1}_A(x)$ denotes the indicator function which is $1$ if $x\in A$ and $0$ otherwise.
\item For the average derivative of the form $E\left[\frac{\partial}{\partial x} f(X,W)\right]$ \citep*{imbens2009identification}, the objective function takes the form $E_{P_W}[f(Y_x,x)] = E_{P_W}\left[\frac{\partial}{\partial x} Y_x\right]$, if the process $Y_x(w)$ is differentiable. 
\end{enumerate}
Note that in all cases, we can straightforwardly replace the conditioning on values $x_0$ or $x_1$ by more general sets $\mathcal{A}_{x_0}$ and $\mathcal{A}_{x_1}$ without changing anything in the results or proofs. For the sake of notation, we always just consider a point $x_0$ instead of a set. 

Let us turn to the constraint of \eqref{minimizereq}. Since all paths $Y_x(w)$ and $X_z(w)$ are indexed by one realization $w$ of the random variable $W$, $P_W$ must weight these paths in such a way that the joint law $F_{[Y,X]^*(z)}(y,x)$ of the counterfactual stochastic processes $Y_x(w)$ and $X_z(w)$ induced by $P_W(Y_x\leq y, X_z\leq x)$ coincides with the observable distribution $F_{Y,X|Z=z}$.\footnote{Formally, $F_{[Y,X]^*(z)}(y,x)$ is the pushforward measure of $P_W$ via the process $[Y,X]^*_z$, see e.g.~\citet*[Definition 3.2]{bauer1996probability} for the definition of a pushforward- or image measure.} The construction of the joint counterfactual process $[Y,X]^*_z$ by the marginal processes $Y_x$ and $X_z$ follows from the structure of model \eqref{mainmodel}. The idea is that for given paths $Y_x(w)$ and $X_z(w)$ indexed by some $w$, 
\[[Y,X]^*_z(w) \coloneqq (Y_{X_z(w)}(w), X_z(w)).\] This means that for each value $z$ in the support $\mathcal{Z}$ of $Z$, the path $Y_{X_z(w)}(w)$ is constructed by composing the path $Y_x(w)$ with $X_z(w)$. Specifically, $Y_x(w)$ depends only on the current position in $\mathcal{X}$ of $X_z(w)$ for the given $z$, but no other properties of the respective path $X_z(w)$. This follows directly from the form of model \eqref{mainmodel}, as $h$ does not depend on $z$, so that intuitively
\[Y_{X_z(w)}(w) = h(g(z,w),w).\] The property that $h$ does not depend on $Z$ is called the \emph{exclusion restriction}. A depiction of the construction of $[Y,X]^*_z$ for two continuous processes $Y_x$ and $X_z$ is provided in Figure \ref{pathfigure}.
\begin{figure}[htb!]
\centering
\includegraphics[width=16cm,height=10cm]{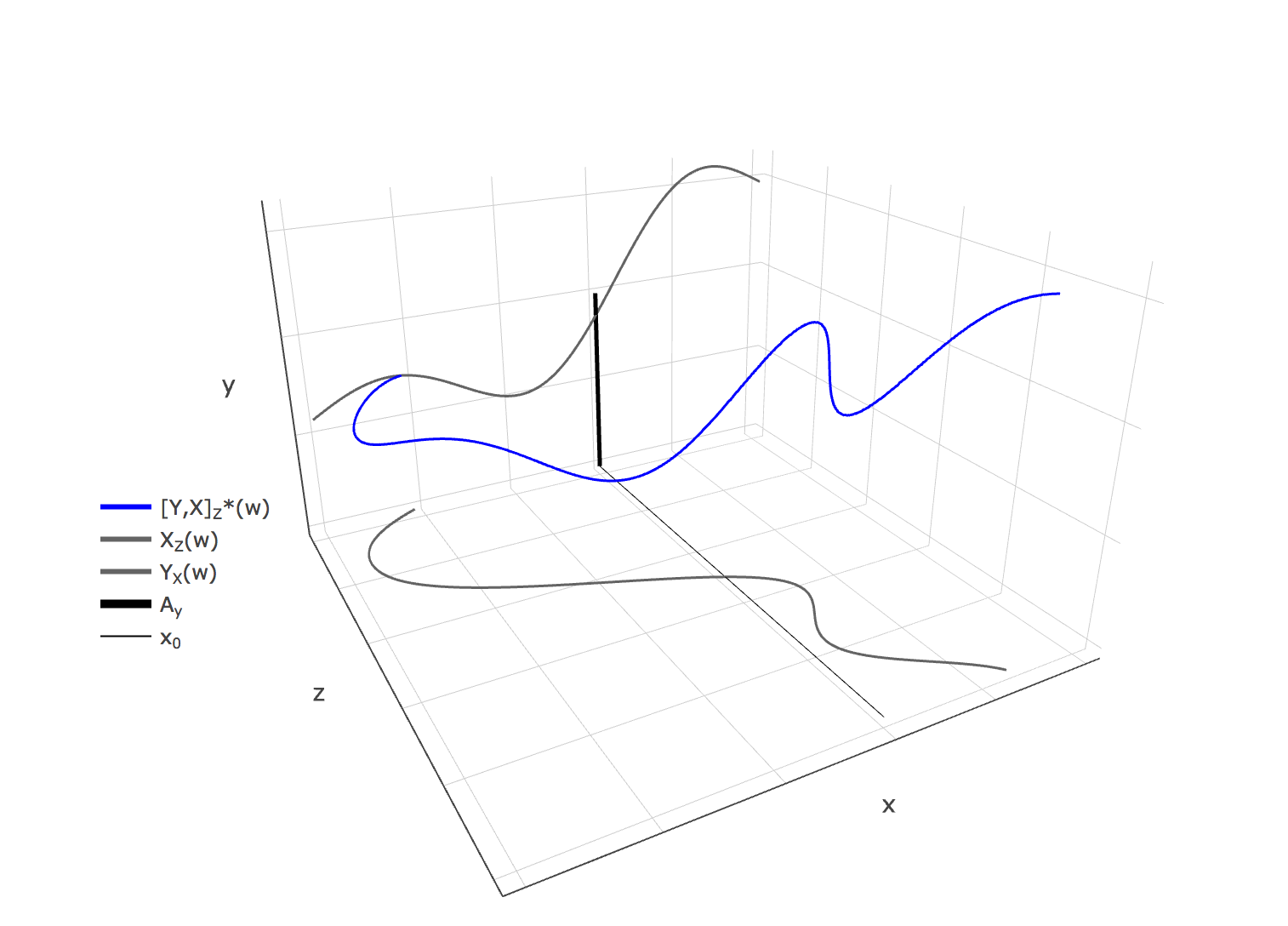}
\caption{Construction of one path $[Y,X]^*_z(w)$ of the joint counterfactual process via marginal paths $Y_x(w)$ and $X_z(w)$.}
\label{pathfigure}
\end{figure}

Figure \ref{pathfigure} also shows how an objective function of the form $E_{P_W}[\mathds{1}_{A_y}(Y_{x_0})]$ is included in this approach. Here, the marginal path $Y_x(w)$ ``goes through'' the event $A_y$ at the point $x_0$. So for the upper bound (i.e.~the maximization of \eqref{minimizereq}), the optimal $P_W$ should put a lot of weight on this path. For the lower bound (i.e.~the minimization), it should not put a lot of weight on the path. 

This representation via $Y_x$ and $X_z$ also allows to conveniently introduce any form of structural assumption into model \eqref{mainmodel} by ruling out certain paths. For instance, if the researcher is willing to assume that $h(\cdot,W)$ is continuous in $X$, then she will only require continuous paths $Y_x(w)$. This means that $P_W$ will only put mass on continuous paths and not ones with jumps. Similarly, one can introduce any other form of functional form restriction into the model. Note in general, that all functional form restrictions need to be made with respect to the observable random variables, never with respect to $W$. The unobservable $W$ is simply an index for the respective paths; this is a main difference to other canonical models in econometrics (e.g.~\citet*{chernozhukov2005iv}, \citet*{imbens2009identification}), which assume continuity and monotonicity of $h(x,\cdot)$ in $W$. 

In this respect, it is also interesting to consider point-identification results from the literature on nonseparable triangular models, as \citet*{imbens2009identification}, \citet*{d2015identification}, or \citet*{torgovitsky2015identification}, which require monotonicity of $h(x,\cdot)$ and $g(z,\cdot)$ in $W$. The idea is that monotonicity makes $h$ and $g$ injective in $v$ and $u$, respectively. In our setting, injectivity of $h$ in $v$ means that the paths of $Y_x(v)$ never intersect, so that for each $(y,x,z)$ there is a unique $w$---this would guarantee that we could point-identify (under some more regularity assumptions) the latent distributions $P_{Y(x)}$ and $P_{X(z)}$. Without monotonicity in the unobservable the joint paths $[Y,X]^*_z$ intersect almost everywhere, so that point-identification is impossible and one has to resort to the partial identification approach using programs \eqref{minimizereq}. This is another intuitive explanation for why monotonicity is such a staple in the literature on point-identification of nonparametric models. 

\paragraph{Intuition in terms of randomized controlled trials with imperfect compliance}\mbox{}\\
The setting of estimating dose-response functions of a treatment through a randomized controlled trial with imperfect compliance is a convenient setting to provide some more intuition of the approach. Let the treatment $X$ be continuous and also assume that the outcome of interest $Y$ is continuous. Then the initially randomly assigned dose $Z$ to each participant is the instrument to $X$. 
With a continuous treatment $X$, the amount of the treatment actually taken by the participants can deviate from the assigned treatment in an infinite number of ways. The experimenter can only work with the overall distribution $P_{Y,X|Z=z}$ since she knows that the observed distributions $P_{Y|X}$ and $P_{X|Z}$ are for an endogenous $X$ due to the imperfect compliance. 

In this setting, the stochastic processes $Y_x$ and $X_z$ have a natural interpretation as the response profiles of hypothetical participants, each being indexed by one element $w\in\mathcal{W}$. A response profile $X_z(w)$ tells the respective hypothetical participant indexed by $w$ which treatment dose $x$ she actually takes for each initially assigned treatment dose $z$. For instance, in analogy to the binary case \citep*{angrist1996identification}, the never taker is the path of the stochastic process which is always zero for each value of $Z$, i.e.~this participant's response profile is to take no treatment $x$ for every assigned dose $z$. The path of the complier would be the $45^\circ$ line, since this participant perfectly takes $x=z$ for all $z$.

Analogously a response profile $Y_x(v)$ ``tells'' a hypothetical participant how much she reacts ($y$) to a certain level of the treatment $x$. In this case the process $Y_x$ models the unobserved heterogeneity in treatment response of each participant. The hypothetical participant who does not respond to the treatment at all has the path $Y_x\equiv 0$ for any value of $x$ for instance. 
Very intuitively, the idea then is to find the relative optimal composition of hypothetical participants (i.e.~the weights $P_W$ on the respective paths of the stochastic processes) that maximizes (for an upper bound) or minimizes (for a lower bound) the relative number of hypothetical participants which on average have the strongest or weakest response to a change of the dose from $x_0$ to $x_1$, subject to the constraint that the composition $P_W$ of the hypothetical participants replicates the joint observable weight $P_{Y,X|Z=z}$. 
This framework is the natural generalization of the distinction into always taker, never taker, complier, and defier from \citet*{angrist1996identification} to the continuous setting, where there is a continuum of response profiles, not just four as in the binary case.

A way to introduce assumptions in this setting is via $P_W$. In particular, one could assume that all hypothetical participants have continuous responses, i.e.~if participant's response profile indexed by $w$ is to take $x_0$ units of the treatment when assigned $z_0$ units, then they will also take $x_0+\delta$ units when assigned $z_0+\delta$ units for some small $\delta>0$. An assumption like this would rule out drop-outs from the trial for instance. Response profiles that allow for drop outs are important in trials where the treatment can have severe side-effects under too high doses, like chemotherapy treatments for instance. One can model this by allowing for processes that jump to zero after a given $z_0$ and stay there for all $z>z_0$. This example should make clear that the stochastic process representation we introduce in this article is convenient for modeling latent behavior in the unobserved heterogeneity.

Let us now turn to examples of applications of model \eqref{mainmodel} and the corresponding programs \eqref{minimizereq} which will also highlight the curse of cardinality when trying to solve \eqref{minimizereq} in practice. We start with the binary analogue to our treatment setting.

\subsection{Examples of applications and the curse of cardinality}

\paragraph{Example 1: Curse of cardinality in treatment estimation}\mbox{}\\
Consider the treatment setting as above, but with a binary treatment $X\in\{0,1\}$, binary outcome $Y\in\{0,1\}$, and binary intent-to-treat $Z\in\{0,1\}$. This is the classical setting, see for instance \citet*{angrist1996identification}, \citet*{balke1997bounds}, \citet*{angrist1996identification}, \citet*{manski2003partial}, \citet*{russell2019sharp}. 

Suppose we are interested in estimating the ATE of a change from $X=0$ to $X=1$. In this simple case, \eqref{minimizereq} reduces to a linear program with $9$ equality constraints and $16$ non-negativity constraints
\[P_{Y,X|Z=z}(y,x) = P_W(Y_{X_z}=y, X_z=x)\] over distributions $P_W$ on the support $\mathcal{W}$ of $W$, which in this case consists of $(2^2)^2=16$ elements (\citet*{balke1994counterfactual} and \citet*{balke1997bounds}). 

The reason for this is as follows. 
In both the first- and second stage there are four possible functions mapping $Z$ to $X$ and $X$ to $Y$, respectively. In the first stage, these paths of processes have been given specific names \citep*{angrist1996identification}:
the never takers, i.e.~the path $X_z=0$ for all $z\in\{0,1\}$; the compliers, i.e.~the path $X_0=0$ and $X_1 = 1$; the defiers, i.e.~the path $X_0=1$ and $X_1=0$; the always takers, i.e.~the path $X_z =1$ for all $z\in\{0,1\}$.
The analogous set-up holds for the second stage with the same four types. $W$ therefore indexes every combination between one of the types in the first-and second stage, which is $16$. The number of equality constraints follows from the fact that the observable distribution $P_{Y,X|Z=z}$ can take $2^3=8$ values: each combination of $Y, X$, and $Z$. One more equality constraint is needed to guarantee that all probabilities sum to $1$, and $16$ inequality constraints are required to require that all proabilities are non-negative. 

Note that one can arrive at the same program via many different routes: random set theory using Artstein's inequality (\citet*{molchanov2005theory}, \citet*{beresteanu2012partial}, \citet*{chesher2017generalized}, \citet*{russell2019sharp}) or optimal transport theory \citep*{galichon2011set}. \citet*{russell2019sharp} gives a clear account of this setting.

As stated, this problem is a simple linear program and is trivial to compute. However, in settings where $Y$, $X$, and $Z$ are allowed to have more than two points in their support, the complexity of the linear program grows at least exponentially. Consider the case where $Y$ can take $q$ different values, $X$ can take $r$ different values, and $Z$ can take $s$ different values. In this case, the linear program will have $q\cdot r\cdot s +1$ equality constraints. Even more drastically, the number of elements in the support of $\mathcal{W}$ is $r^s\cdot q^r$ with the same number of inequality constraints.

This implies that already in a setting where $Y$, $X$, and $Z$ are allowed to take three variables each, $\mathcal{W}$ would consist of $27^2=729$ points, while the number of equality constraints would be $28$, with another $729$ nonnegativity constraints. Already setting up the linear programming problem, either through Arstein's inequality, optimal transport, or the stochastic process method depicted here, will be complicated to do. In the existing literature, \citet*{cheng2006bounds} circumvent this problem in the setting of three-armed randomized controlled trials with imperfect compliance by making monotonicity assumptions which rule out most of the paths, reducing the linear program to a simpler form. 

However, already in the case of where $Y$, $X$, and $Z$ are allowed to take $5$ elements each is it practically impossible to set up the linear program over $(5^5)^2=9,765,625$ elements in $\mathcal{W}$ with the same number of non-negativity constraints and $5^3+1=126$ equality constraints. This means that any existing approach requiring to set up an optimization problem of this form, such as \citet*{balke1994counterfactual}, \citet*{chesher2017generalized}, \citet*{beresteanu2012partial}, \citet*{galichon2011set}, \citet{russell2019sharp} and others cannot be applied in these practical settings, even though they show how to set-up these linear programs in theory. This constitutes what we call the ``curse of cardinality'', and is the main reason for why current practical application of this idea as in \citet*{demuynck2015bounding}, \citet*{laffers2015bounding}, \citet*{mogstad2018using}, or \citet*{chesher2017generalized} focus on the case where $X$ is binary. 

The example with randomized controlled trials and a continuous treatment we introduced above is the extreme case of this, where $Y$, $X$, and $Z$ are continuous. In this case, there currently does not exist an approach to solve the programs \eqref{minimizereq}. In the next section, we introduce a sampling approach that approximates the solutions to the problems \eqref{minimizereq} while giving probabilistic approximation guarantees, even if all $Y$, $X$, and $Z$ are continuous. Before we do this, let us give other examples for potential applications.

\paragraph{Example 2: Time-varying treatments in program evaluation}\mbox{}\\
In general, $Z$ in model \eqref{mainmodel} is a classical instrument. However, the model and the approach presented in this paper are even more broadly applicable. For instance, in models for the evaluation of a training program \citep*{ashenfelter1978estimating} one can consider $Z$ to be time. In this case, solving the programs \eqref{minimizereq} still provides bounds on the causal effect of interest, because the model assumes that the counterfactual distribution $P_{Y(x)}$ does not change over time, i.e.~that the actual causal effect of the training $x$ on the outcome of interest is fixed. This follows from the fact that $Y_x\coloneqq h(x,W)$ does not explicitly depend on $t$.

To make this more tangible, consider the hypothetical example of a rehabilitation program (such as a self-help group, a health program, etc.) that runs over several weeks. The treatment $X$ is the time (e.g.~days) spent in the program, which makes it a continuous treatment. $X$ being continuous introduces the exact curse of cardinality as described in Example 1. $Y$ could be a binary variable $Y\in \{0,1\}$, indicating if the individual has not relapsed $6$ months after the full program ended. Since participation in the program is voluntary, participants can drop out any time; this introduces an endogeneity bias as participants will choose to opt in or out depending on their expected outcome. One instance of this endogeneity problem is the classical Ashenfelter dip \citep*{ashenfelter1984using}, which shows that the outcome of interest changes to the negative right before they opt to partake in a program. 

In this setting, model \eqref{mainmodel} implies that the second stage process $Y_x\coloneqq h(x,W)$ is independent of time, while the participation in the treatment $X$ can vary with time. Time $Z=t$ is never influenced by the model; furthermore, the unobservable heterogeneity $W$ is fixed in time (it indexes the response profiles of the hypothetical participants), so that the independence restriction $Z\independent W$ is satisfied. Using the linear programming approach \eqref{minimizereq}, one can make assumptions on the participants in the first stage $X_t=g(t,W)$, by modeling their response profiles including drop-in and drop-out rates. Since $X$ is the time spent in the program, all paths $X_t(w)$ are by definition restricted to the $45$-degree line on a plot between $X$ and $Z$, the only difference being the drop-in, or drop-out rate. To adequately model these drop-ins and drop-outs, the stochastic processes $X_t$ in this case should allow for jumps. Suppose for the sake of exposition that the program allows participants to drop in or drop out on a continuous basis. In this case, we still have a continuum of paths $X_t(w)$, i.e.~a continuum of elements $w$. Two examples of paths are depicted in Figure \ref{dropfig}. This continuum of paths is the extreme case of the curse of cardinality, just as in Example 1; in particular, no current approach can deal with a model like this without making stronger assumptions.
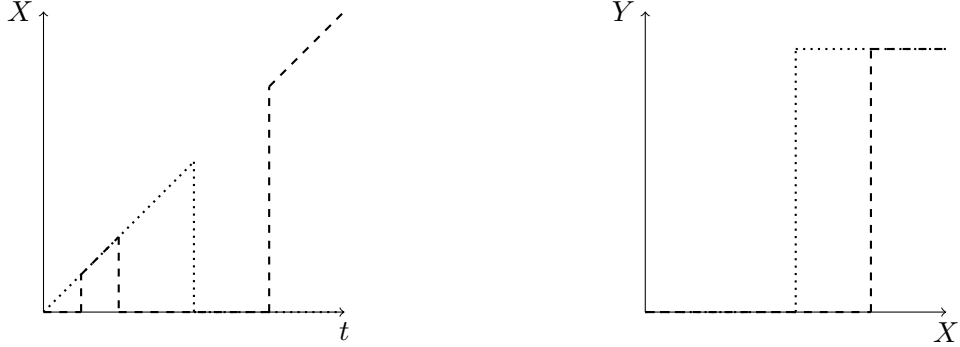
\begin{figure}[h!t]
\centering
\begin{tikzpicture}
\draw[->] (0,0) to (0,4);
\draw[->] (0,0) to (4,0);
\draw node[left] at (0,4) {$X$};
\draw node[below] at (4,0) {$t$};
\draw[-,dotted, thick] (0,0) to (2,2);
\draw[-,dotted, thick] (2,2) to (2,0);
\draw[-,dotted, thick] (2,0) to (4,0);
\draw[-,dashed,thick] (0,0) to (0.5,0);
\draw[-,dashed,thick] (0.5,0) to (0.5,0.5);
\draw[-,dashed, thick] (0.5,0.5) to (1,1);
\draw[-,dashed, thick] (1,1) to (1,0);
\draw[-,dashed, thick] (1,0) to (3,0);
\draw[-,dashed, thick] (3,0) to (3,3);
\draw[-,dashed, thick] (3,3) to (4,4);

\draw[->] (8,0) to (8,4);
\draw[->] (8,0) to (12,0);
\draw node[left] at (8,4) {$Y$};
\draw node[below] at (12,0) {$X$};
\draw[-,dotted, thick] (8,0) to (10,0);
\draw[-,dotted, thick] (10,0) to (10,3.5);
\draw[-,dotted, thick] (10,3.5) to (12,3.5);
\draw[-,dashed, thick] (8,0) to (11,0);
\draw[-,dashed, thick] (11,0) to (11,3.5);
\draw[-,dashed, thick] (11,3.5) to (12,3.5);
\end{tikzpicture}
\caption{Two examples of possible paths for the first stage $X_t(w)$ and the second stage $Y_x(w)$, respectively, in Example 2. First stage (left): The dotted line depicts a hypothetical participant starting the program from the beginning and dropping out half-way. The dashed line depicts a hypothetical participant dropping into the program, dropping out after a short while, then dropping in again at a later stage and staying with the program until the end. Second stage (right): The dotted and dashed lines depict tow response profiles of hypothetical participants that need different time spells in the program to receive the desired outcome. }
\label{dropfig}
\end{figure}

This example also shows how one can introduce assumptions into the model via restricting the paths. For instance, it could be the case that the program does not allow people to drop-in, and requires everyone to start at the same time. In this case paths of hypothetical participants that model drop-ins like the dashed path in Figure \ref{dropfig} should not be part of the model. Another assumption could be that for certain time-periods no drop-ins or drop-outs are allowed, in which case paths with drop-ins or drop-outs during this period would not be part of the model. When solving the programs \eqref{minimizereq}, such an assumption could be included by requiring the optimal $P_W$ to always put weight of zero on these kinds of paths. It is here where the stochastic process framework is so convenient, as the researcher can model individual behavior by restricting the unobserved heterogeneity in very specific ways. Introducing these assumptions is significantly harder to do in the standard representations of the instrumental variable mentioned above.

In order to model the second stage $Y_x$, one could assume stochastic processes $Y_x$ which are zero everywhere but jump to $1$ after a certain length $[x_0,x_1]$, where $x_0$ is the left-most value of the overall interval. This assumption encapsulates the unobserved heterogeneity as different responses to treatment length: participants might need different time lengths to obtain the required outcome. Figure \ref{dropfig} depicts two potential paths. Together those marginal paths introduce joint paths as depicted in Figure \ref{pathfigure}. 

For this model, an objective of interest might be the probability that the participant does not relapse $6$ months after treatment if the participant has continuously stayed in the program for a period $[x_0,x_1]$. The objective function in \eqref{minimizereq} then takes the form 
\[E_{P_W}[f(Y_x,x)]\coloneqq E_{P_W}\left[Y_{x\in [x_0,x_1]}\right],\] where $Y_{x\in [x_0,x_1]}$ means the part of the process $Y_x$ when $x\in[x_0,x_1]$. As we show below, the proposed method below can handle discontinuous objective functions if the discontinuity is of the form $f(P_W;x_0)=E_{P_W}[\mathds{1}_{A_y}(Y_{x_0})]$, where $A_y$ is a Borel subset of the support $\mathcal{Y}$ of $Y$. 
The method introduced in this article can deal with complex models like this, where the treatment is considered to be more intense over time. This is a case of the classical ``dose-response'' estimation. Here, the dose is the number of days in the program.

\paragraph{Example 3: Partial identification of average welfare}\mbox{}\\
Unobserved individual heterogeneity is a major factor contributing to the variability in empirical demand estimation, which is why models with general unobserved heterogeneity are particularly important in this setting, as argued in \citep*{hausman2016individual}. In that article, the authors consider the following model for the demand function
\[q(x,\eta) = \argmax_{q\geq0,a\geq0} U(q,a,\eta)\qquad\text{s.t}\qquad p'q+a\leq y.\] The demand function $q(x,\eta)$ is nonseparable in the unrestricted unobserved heterogeneity $\eta$ and $x\coloneqq(p',y)'$ consisting of the price vector $p$ relative to a numeraire good $a$, and the income level $y$ of the individual relative to the 
same numeraire. $U(q,a,\eta)$ is a (concave in $\eta$) utility function which the authors assume to be increasing in $q$ and $a$. $p$ and $y$, and therefore $x$, are usually continuous variables, which again introduces the extreme curse of cardinality from Example 1 into the problem.

The demand function corresponds to the second stage of model \eqref{mainmodel}. Concavity of $U$ guarantees that $q(x,\eta)$ is an actual function in $\eta$ and not a correspondence, so that each path of the process $q(x,\eta)$ is a function. This is the same setting we consider in this article.\footnote{Our approach can straightforwardly be adapted to allow for correspondence-valued processes without any changes in the overall approach.} 
Allowing for general $\eta$ precludes point-identification without stronger assumptions in general, which is why \citet*{hausman2016individual} provide an informal approach to obtain bounds on the average consumer surplus. 

The approach proposed in this article can be applied to this setting and hence provides the formal justification for the approach in \citet*{hausman2016individual}. Furthermore, we also provide the asymptotic properties of the upper- and lower bound, which can be used to perform inference in their setting. \citet*{hausman2016individual} also assume that $x\independent \eta$ in their application, a strong assumption in practice. Our approach is designed for handling more general models allowing for $X\not\independent\eta$ and instrumenting for $X$. 

The stochastic process approach formalized in this article has another important benefit in this setting. One can introduce any form of assumption into the model (the demand function $q(x,\eta)$) by ruling out paths $q(x,\eta)$ which do not conform with this assumption. This is the same idea as in the previous example of response profiles. In particular, each path is a response profile of an actor $\eta$ in the market. 

Note that in combination with the idea of \citet*{hausman2016individual}, our approach is a complementary approach to the finitary method proposed in \citet*{kitamura2018nonparametric} for partial identification of average welfare. \citet*{kitamura2018nonparametric} work directly with individual budget sets and define a finite linear programming problem which grows with the number of data points in practice. Our method is an attractive alternative to this in two main ways. First, we can allow for continuous functions and variables. In this setting the programs \eqref{minimizereq} become infinite dimensional optimization programs. The practical advantage of this is that we can put a grid on the data and need not evaluate the linear program at each data point, which can be impossible to do when the data-set consists of many oservations. Second, and more importantly, the continuous setting allows for the introduction of assumptions on the demand functions $q(x,\eta)$ like smoothness or differentiability. These assumptions are vacuous in a discrete setting such as \citet*{kitamura2018nonparametric}. Such regularity assumptions follow naturally from the economic models (e.g.~the properties of the Slutsky matrix) and should be included as assumptions in the empirical setting, which our continuous setting allows for.

\section{The formal method for approximately solving the programs \eqref{minimizereq}}\label{sec:mainsection}
\subsection{A rigorous definition of the stochastic process representation of model \eqref{mainmodel}}\label{stochasticpathsec}
Before providing the solution approach to solving the programs \eqref{minimizereq} in practice, we define the stochastic process representation of model \eqref{mainmodel} formally. From now on, we let $Y$, $X$, and $Z$ be random variables whose corresponding laws $P_Y$, $P_X$, and $P_Z$ have compact supports $\mathcal{Y}$, $\mathcal{X}$, $\mathcal{Z}\subset\mathbb{R}^d$ for some $d\geq1$. In the following, we set $d=1$ for notational convenience, but all results in higher-dimensional settings.\footnote{One can allow for Polish spaces, i.e.~complete, separable and metrizable topological spaces, under some formal complications. The compactness restriction can be dropped at the cost of more mathematical formalism.} Moreover, since we focus on compact supports on $\mathbb{R}$ it is also without loss of generality to assume $\mathcal{Y}=\mathcal{X}=\mathcal{Z}=[0,1]$, which we do for the sake of notation.

As stated previously, the main idea for the approach is to interpret the counterfactual laws $P_{Y(x)}$ and $P_{X(z)}$ of the instrumental variable model as laws of counterfactual stochastic processes $Y_x$ and $X_z$, where the randomness is induced by the unobservable $W$. Each element $w\in\mathcal{W}$ indexes one path $Y_x(w)$ of the process. This representation allows to optimize over the unobserved heterogeneity $W$ directly instead of the laws $P_{Y(x)}$ and $P_{X(z)}$ induced by it, which is the key for making the approach feasible. 
The first question to answer is in which space these paths $Y_x(w)$ and $X_z(w)$ live. Completely unrestricted functions $h$ and $g$ as allowed for in model \eqref{mainmodel} correspond to stochastic processes defined in $\mathbb{R}^{\mathbb{R}}$, i.e.~the space of all functions from $\mathbb{R}$ to $\mathbb{R}$. This space is large enough to allow for response profiles such as in Example 2, where candidates are in principle allowed to drop in and drop out uncountably many times.

From a mathematical perspective, $\mathbb{R}^{\mathbb{R}}$ is too big and the Borel $\sigma$-algebra is too coarse for many assumptions of interest, however. For example, many subspaces like $C(\mathbb{R})$, the space of continuous functions on $\mathbb{R}$, are not even measurable in the Borel $\sigma$-algebra on $\mathbb{R}^{\mathbb{R}}$ \citep*[Corollary 38.5]{bauer1996probability}. Also, the response profiles allowed by $\mathbb{R}^{\mathbb{R}}$ are too general for most practical purposes. For instance, allowing for uncountably many drop-ins and drop-outs as one response profile in Example 2 is of no practical interest, as one always observes finitely many drop-ins and drop-outs of a person. This level of generality is hence really not needed for modeling the hypothetical participants. On the other hand, restrictions like continuity which require all response profiles to be continuous, might be too restrictive in many settings. An example is again Example 2, where allowing for jumps is required by the problem.

It is here where the stochastic process representation is useful for the first time, as it allows to introduce very weak but easy to understand assumptions which restrict the unobservable heterogeneity $W$. A rather weak but restrictive enough assumption for all practical settings is to only consider stochastic processes that possess almost surely \emph{c\`adl\`ag} paths, i.e.~paths that are continuous on the right with limits on the left. These paths allow to model all possible practically relevant response profiles, as they allow for finitely many jumps larger than any value $\varepsilon>0$ in the paths. This means that one can realistically model the response profiles of Example 2; moreover, paths that require continuity of the responses are measurable in this space, too.

The standard space for these processes is the Skorokhod space $D(\mathcal{X})$, which we equip with the Skorokhod metric 
\[d_S(Y_x,Y'_x)\coloneqq \inf_{\lambda\in\Lambda}\left\{\max\{\|\lambda-I\|_\infty, \|Y_x - \lambda Y'_x\|_\infty\}\right\}\] to make it a separable and metrizable space. In particular, a metric similar to the Skorokhod metric induces the same topology on $D(\mathcal{X})$ and makes it a Polish space, so that from now on we consider $D(\mathcal{X})$ with the Skorokhod metric a Polish space.\footnote{The reason for assuming compactness of $\mathcal{Y}$, $\mathcal{X}$, and $\mathcal{Z}$ is to be able to work within the Skorokhod space over a compact set, which is nicer to handle. Furthermore, one can define the Skorokhod space also when $\mathcal{X}$ is higher-dimensional, see \citet*{neuhaus1971weak}, \citet*{straf1972weak}, and \citet*{bloznelis1994central}. The metric which is topologically equivalent to the Skorokhod metric is often called the Billingsley metric \citep*{billingsley1999convergence}.} Here, $\Lambda$ is the set of all strictly increasing and continuous functions $\lambda:\mathcal{X}\to\mathcal{X}$, $I:\mathcal{X}\to\mathcal{X}$ is the identity function on $\mathcal{X}$, and $\|\cdot\|_\infty$ is the supremum norm \citep*[Section 12]{billingsley1999convergence}. The Skorokhod space is perfect as the underlying space for the representation of model \eqref{mainmodel} via stochastic processes, as it allows for all practically relevant response profiles while being a Polish space, which provides good measurability properties.

One sufficient and necessary condition for a function $f:\mathcal{X}\to\mathbb{R}$ to lie in the Skorokhod space is that its extended modulus of continuity $\omega'_f(\delta)$ satisfies $\omega'_f(\delta)\to0$ for any $\delta\downarrow 0$ \citep*[p.~123]{billingsley1999convergence}.\footnote{$\omega'_{f}(\delta)\coloneqq \inf_{\{x_i\}}\max_{i\in\mathbb{N}}\sup_{x,x'\in[x_{i-1},x_i)}|f(x)-f(x')|$, where the infimum extends over all $\delta$-partitions $\{x_i\}_{i\in\mathbb{N}}$, i.e.~partitions where the values $x_i$ are at a distance of at least $\delta>0$ from each other, see \citet*[p.~122]{billingsley1999convergence}. This definition extends the modulus of continuity for continuous functions.}
We therefore require the following Assumption throughout.
\begin{assumption}[Skorokhod space]\label{skorokhodpath}
$Y$, $X$, and $Z$ take values in $\mathbb{R}$ with supports $\mathcal{Y}, \mathcal{X}, \mathcal{Z}=[0,1]$. Moreover, for every $\delta>0$ with $\delta\to0$ it holds that
\[\lim_{\delta\to0}\sup_{w\in\mathcal{W}}\omega_{Y_x(w)}'(\delta) = 0\qquad\text{and}\qquad \lim_{\delta\to0}\sup_{w\in\mathcal{W}}\omega_{X_z(w)}'(\delta) = 0.\footnote{$\mathcal{W}$ denotes the support of $W$, which we define more formally below. Also note the difference between the extended modulus of continuity $\omega$ and a realization $w$ of $W$.}\]
\end{assumption}
Assumption \ref{skorokhodpath} only allows paths $Y_x(w)$ and $X_z(w)$ that have finitely many discontinuities which exceed any given real number. This assumption is weak enough in all practical settings, as argued above.
Under Assumption \ref{skorokhodpath}, one can formally introduce the stochastic process representation of model \eqref{mainmodel}. Throughout the article, these stochastic processes are defined for ``deterministic time'', even though $X$ and $Z$ are random variables, as the supports $\mathcal{X}$ and $\mathcal{Z}$ are fixed. This is without loss of generality in this model; in fact, all of the results in this article go through almost verbatim under the assumption that $X$ and $Z$ have random supports $\mathcal{X}$ and $\mathcal{Z}$. The only difference will be the formal complication of working with stochastic processes in random time.\footnote{In this case, one would need to replace Kolmogorov's extension theorem by the random version provided in Theorem 1 of \citet*{hu1988generalization}. } It is, moreover, not even clear that the random support assumption generates more general models than the deterministic support assumption.

The following proposition provides the formal result for the proposed stochastic process representation of model \eqref{mainmodel}.
\begin{proposition}[Stochastic process representation]\label{myprop}
Model \eqref{mainmodel} is equivalent to a system of counterfactual stochastic processes $Y_x$ and $X_z$ on $\mathbb{R}^{\mathbb{R}}$ equipped with the Borel $\sigma$-algebra on the cylinder sets and with corresponding laws $P_{Y(x)}$ and $P_{X(z)}$. The randomness of these processes is induced by $W$ defined on the probability space $([0,1],\mathscr{B}_{[0,1]},P_W)$.

Under Assumption \ref{skorokhodpath}, the processes are measurable in $D(\mathcal{X})$ and $D(\mathcal{Z})$, respectively. The exclusion restriction implies that (i) $X_z(w)$ is a stopping time for the process $Y_{X_z(w)}(w)$ and (ii) the laws $P_{Y|X=x}$ and $P_{X|Z=z}$ induce a joint law 
\[P_{[Y,X]^*(z)}(A_y,A_x) = \int_{A_x}P_{Y(x)}(A_y)dP_{X(z)}(x),\qquad A_y\in\mathscr{B}_{\mathcal{Y}},A_x\in\mathscr{B}_{\mathcal{X}},\]
which corresponds to a joint counterfactual process $[Y,X]^*_z$ in $(\mathbb{R}^2)^{\mathbb{R}}$. Furthermore, $[Y,X]_z^*$ can be defined on $D(\mathcal{Z})$ with codomain $\mathbb{R}^2$ and is measurable with respect to the filtration $\mathscr{F}_z^{[Y,X]}$ of all events of $[Y,X]_z^*$ that have happened before $z\in\mathcal{Z}$. The independence restriction $Z\independent W$ allows to compare the counterfactual process $[Y,X]^*_z$ to the stochastic process $[Y,X]_z$ corresponding to the observable joint law $P_{Y,X|Z=z}$.
\end{proposition}

Proposition \ref{myprop} states that the counterfactual probability measures $P_{Y(x)}$ and $P_{X(z)}$ are the laws of $Y_x$ and $X_z$, where the randomness stems from the unobservable confounder $W$. Since $W$ is unobserved, one can define it on $([0,1],\mathscr{B}_{[0,1]},P_W)$, i.e.~to identify the support $\mathcal{W}$ of $W$ with the unit interval. This follows from standard isomorphism results, see for instance \citet*[Theorem 9.2.2]{bogachev2007measure2}. It is in fact possible to construct these explicit isomorphisms between general Borel measures on Polish spaces and Lebesgue measure on the unit interval based on the approach laid out in \citet*{kuratowski1934generalisation}. One standard application of these isomorphisms is the construction of Wiener measure, see for instance \citet*{hess1982kuratowski}. 

This is the main difference to the other existing identification results in the literature: we do \emph{not} make topological or algebraic assumptions on the support $\mathcal{W}$. $W$ is always infinite dimensional in our setting and completely abstract: only the \emph{cardinality} of $\mathcal{W}$ is important, which needs to be such that each path $Y_x(w)$ and $X_z(w)$ in the Skorokhod space---a Polish space under the Skorokhod metric---can be indexed by one $w$. By the above isomorphism results, this is already possible if $\mathcal{W}\subset[0,1]$. This is the main reason why we can define $\mathcal{W}$ to be the unit interval while the existing identification approaches in econometrics make structural assumptions on $\mathcal{W}$.

The construction of $W$ on the unit interval allows us to define $\mathscr{P}^*(\mathcal{W})$ as the set of all measures $P_W$ on $([0,1],\mathscr{B}_{[0,1]})$ which induce paths on the Skorokhod space defined in Assumption \ref{skorokhodpath}. The idea is that by a construction like the one in \citet*{kuratowski1934generalisation}, a measure $P_W$ will induce a measure on path space, so that we can optimize over measures on the unit interval in the programs \eqref{minimizereq}. 
In this respect, all additional functional form assumptions on the model will be made by shrinking $\mathscr{P}^*(\mathcal{W})$. For instance, making continuity assumptions between $Y$ and $X$ or $X$ and $Z$ in \eqref{mainmodel} translates to continuity assumptions on the paths $Y_x$ or $X_z$. In this case, $\mathscr{P}^*(\mathcal{W})$ contains only measures $P_W$ on $\mathcal{W}$ which put positive probability on continuous paths. Since under Proposition \ref{myprop} all processes are progressively measurable with respect to their natural filtration, one can even introduce dynamic assumptions this way, like mixing properties. This can be interesting in settings similar to Example 2, where the treatment is time dependent.

\subsection{The formal method for approximating the solutions to the programs \eqref{minimizereq}}\label{mainsubsec}
This section introduces the probabilistic approach to solve the programs \eqref{minimizereq} approximately in practice with probabilistic finite sample guarantees for the validity of the approximation. The method proceeds by sampling bases of the path space and solving semi-infinite dimensional analogues \citep*{anderson1987linear} of \eqref{minimizereq} over this sample of basis functions. The randomness introduced by sampling the paths permits the use of concentration results \citep*[Chapter 2.14]{wellner2013weak} which provide probabilistic guarantees for the approximation of \eqref{minimizereq} by the random sample.

The idea is to approximate the infinite dimensional problems \eqref{minimizereq} by their semi-infinite analogues
\begin{equation}\label{minimizereqapprox}
\begin{aligned}
& \underset{\substack{\hat{P}_W\in\hat{\mathscr{P}}^*(\mathcal{W})}}{\min/\max} \quad \frac{1}{l} \sum_{i=1}^lf(\tilde{Y}^\kappa_{x}(i),x)\frac{d\hat{P}_W}{d\hat{P}_0}(i)\\
\text{s.t.}\thickspace\thickspace &\left\|F_{Y,X|Z=z} - \frac{1}{l}\sum_{i=1}^l \mathds{1}_{[0,\cdot]\times[0,\cdot]}\left(\tilde{Y}^\kappa_{\tilde{X}^\kappa_{z}(i)}(i),\tilde{X}^\kappa_{z}(i)\right)\frac{d\hat{P}_W}{d\hat{P}_0}(i)\right\|^2_{L^2([0,1]^2)}\leq \varepsilon
\end{aligned}
\end{equation}
for some small $\varepsilon>0$ and all $z\in[0,1]$, and where the $L^2$-norm is taken with respect to $y,x$.\footnote{Rewriting the programs \eqref{minimizereqapprox} in their penalized form later on will reveal that $\varepsilon$ fulfills the same purpose as a penalty term for the constraint. In this sense, $\varepsilon$ is a penalty parameter of the program which needs to be chosen appropriately. See the next section for a discussion. Overall, the relaxation of the linear constraints to an $L^2$-constraint is made exclusively based on finite sample considerations. In the population $\varepsilon=0$.} The indicator functions $\mathds{1}_{[0,\cdot]\times[0,\cdot]}\left(\tilde{Y}^\kappa_{\tilde{X}^\kappa_{z}(i)}(i),\tilde{X}^\kappa_{z}(i)\right)$ capture the requirement that $\left(\tilde{Y}^\kappa_{\tilde{X}^\kappa_z(i)}(i),\tilde{X}_z^\kappa(i)\right)\leq(y,x)$ at $z$, where the inequality is considered element-wise. An equivalent way to writing this is $\mathds{1}_{[\tilde{Y}_{\tilde{X}_z^\kappa(i)}^\kappa(i),1]\times [\tilde{X}^\kappa_z(i),1]}(y,x)$. 

Underlying the approximations \eqref{minimizereqapprox} is a random sample of size $l$ of a number of $\kappa$ basis functions which approximate the path space of the processes $Y_x$ and $X_z$. These approximations are denoted by 
\begin{equation}\label{sieveseq}
\tilde{Y}^\kappa_{x}(i)\coloneqq \sum_{j=1}^\kappa \beta_j(i)b_j(x)\qquad\text{and}\qquad \tilde{X}^\kappa_{z}\coloneqq \sum_{j=1}^\kappa \alpha_j(i)a_j(z)
\end{equation} for coefficients $\beta,\alpha$ and basis functions $b(x),a(z)$. Of particular convenience are basis functions used in Sieves estimation such as (trigonometric) polynomials, splines, wavelets, etc.~\citep*{chen2007large}. 

The dependence of $\tilde{Y}_{x}(i)$ and $\tilde{X}_{z}(i)$ on the index $i$ shows that the problems \eqref{minimizereqapprox} are indeed semi-infinite dimensional in the sense of \citet*{anderson1987linear}. In particular, the index $i$ now runs over finitely many elements $l$ and replaces the variable $w\in\mathcal{W}$ on the state space of $P_W$. 

Let us now go over the most important terms in \eqref{minimizereqapprox}, $\frac{d\hat{P}_W}{d\hat{P}_0}(i)$, $\tilde{Y}^\kappa_{x}(i)$ and the objective function $f(Y_x,x_0)$, one by one.

\paragraph{1. The representative law}\mbox{}\\
The term $\frac{d\hat{P}_W}{d\hat{P}_0}(i)$ is of fundamental importance. In particular, the \emph{empirical sampling law} $\hat{P}_0\in\hat{\mathscr{P}}^*(\mathcal{W})$ is one \emph{representative} law on the paths of stochastic processes, which is used for sampling the basis functions. The optimization is then over all $\hat{P}_W$ which are \emph{absolutely continuous} with respect to $\hat{P}_0$, so that $\frac{d\hat{P}_W}{d\hat{P}_0}(i)$ is the Radon-Nikodym derivative. This construction arises naturally, as the empirical sampling law $\hat{P}_0$ determines the sample of the $l$ paths, over which an optimal law $\hat{P}_W$ will be chosen to solve the programs \eqref{minimizereqapprox}. 

$\hat{P}_W$ must by construction be absolutely continuous with respect to $\hat{P}_0$ as it can only place positive mass on the $i$ sampled paths, which have been determined via $\hat{P}_0$. In other words, $\hat{\mathscr{P}}^*(\mathcal{W})$ is the set of all probability measures which do not put positive measure on paths other than the $l$ paths sampled via $\hat{P}_0$. This implies a natural assumption on the data-generating process in the population.

\begin{assumption}[Representative law $P_0$ of $\mathscr{P}^*(\mathcal{W})$]\label{radonnikodymass}
The sampling law $P_0$ is a representative law of $\mathscr{P}^*(\mathcal{W})$ in the sense that (i) $P_0\in\mathscr{P}^*(\mathcal{W})$ and (ii) every $P_W\in\mathscr{P}^*(\mathcal{W})$ is absolutely continuous with respect to $P_0$ with Radon-Nikodym derivative satisfying $\sup_{w\in\mathcal{W}}\frac{dP_W}{dP_0}(w)\leq C_{RN}<+\infty$ for a fixed constant $C_{RN}$.
\end{assumption}
Assumption \ref{radonnikodymass} is the theoretical analogue to the fact that all $\hat{P}_W$ are absolutely continuous with respect to $\hat{P}_0$ by construction. Since by the finite dimensional construction all measures $\hat{P}_W$ are automatically absolutely continuous with respect to $\hat{P}_0$, theoretical measures $P_W$ which are not absolutely continuous to $P_0$ can never be detected. This implies that Assumption \ref{radonnikodymass} is non-testable. 

In practice, it is through $P_0$ that the researcher introduces functional form restrictions into the model. For instance, if one wants to assume continuity in the relation between $Y$ and $X$, one will choose a $P_0$ which only puts positive measure on continuous paths. $\hat{P}_0$ will then only sample continuous paths in practice. This way, one can theoretically introduce any form of functional form restriction into the model. In Example 2 for instance, one would define $P_0$ to be a measure that only puts positive probability on paths $X_t$ on the $45$-degree line with finitely many jumps to zero. 

Theorem \ref{maintheorem2}, the main approximation result, requires some more regularity of the representer. In particular, it requires that $P_W$ does not have atoms. This implies that any constraint in \eqref{minimizereq} has zero influence on the overall constraint. Since \eqref{minimizereq} allow for a continuum of constraints, they still restrict the problem jointly. Below, right after the statement of Theorem \ref{maintheorem2} we argue intuitively what this requirement of absolutely continuous $P_W$ means in other settings.
\begin{assumption}[Smoothness of the representative law]\label{repsmoothnessass}
$P_0\in\mathscr{P}^*(\mathcal{W})$ is absolutely continuous with respect to Lebesgue measure. Furthermore, $\frac{dP_W}{dP_0}$ is $\beta$-H\"older continuous with $\beta>\frac{1}{2}$ for any $P_W\in\mathscr{P}^*(\mathcal{W})$.\footnote{A function $f:\mathcal{X}\to\mathbb{R}$ is $\beta $-H\"older continuous if $\sup_{x\neq x'\in\mathcal{X}}\frac{|f(x)-f(x')|}{\|x-x'\|^\beta}<+\infty$ \citep*[p.~138]{folland2013real}. This assumption is reasonable for instance when the sampling law $P_0$ generates a diffusion process and one uses Girsanov's theorem \citep*[p.~190ff.]{karatzas1998brownian} to change the diffusion to another absolutely continuous measure.}
\end{assumption}

\paragraph{2. The wavelet basis}\mbox{}\\This articles works with a shape-preserving wavelet basis, which is defined as 
\[\varphi_{\kappa j}(x) \coloneqq 2^{\frac{\kappa}{2}}\varphi(2^{\kappa}x-j)\qquad \kappa,j\in\mathbb{Z}\] with $\varphi:\mathbb{R}\to[0,1]$ of the form
\[\varphi(x) \coloneqq \begin{cases} x+1&\text{if $-1\leq x\leq 0$}\\ 1-x&\text{if $0< x\leq 1$}\\ 0&\text{otherwise}\end{cases}.\] 
Based on this the notation for the paths sampled via this wavelet basis for dilation $\kappa$ is 
\begin{equation}\label{waveleteq}
\tilde{Y}^\kappa_{x}(i) \coloneqq \sum_{j=-\infty}^{\infty} \alpha_j(i)\varphi_{\kappa j}(x)\quad\text{and}\quad \tilde{X}^\kappa_{z} (i) \coloneqq \sum_{j=-\infty}^{\infty} \gamma_j(i)\varphi_{\kappa j}(z),
\end{equation} where the sums in the definition are both finite since the unit interval is bounded.
This wavelet basis preserves shapes in the sense that an approximation of a monotone (convex) function via this basis will itself be monotone (convex), see \citet*{anastassiou1992monotone} and \citet*{anastassiou1992convex}; this is an important feature when introducing shape assumptions into the model in practice.

Other approaches work equally well. For instance, if one is willing to assume that all paths are smooth, one could simulate paths from a standard diffusion process via the Karhunen-Lo\`eve transform. In this case, one would have to extend Proposition \ref{myprop} to the case where $Y$, $X$, and $Z$ are not compact. This can be done based on results in \citet*[section 16]{billingsley1999convergence}.

\paragraph{3. Regularity assumptions on the objective function}\mbox{}\\
One also needs to impose some regularity on the objective function. 
\begin{assumption}[Regularity of objective function]\label{objectivefunctionstronger}
The kernel $f(Y_x,x)$ is either 
\begin{itemize}
\item[(i)] bounded and $\alpha$-H\"older continuous in its first argument with constant $K<+\infty$ or
\item[(ii)] takes the form of an indicator function, i.e.~$\mathds{1}_{[0,y]}(Y_{x_0}(w))$ for some events $A_y\subset\mathcal{Y}$ and $A_x\subset\mathcal{X}$.\footnote{Recall that we can allow for more general indicator functions of the form $\mathds{1}_{A_y}(Y_{A_x}(w))$ for some events $A_y\subset\mathcal{Y}$ and $A_x\subset\mathcal{X}$.}
\end{itemize}
\end{assumption}
The assumption allows for general objective functions, including the ones from Section \ref{intuitionsec}. Lipschitz continuity is a very weak requirement. Also, allowing for indicator functions is important in order to allow for counterfactual probabilities as argued above.

Before stating the main approximation result, we rewrite \eqref{minimizereq} and \eqref{minimizereqapprox} as penalized programs. The reason is that in this form, we can use the standard theory of $M$-estimators \citep*[chapter 3.4]{wellner2013weak} to obtain the approximation results.
\begin{lemma}[Penalized versions of \eqref{minimizereq} and \eqref{minimizereqapprox}]\label{penalizedlemma}
Under Assumption \ref{radonnikodymass} the solutions to the programs \eqref{minimizereq} coincide with the solutions of the following programs 
\begin{multline}\label{minimizereqpen}
\underset{\substack{\frac{dP_w}{dP_0}\\P_0,P_W\in\mathscr{P}^*(\mathcal{W})}}{\min/\max} \int f(Y_{x}(w),x)\frac{dP_W}{dP_0}(w)dP_0(w)\\
+\lambda\int\left\|F_{Y,X|Z=z} - \mathds{1}_{[0,\cdot]\times[0,\cdot]}\left(Y_{X_z(w)}(w),X_z(w)\right)\frac{dP_W}{dP_0}(w)\right\|^2_{L^2([0,1]^2)}dP_0(w),
\end{multline}
as $\lambda\to\infty$ for $P_Z$-almost all $z\in\mathcal{Z}$.
Analogously, for each $\varepsilon>0$ there exists a $\lambda(\varepsilon)<+\infty$ such that the solutions to the programs \eqref{minimizereqapprox} coincide with the solutions to the following programs
\begin{multline}\label{minimizereqapproxpen}
\underset{\substack{\frac{d\hat{P}_W}{d\hat{P}_0}\\\hat{P}_0,\hat{P}_W\in\hat{\mathscr{P}}^*(\mathcal{W})}}{\min/\max}\frac{1}{l}\sum_{i=1}^l\left[f(\tilde{Y}^\kappa_{x}(i),x)\frac{d\hat{P}_W}{d\hat{P}_0}(i) \vphantom{\lambda_0\left\|F_{Y,X|Z=z} - \mathds{1}_{[0,\cdot]\times[0,\cdot]}(\tilde{Y}_x^\kappa(i),\tilde{X}_z^\kappa(i))\frac{d\hat{P}_W}{d\hat{P}_0}(i)\right\|^2_{L^2([0,1]^2)}}\right.\\
\left.+\lambda(\varepsilon)\left\|F_{Y,X|Z=z} - \mathds{1}_{[0,\cdot]\times[0,\cdot]}\left(\tilde{Y}_{\tilde{X}^\kappa_{z}(i)}^\kappa(i),\tilde{X}_z^\kappa(i)\right)\frac{d\hat{P}_W}{d\hat{P}_0}(i)\right\|^2_{L^2([0,1]^2)}\right],
\end{multline}
for $P_Z$-almost all $z\in\mathcal{Z}$.
\end{lemma}

Lemma \ref{penalizedlemma} is important, as the alternative programs \eqref{minimizereqpen} and \eqref{minimizereqapproxpen} are written in such a way that we can apply the standard theory of $M$-estimators \citep*[chapter 3.2]{wellner2013weak} to the Radon-Nikodym derivatives $\frac{dP_W}{dP_0}$ in order to obtain the concentration bounds. Note how the $L^2$ norm and the integral with respect to $P_0$ have switched places. We view the programs as $M$-estimators with respect to the random sample $l$ of paths. All classical concentration results are applicable in this setting, the only difference is that we now apply them with respect to the artificially introduced randomness by sampling paths. From now on, we will focus exclusively on the programs \eqref{minimizereqpen} and \eqref{minimizereqapproxpen} when deriving the properties of our approach.

The following theorem is the main result of the article and gives probabilistic guarantees of for the approximation of the infinite programs \eqref{minimizereqpen} by the semi-infinite programs \eqref{minimizereqapproxpen}. For this bound, we need to require at Lipschitz-continuity of the objective function, which is not satisfied under part (ii) of Assumption \ref{objectivefunctionstronger}. In this case, we need to approximate the indicator function $\mathds{1}_{[0,y]}(Y_{x_0}(w))$ by a logistic function $\mathcal{S}(Y_{x_0}(w)),y,\eta)$ on $[0,1]$, which takes the form 
\[\mathcal{S}(Y_{x_0}(w),y,\eta)\coloneqq\frac{1}{1+e^{-\eta(y-Y_x(w))}}.\] 

\begin{theorem}[Finite probabilistic approximation via sampling paths]\label{maintheorem2}
Denote by $V^*$ and $V_*$ the value functions of \eqref{minimizereqpen} and by $\tilde{V}^*_{l,\kappa}$ and $\tilde{V}_{*,l,\kappa}$ the value functions of \eqref{minimizereqapproxpen} for some fixed distribution $F_{Y,X|Z}$, respectively.
Under Assumptions \ref{skorokhodpath} -- \ref{repsmoothnessass} and part (i) of \ref{objectivefunctionstronger}, it holds with probability $1-\rho$
\begin{multline}\label{mainapproxeq1}
\max\{|V^* - \tilde{V}^*_{l,\kappa}|,|V_*-\tilde{V}_{*,l,\kappa}|\}\\
\leq \left[\sup_{w\in[0,1]}\omega'_{Y_x(w)}(2^{-\kappa+1})\omega'_{X_z(w)}(2^{-\kappa+1})C_{RN}+\sup_{w\in[0,1]}K\left(\omega'_{Y_x(w)}(2^{-\kappa+1})\right)^\alpha \right]+\sqrt{\frac{\log\left(\frac{\bar{C}}{\rho}\right)}{\bar{D}l}}
\end{multline}
for every number in the series approximation $\kappa\in\mathbb{N}$ and all $z\in\mathcal{Z}$.
$0<\bar{C}<+\infty$ and $0<\bar{D}<2$ are constants depending on $\kappa$, the H\"older coefficient $\beta$, $F_{Y,X|Z=z}$, the bound on the Radon-Nikodym derivative $C_{RN}$, the penalty term $\lambda$, and the Lipschitz constant $K$ of the kernel $f$ of the objective.

Under part (ii) of Assumption \ref{objectivefunctionstronger}, it holds with probability $1-\rho$ 
\begin{multline}\label{mainapproxeq2}
\max\{|V^*-\tilde{V}^*_{l,\kappa}|,|V_*-\tilde{V}_{*,l,\kappa}|\}\\
\leq \left[\sup_{w\in[0,1]}[\omega'_{Y_x(w)}(2^{-\kappa+1})+1]\omega'_{X_z(w)}(2^{-\kappa+1})C_{RN}\right]+ \sqrt{\frac{\log\left(\frac{\bar{C}}{\rho}\right)}{\bar{D}l}}\\
+K\frac{\log(\eta+1)}{\eta+1}\left(1+O\left(\frac{\log\log(\eta+1)}{\log(\eta+1)}\right)\right),
\end{multline}
for all $z\in\mathcal{Z}$ and where $K=\frac{\eta\cdot e^{\eta(y+Y_{x_0}(w))}}{\left(e^{\eta Y_{x_0}(w)}+e^{\eta y}\right)^2}$.
\end{theorem}

Theorem \ref{maintheorem2} jointly provides a lower bound on the number of sampled paths $l$ as well as the number of terms in the wavelet decomposition $\kappa$ for the semi-infinite program to provide a good approximation to the infinite dimensional program for a requested probability $\rho$. In fact, it is a finitary analogue of a classical consistency result, as it also provides the rate of convergence as $l,\kappa\to\infty$. Theorem \ref{maintheorem2} holds for a fixed observable distribution $F_{Y,X|Z}$; this observable distribution can be an estimator $\hat{F}_{Y,X|Z;n}$, in which case we consider the number of data points $n$ as fixed. The result follows from concentration results for empirical processes with functions of finite entropy numbers \citep*[chapter 2.14]{wellner2013weak}. The bound on the right hand side is composed of two terms. The first term in brackets captures the nonprobabilistic approximation of the value functions when the paths $Y_x$ and $X_z$ are approximated by the wavelet bases $\tilde{Y}_x^\kappa$ and $\tilde{X}_z^\kappa$, respectively and depends on the extended modulus of continuity $\omega'$ of the paths of either process. The two terms here correspond to the constraint and the objective function, respectively. The second term is the probabilistic concentration bound based on the sample of $l$ paths. The constants here also depend on $\kappa$. In the case where the objective function is an indicator function, we need an additional approximation term which follows from approximating the indicator function by a logistic function.

\paragraph{Intuition for Theorem \ref{maintheorem2} and comparison to other approaches}\mbox{}\\
Theorem \ref{maintheorem2} is but one example of many different approximation guarantees of similar form, depending on what kind of assumptions one is willing to make on the paths (via $P_0$) and $\frac{dP_W}{dP_0}$. In fact, it is the introduction of the randomness via sampling that permits the use of statistical procedures for the purpose of function approximation. The idea is to perform statistical estimation \emph{with respect to the counterfactual elements via sampling}. In this respect, Assumption \ref{repsmoothnessass} is important for obtaining the rate of convergence in Theorem \ref{maintheorem2}: if all Radon-Nikodym derivatives $\frac{dP_W}{dP_0}(w)$ are smooth and have support in all of $[0,1]$, one does not need many draws $l$ to approximate the optimization problem. For more concentrated densities, the approximation can be much worse in the sense that one needs a substantially larger sample $l$ to approximate the optimal density. The approach proposed via Theorem \ref{maintheorem2} can hence be seen as a nonparametric estimation of a probability density with infinite dimensional support (i.e.~where each data-point is a path of a process) and where one has complete control over the data sample. Different assumptions on the smoothness of the densities will lead to different lower bounds on the number of sampled paths. 

Theorem \ref{maintheorem2} is a quantitative approximation result which is significantly harder, if not impossible, to achieve with an approach that obtains a complete set of inequalities describing the identified set (e.g.~\citeauthor*{chesher2017generalized} \citeyear{chesher2017generalized}). In fact, Assumption \ref{repsmoothnessass} has potential analogues in the inequality approach: smoothness of the Radon-Nikodym derivative is similar to the assumption that the identified set can be described by (uncountably) many inequalities which all only contribute minimally to the identified set. If one is not willing to make this assumption, it can happen that the Radon-Nikodym derivatives are much more concentrated, i.e.~putting significantly more weight on only a few paths. In this case, the bound in Theorem \ref{maintheorem2} implies that one needs to sample more paths depending on the new assumption on $P_0$ and $\frac{dP_W}{dP_0}$. The latter setting corresponds to the case where there are only a few inequalities which ``determine the form of the identified set''. In this case, one would also need to sample many more inequalities in order to ``stumble upon'' the few which determine the shape of the set (see the argument in \citeauthor*{de2004constraint} \citeyear{de2004constraint}). 

Theorem \ref{maintheorem2} hence complements the existing results on the probabilistic approximation of infinite dimensional (linear) programs by semi-infinite dimensional programs proposed in \citet*{girosi1995approximation} and \citet*{de2004constraint}. In particular, the method allows for standard statistical methods to obtain approximation results by introducing the sampling probability $P_0$, which leads to flexible and general quantitative results like Theorem \ref{maintheorem2}. In contrast, \citet*{de2004constraint} use a learning result (Theorem 8.1.4 in \citeauthor{anthony1997computational} \citeyear{anthony1997computational}) to obtain a probabilistic approximation result, which is not as flexible as using nonparametric estimation results. Intuitively, Theorem \ref{maintheorem2} puts a probability measure on the ``set of all inequalities'' in a classical linear program. Instead of working with inequalities like \citet*{de2004constraint} or \citet*{chesher2017generalized}, we work with stochastic processes and put a probability measure on the paths.

\subsection{Inference results}\label{inferencesec}
The idea of the sampling approach is to introduce additional randomness into the problem by sampling $l$ paths, i.e.~sampling the coefficients of the basis functions.
The classical statistical randomness in problems \eqref{minimizereqpen} and \eqref{minimizereqapproxpen} follows from approximating the population distribution $F_{Y,X|Z}$ by a finite-sample estimator $\hat{F}_{Y,X|Z;n}$, potentially smoothed via some bandwidth $h_n$, where $n$ denotes the size of this sample.\footnote{The theoretical results in this section are derived for the standard empirical cumulative distribution function $\hat{F}_{Y,X|Z;n}$. They extend straightforwardly to smoothed estimators $\hat{F}_{Y,X|Z;h_n}$. For this, all one has to do in the proofs is to replace the classical Glivenko-Cantelli and Donsker theorems by analogous versions for smoothed empirical processes. These results (and the corresponding weak assumptions) are contained in \citet*{gine2008uniform} for instance.} 

This subsection introduces large sample results which enable the researcher to perform inference on the solution of the programs \eqref{minimizereqpen}. These results are only derived for each bound separately, i.e.~for the lower bound and the upper bound. In order to derive inference bounds on the whole identified set, it might be possible to use well-established results from the literature, such as \citet*{imbens2004confidence}, \citet*{stoye2009more}, and especially \citet*{kaido2019confidence}.

Even though the theoretical programs \eqref{minimizereqapproxpen} have relaxed constraints, it could potentially still be the case that for very large $\lambda$ there exist data-generating processes $F_{Y,X|Z}$ that no $P_{W}\in\mathscr{P}^*(\mathcal{W})$ can replicate; this is especially true in the case where many additional form restrictions are imposed on the model via $P_0$. In this case, however, the data-generating process directly introduces testable assumptions on the model, as the model under the respective assumptions is not able to replicate the observed data. Since the focus of this article is on estimation of bounds, it is convenient to introduce an assumption on the data-generating process which guarantees that the constraint is non-empty. 

In the following, $\mathcal{F}_{Y,X|Z=z}$ denotes a set of all conditional cumulative distribution functions on $[0,1]^2$ satisfying certain assumptions the researcher is comfortable to assume for the given data-generating process. This set is equipped with the $L^\infty([0,1]^2)$-norm. 
\begin{assumption}[Non-emptiness of the constraint set] \label{nonemptinessass}
For given $F_{Y,X|Z=z}\in\mathcal{F}_{Y,X|Z=z}$ there exists a ball $\mathcal{B}_r\in\mathcal{F}_{Y,X|Z=z}$ of radius $r>0$ such that the constraint set $\mathcal{C}\coloneqq$ \[\left\{P_W\in \mathscr{P}^*(\mathcal{W}): \int\left\|F'_{Y,X|Z=z}(y,x) - \mathds{1}_{[0,\cdot]\times[0,\cdot]}(Y_{X_z(w)}(w),X_{z}(w))\frac{dP_W}{dP_0}(w)\right\|_{L^2([0,1]^2)}^2dP_0\leq \varepsilon\right\}\]
is non-empty for some small $\varepsilon>0$ and all $F'_{Y,X|Z=z}\in\mathcal{B}_r$, $z\in[0,1]$.
\end{assumption}

Assumption \ref{nonemptinessass} is deliberately high-level, because (i) specific assumptions on the data-generating process $\mathcal{F}_{Y,X|Z=z}$ and the model $\mathscr{P}^*(\mathcal{W})$ usually come from economic theory, (ii) an empty constraint for a given data-generating process corresponds to the existence of testable implications on the model, and (iii) Assumption \ref{nonemptinessass} is only required to obtain regular and well-behaved asymptotic results, but not for any other results in this article. It is also straightforward to derive a low-level sufficient condition on $\mathcal{F}_{Y,X|Z=z}$ implying Assumption \ref{nonemptinessass} in the case where only Assumption \ref{skorokhodpath} but no other shape restrictions are imposed on $\mathscr{P}^*(\mathcal{W})$. 

For instance, when the set $\mathcal{F}_{Y,X|Z=z}$ consists only of distribution functions $F_{Y,X|Z=z}$ which are laws to stochastic processes $[Y,X]_z$ whose paths have an extended modulus of continuity $\omega'_{[Y,X]_z}(\delta)$ satisfying $\limsup_{\delta\to0} \omega'_{[Y,X]_z}(\delta) = 0$, then Assumption \ref{nonemptinessass} is fulfilled. This follows directly from the fact that this condition on the extended modulus of continuity implies that almost all paths of the observable process $[Y,X]_z$ lie in the Skorokhod space defined by Assumption \ref{skorokhodpath}, which is the assumption made on the latent process $[Y,X]^*_z$.

$\hat{F}_{Y,X|Z;n}$ denotes the conditional empirical distribution function, $\hat{\tilde{V}}_{*,l,\kappa}(\hat{F}_{Y,X|Z;n})$ and $\hat{\tilde{V}}^*_{l,\kappa}(\hat{F}_{Y,X|Z;n})$ denote the value functions of \eqref{minimizereqapproxpen} when replacing $F_{Y,X|Z}$ by an estimator $\hat{F}_{Y,X|Z;n}$ and $V_{*}(F_{Y,X|Z})$ and $V^*(F_{Y,X|Z})$ denote their counterparts in the population. 

The first result concerns the consistency and is based on the Glivenko-Cantelli theorem \citep*[Theorem 19.1]{van2000asymptotic} which provides the convergence $\hat{F}_{Y,X|Z;n}$ to $F_{Y,X|Z}$ in $L^\infty([0,1]^2)$-norm. 
\begin{proposition}[Consistency]\label{consistencyprop}
Under Assumptions \ref{skorokhodpath} -- \ref{nonemptinessass}
\[P\left(|\hat{\tilde{V}}_{*,l,\kappa}(\hat{F}_{Y,X|Z=z;n}) - V_{*}(F_{Y,X|Z=z})|\right)\to0\quad\text{and}\quad P\left(|\hat{\tilde{V}}^*_{l,\kappa}(\hat{F}_{Y,X|Z=z;n}) - V^*(F_{Y,X|Z=z})|\right)\to0\]
as $l,\kappa,n\to\infty$ for all $\lambda\in\mathbb{R}$ and $z\in\mathcal{Z}$.
\end{proposition}

Similar to the consistency result, the derivation of the large sample distribution follows form Donsker's theorem \citep*[Theorem 19.3]{van2000asymptotic} in combination with standard sensitivity arguments in optimization problems \citep*{bonnans2013perturbation} and the functional delta method \citep*[Theorem 2.1]{shapiro1991asymptotic}. This is captured in the following
\begin{proposition}[Asymptotic distribution]\label{largesampleprop}
If Assumptions \ref{skorokhodpath} -- \ref{nonemptinessass} hold and if 
\begin{equation}\label{weneedconv}
\begin{aligned}
&\sqrt{n}\left[\sup_{w\in[0,1]}\omega'_{Y_x(w)}(2^{-\kappa+1})\omega'_{X_z(w)}(2^{-\kappa+1})C_{RN}+\sup_{w\in[0,1]}K\left(\omega'_{Y_x(w)}(2^{-\kappa+1})\right)^\alpha\right]\to0\qquad\text{and}\\
&\frac{n}{l}\to 0 
\end{aligned}
\end{equation}
as $n\to\infty$, then
\begin{align*}
&\sqrt{n}(\hat{\tilde{V}}_{*,l,\kappa}(\hat{F}_{Y,X|Z=z;n})-V_{*}(F_{Y,X|Z=z}))\rightsquigarrow dV_{*,F_{Y,X|Z=z}}(\mathbb{G}_{F_{Y,X|Z=z}})\qquad\text{and}\\
&\sqrt{n}(\hat{\tilde{V}}^*_{l,\kappa}(\hat{F}_{Y,X|Z=z;n})-V^*(F_{Y,X|Z=z}))\rightsquigarrow dV^*_{F_{Y,X|Z=z}}(\mathbb{G}_{F_{Y,X|Z=z}}),
\end{align*}
for almost all $z\in\mathcal{Z}$ as $n\to\infty$. 
\end{proposition}
Here, $dV_{*,F_{Y,X|Z=z}}(F'_{Y,X|Z=z})$ is the directional Hadamard derivative of $V_*$ defined by \eqref{minimizereqpen} at $F_{Y,X|Z=z}$, $\mathbb{G}_{F_{Y,X|Z=z}}$ is a Brownian bridge with covariance function 
\[\text{Cov}_{\mathbb{G}_{F_{Y,X|Z=z}}} = F_{Y,X|Z=z}(\min\{y,y'\},\min\{x,x'\})-F_{Y,X|Z=z}(y,x)F_{Y,X|Z=z}(y',x')\] for all $(y,x),(y',x')\in[0,1]^2$, $z\in\mathcal{Z}$, and ``$\rightsquigarrow$'' denotes weak convergence.

The directional Hadamard derivative takes the form 
\[\delta_{F_{Y,X|Z}}V_{*}(F) = \min_{h\in\mathcal{S}(F_{Y,X|Z})}\int2\lambda\left\langle F, F_{Y,X|Z}-\mathds{1}_{[0,\cdot]\times[0,\cdot]}(Y_{X_z(w)}(w),X_z(w))\frac{dP_W}{dP_0}(w)\right\rangle dP_0(w),\]
where $\mathcal{S}(F_{Y,X|Z})$ is the solution set of $V_{*}(F_{Y,X|Z})$, i.e.~the set of all $\frac{dP_W}{dP_0}(w)$ with $P_W,P_0\in\mathscr{P}^*(\mathcal{W})$ that solve \eqref{minimizereqpen} for $F_{Y,X|Z}$.\footnote{Proposition \ref{largesampleprop} works equally well with a smoothed estimator $\hat{F}_{Y,X|Z=z;h_n}$ of $F_{Y,X|Z=z}$ and bandwidth $h_n$. The only important requirement is that the respective empirical process converges to a Brownian bridge.} $\left\langle f,g\right\rangle$ is the inner product in $L^2([0,1]^2)$, i.e.
\[\left\langle f,g\right\rangle\coloneqq\int_{[0,1]^2}f(y,x)g(y,x)dydx.\]

\eqref{weneedconv} requires that the sample size of the paths grows faster than the data, and the wavelet approximation $\kappa$ is also sufficiently fast given the continuity properties of the paths. Even though the large sample distribution of the value functions is not a standard Brownian bridge process, it still has a relatively common form from a purely statistical perspective, as it takes the form of the first-order directional Hadamard derivative of the value function taken at $F_{Y,X|Z=z}$ in directions $F'_{Y,X|Z=z}\in\mathcal{F}_{Y,X|Z=z}$. 

In addition, there are several results in the literature (\citeauthor*{dumbgen1993nondifferentiable} \citeyear{dumbgen1993nondifferentiable}, \citeauthor*{fang2018inference} \citeyear{fang2018inference}, \citeauthor*{hong2018numerical} \citeyear{hong2018numerical}) which establish bootstrap methods for estimating this type of large sample distribution in practice. In particular, they deal with general directional Hadamard differentiability \citep*{shapiro1991asymptotic}, which conforms with Proposition \ref{largesampleprop}, so that these subsampling/bootstrap results are directly applicable to the problems \eqref{minimizereqapprox}. These bootstrap-type arguments are convenient in models with a light computational burden mostly. In more complex models one should use the analytically derived large sample theory.

\section{Practical implementation}\label{practicalimp}
The programs \eqref{minimizereqapprox} are semi-infinite programs \citep*{anderson1987linear}, which naturally reduce to finite dimensional problems in practice by approximating the space $[0,1]^3$ where $Y$, $X$, and $Z$ live. One can do this in two general ways. The first is to simply evaluate $\hat{F}_{Y,X|Z=z;n}$ on the values taken by the sample $(Y_i,X_i,Z_i)_{i=1,\ldots,n}$. The second is to evaluate $\hat{F}_{Y,X|Z=z;n}$ on a finite grid that spans $[0,1]^3$. This article focuses on the latter part as a grid approach gives more flexibility with respect to the computational requirements: one can make the grid coarser or finer, depending on the available memory.\footnote{This is also a difference to the computational approach in \citet*{kitamura2018nonparametric}, who set up their problem based on the observed realizations in the data and not a grid. This is a direct consequence of their finitary approach. In contrast, the infinite approach in this article allows arbitrary discretization of $Y,X$, and $Z$ and is therefore also a potential complementary approach to \citet*{kitamura2018nonparametric} in their setting.} Throughout this section, the index $\iota$ captures the degree of approximation of the grid. For instance, $\iota=11$ means that this approximation decomposes the unit interval into $11$ points $0,0.1,0.2,\ldots,0.9,1$, which will be taken to be equidistant without loss of generality. Also, and without loss of generality, all three intervals for $Y$, $X$ and $Z$ are decomposed in the same way, so that $\iota$ is the only necessary parameter controlling the approximation.

The practical implementation deviates from the theoretical approach in that it uses a smoothed variant $\hat{F}_{Y,X|Z=z;h_n}$ of the empirical conditional cumulative distribution function, where the bandwidth is determined via cross-validation. Heuristically, it seems as though the introduced smoothness gives more robust results compared to the standard empirical cumulative distribution function.\footnote{For the practical estimation of $\hat{F}_{Y,X|Z=z;h_n}$ the method uses the ``np''-package in $R$ \citep*{hayfield2008nonparametric} with a standard cross-validated bandwidth.}

Under a given finite approximation, the programs take the form
\begin{equation}\label{minimizereqapprox2pen}
\underset{\mu\geq0, \vec{1}'\mu\leq 1}{\text{minimize/maximize}} \qquad\Xi'\mu+\frac{\lambda}{2}\|\tilde{\Theta}\mu - \hat{F}_{Y,X|Z;h_n}\|_2^2
\end{equation}
where $\mu$ is a $l\times 1$ vector which corresponds to the Radon-Nikodym derivative $\frac{d\hat{P}_W}{d\hat{P}_0}(i)$ with row-dimension equal to the number of sampled paths $l$,\footnote{Note that all elements in $\mu$ must lie in $[0,1]$, as $\frac{d\hat{P}_W}{d\hat{P}_0}(i)$ is defined on the \emph{finite and discrete space} of $l$ paths which were sampled by some $\hat{P}_0$. This means that $\frac{d\hat{P}_W}{d\hat{P}_0}(i)$ can only put non-negative probabilities of at most one on the occurrence of each path. Intuitively, this follows from the fact that $\frac{d\hat{P}_W}{d\hat{P}_0}(i)$ is a probability mass function. These bounds on $\mu$ are included as additional constraints using $\vec{1}$.} $\vec{1}$ denotes the vector of the same dimension as $\mu$ containing all ones, $\Xi'$ is a $1\times l$ vector corresponding to $f(Y_x,x_0)$, and $\|\cdot\|_2$ denotes the Euclidean norm. $A'$ denotes the transpose of the matrix $A$. $\tilde{\Theta}$ is a $\iota^3\times l^2$-matrix which maps the realization of the stochastic processes to the distribution $\hat{F}_{Y,X|Z;h_n}$.
The $L^2([0,1]^2)$-norm from \eqref{minimizereqapprox} reduces to the Euclidean norm due to the approximation of $[0,1]^3$ by a finite grid. 

The choice of the Euclidean norm $\|\cdot\|_2$ for the constraint is convenient, as \eqref{minimizereqapprox2pen} can be rewritten as
\begin{equation}\label{minimizereqapprox2penfinal}
\begin{aligned}
&\underset{\mu\geq0, \vec{1}'\mu\leq 1}{\min} \quad\frac{\lambda}{2}\mu' \tilde{\Theta}'\tilde{\Theta}\mu - \left(\lambda\tilde{\Theta}'\hat{F}_{Y,X|Z;h_n}-\Xi\right)'\mu+ \frac{\lambda}{2}\left(\hat{F}_{Y,X|Z;h_n}\right)'\hat{F}_{Y,X|Z;h_n}\\
&\underset{\mu\geq0, \vec{1}'\mu\leq 1}{\min} \quad\frac{\lambda}{2}\mu' \tilde{\Theta}'\tilde{\Theta}\mu - \left(\lambda\tilde{\Theta}'\hat{F}_{Y,X|Z;h_n}+\Xi\right)'\mu+ \frac{\lambda}{2}\left(\hat{F}_{Y,X|Z;h_n}\right)'\hat{F}_{Y,X|Z;h_n}.
\end{aligned}
\end{equation}
The programs \eqref{minimizereqapprox2penfinal} are quadratic due to the Euclidean norm used and can easily be solved. This article uses the alternating direction method of multipliers (ADMM) (\citeauthor*{boyd2011distributed} \citeyear{boyd2011distributed} and \citeauthor*{parikh2014proximal} \citeyear{parikh2014proximal}) for optimization. This algorithm is known to converge rather quickly to reasonable approximations of the optimum, which makes it a perfect tool for this purpose. The algorithm requires two more parameters, the augmented Lagrangian parameter $\rho$ and an over-relaxation parameter $\zeta$, which control the convergence of the ADMM algorithm to the optimum. In practice, an over-relaxation parameter of $\zeta=1.7$ and an augmented Lagrangian parameter of $\rho$ between $100$ and $500$ leads to fast and robust convergence.

The computational bottleneck in a practical implementation is the construction of the matrix $\tilde{\Theta}$, whose dimension grows exponentially with the granulation of the grid $\iota$ and the number of paths sampled $l$. Fortunately, since the penalty terms in \eqref{minimizereqpen} and \eqref{minimizereqapproxpen} include indicator functions, the matrix $\tilde{\Theta}$ takes the form of a binary sparse matrix: for each point $(y_\iota,x_\iota,z_\iota)\in[0,1]^3$ in the grid a given combination of paths $Y_x(i)$ and $X_z(i)$ either gets assigned a $0$ if they jointly ``do not go through'' the intervals $[0,y_\iota]\times [0,x_\iota]$ for given values $z_i$ or a $1$ if they jointly do. This sparseness is helpful as sparse matrices can be stored efficiently. In addition, the process of setting up $\tilde{\Theta}$ can be parallelized, which abates the computational costs even further if the researcher has access to several cores. 

In many cases, however, a researcher might only have access to computational resources with very limited working memory. In such situations, it is still possible to apply the proposed method by a ``sampling trick'' which trades off memory requirements for time. In particular, the idea is to iteratively (i) sample with replacement a relatively small initial number $l_{0}$ of paths (depending on the available memory), (ii) optimize the programs \eqref{minimizereqapprox2penfinal} on this sample, (iii) obtain the value functions \emph{as well as} the optimizers $\mu$, (iv) \emph{drop} all paths which were assigned a probability of (close to) $0$ by the optimizer $\mu$, (v) sample another relatively small number $l_{s}$, add these paths to the already existing paths and go back to (ii). The idea of this ``sampling trick'' is that paths which were assigned a probability of (close to) $0$ by the optimal $\mu$ do not matter for the optimal value. By dropping these paths before sampling new ones, the memory requirements do not grow or only grow modestly in practice---at the additional cost of having to run this optimization for many iterations.\footnote{Discarding elements ex-post in optimization routines is not new. For instance  \citet*{wu2001solving} use discards in solving the general capacity problem on Euclidean state space. Currently, the implemented approach is rather na\"ive as it performs a ``random grid search'' over the infinite dimensional path space. A more efficient implementation uses ideas from sequential MCMC approaches (see \citeauthor{schweizer2012thesis} \citeyear{schweizer2012thesis} for an overview): in the first iteration, randomly sample paths. Then optimize and discard all paths which were assigned a probability of zero. When sampling new paths, do not just sample randomly, but sample a fraction of paths which are close to the paths that were assigned a positive probability and sample another fraction of paths randomly to find other, different binding constraints. An interesting question in this setting is to optimize this procedure, i.e.~to optimize the fraction of randomly sampled paths compared to paths which are close to others. Intuitively, there seems to be an ``exploration-exploitation-tradeoff'', which one could optimize.}

We can now present our algorithm for solving the infinite dimensional linear programs approximately in practice.

\begin{algorithm*}\mbox{}
\begin{itemize}
\item \emph{Initial step}: randomly sample some set 
\[\mathcal{R}_0\coloneqq \left\{\left(Y_{x_\iota}(i), X_{z_\iota}(i)\right), i = 1,\ldots, l_{init}, \iota= 1, \ldots, m_j\right\}\] of initial paths, where $l_{init}$ is the number of initial paths to sample, and where $m_j$ is the number of grid-points on the unit interval based on the dyadic decomposition of order $j$. Sample paths with or without replacement\footnote{In our application we sample with replacement.} with sampling measure $P_0$. Fix some $\delta>0$ and $n_\delta\in\mathbb{N}$, which will control the convergence criterion of the algorithm. Compute the matrix $\Theta_0$ and the corresponding vector $\Xi_0$, based on these paths, either using the indicator function or the logistic approximation $S(Y_x(w),y,\eta)$ for some large $\eta>0$. Set $k=1$ and set 
\[\Theta_0^{max} = \Theta_0^{min} = \Theta_0\qquad\text{and}\qquad \Xi_0^{max} = \Xi_0^{min} = \Xi_0.\]
\item \emph{loop over iterations $k$:} 
\begin{enumerate}
\item Randomly sample a set 
\[\mathcal{R}_k\coloneqq \left\{\left(Y_{x_\iota}(i), X_{z_\iota}(i)\right), i = 1,\ldots, l_{add}, \iota= 1, \ldots, m_j\right\}\] of stochastic paths to add to the program, where $k_{add}$ is the number of paths to add. Sample paths with or without replacement under $P_0$ and make sure the sampled paths are unique. Compute the preliminary matrices $\widetilde{\Theta}^{min}_k$ and $\widetilde{\Theta}^{max}_k$ as well as the vectors $\widetilde{\Xi}_k^{min}$ and $\widetilde{\Xi}_k^{max}$ based on these paths as in the initial step. Update the matrices $\Theta_{k-1}^{min}$ and $\Theta_{k-1}^{max}$ as 
\[\Theta_k^{min} = [\Theta_{k-1}^{min}\thickspace\medspace \widetilde{\Theta}_k^{min}],\qquad \Theta_k^{max} = [\Theta_{k-1}^{max}\thickspace\medspace \widetilde{\Theta}_k^{max}],\] i.e.~by appending the respective columns of $\widetilde{\Theta}_k^{\cdot}$ to $\Theta_{k-1}^{\cdot}$. In addition, update the vectors $\Xi_{k-1}^{min}$ and $\Xi_{k-1}^{max}$ as
\[\Xi_{k}^{min} = \left[\left(\Xi_{k-1}^{min}\right)'\thickspace\medspace \left(\widetilde{\Xi}_{k}^{min}\right)'\right]'\qquad\text{and}\qquad \Xi_{k}^{max} = \left[\left(\Xi_{k-1}^{max}\right)'\thickspace\medspace \left(\widetilde{\Xi}_{k}^{max}\right)'\right]',\] where $A'$ denotes the transpose of the matrix $A$.
\item Solve the programs \eqref{minimizereqapprox2penfinal} and store the optimal solutions to these problems as $V_{k,min}$ and $V_{k,max}$ and the optimizers as $\mu_{k,min}$ and $\mu_{k,max}$. If the moving standard deviations
\[\left(\frac{1}{n-1}\sum_{j=1}^{n_\delta} \left(V_{k-j,min} - \bar{V}_{k,min}\right)^2\right)^{1/2} \leq \delta \quad\text{and}\quad \left(\frac{1}{n-1}\sum_{j=1}^{n_\delta} \left(V_{k-j,max} - \bar{V}_{k,max}\right)^2\right)^{1/2}\leq \delta,\] for the window length $n_\delta$ and the $\delta>0$ chosen in stage 0, stop and output $V_{k,min}$ and $V_{k,max}$ as the solution. Here,
\[\bar{V}_{k,min} = \frac{1}{n}\sum_{j=1}^{n_\delta}V_{k-j,min}\qquad\text{and}\qquad \bar{V}_{k,max} = \frac{1}{n}\sum_{j=1}^{n_\delta}V_{k-j,max}\] are the moving averages with window length $n_\delta\leq k$.
\item Delete all columns from $\Theta_k^{min}$ and all rows from $\Xi_k^{min}$ for which the corresponding values of $\mu_{k,min}$ are zero. Analogously for $\Theta_k^{max}$. 
\item If convergence criterion is met, output solution. Else, increment $k\to k+1$ and go to step 1.
\end{enumerate}
\end{itemize}
\end{algorithm*}

When applying this sampling trick, the solution will be expressed as a \emph{solution path} over the sampling iterations.
\begin{figure}[h!]
\centering
\includegraphics[width=5.5cm,height=5.5cm]{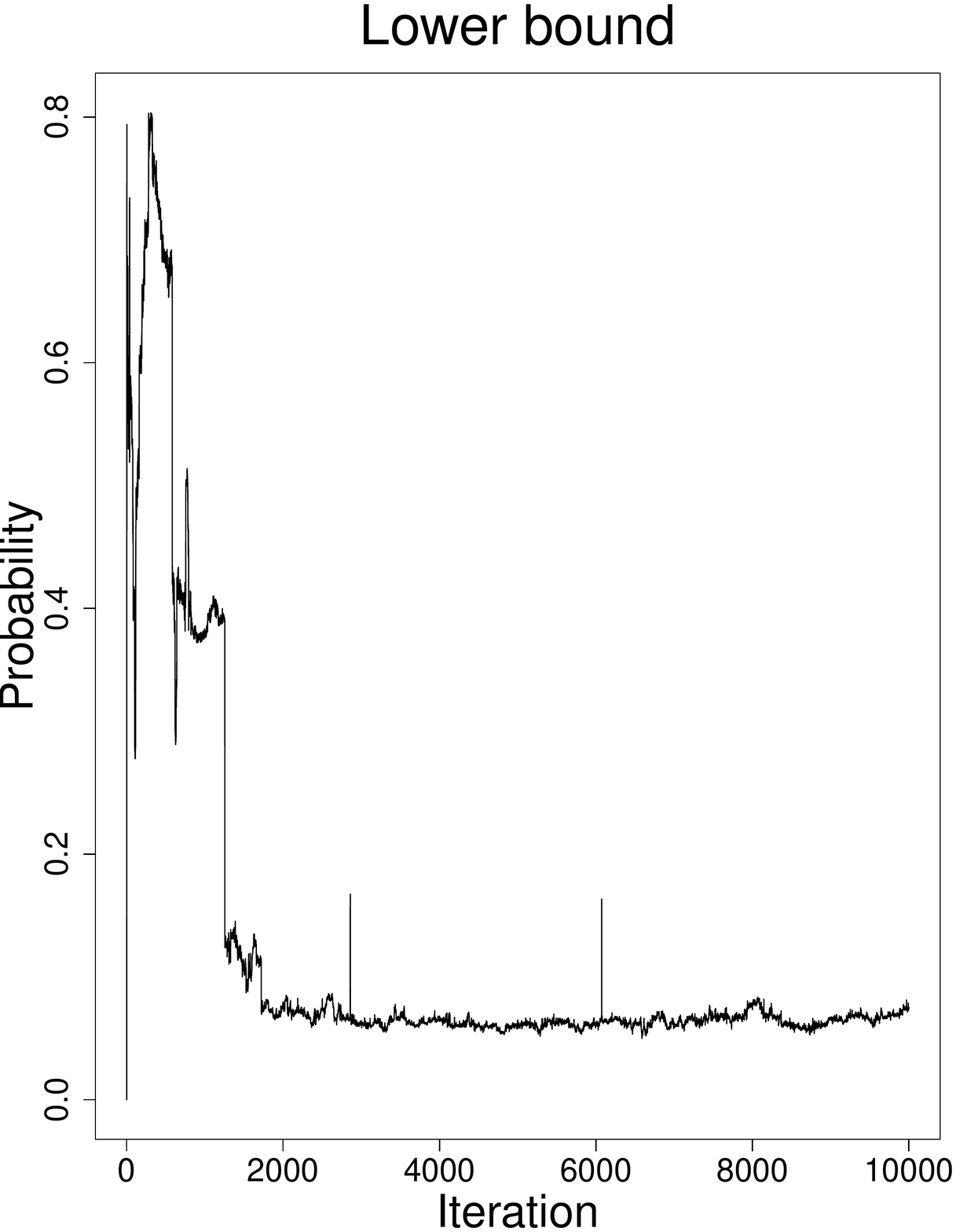}
\includegraphics[width=5.5cm,height=5.5cm]{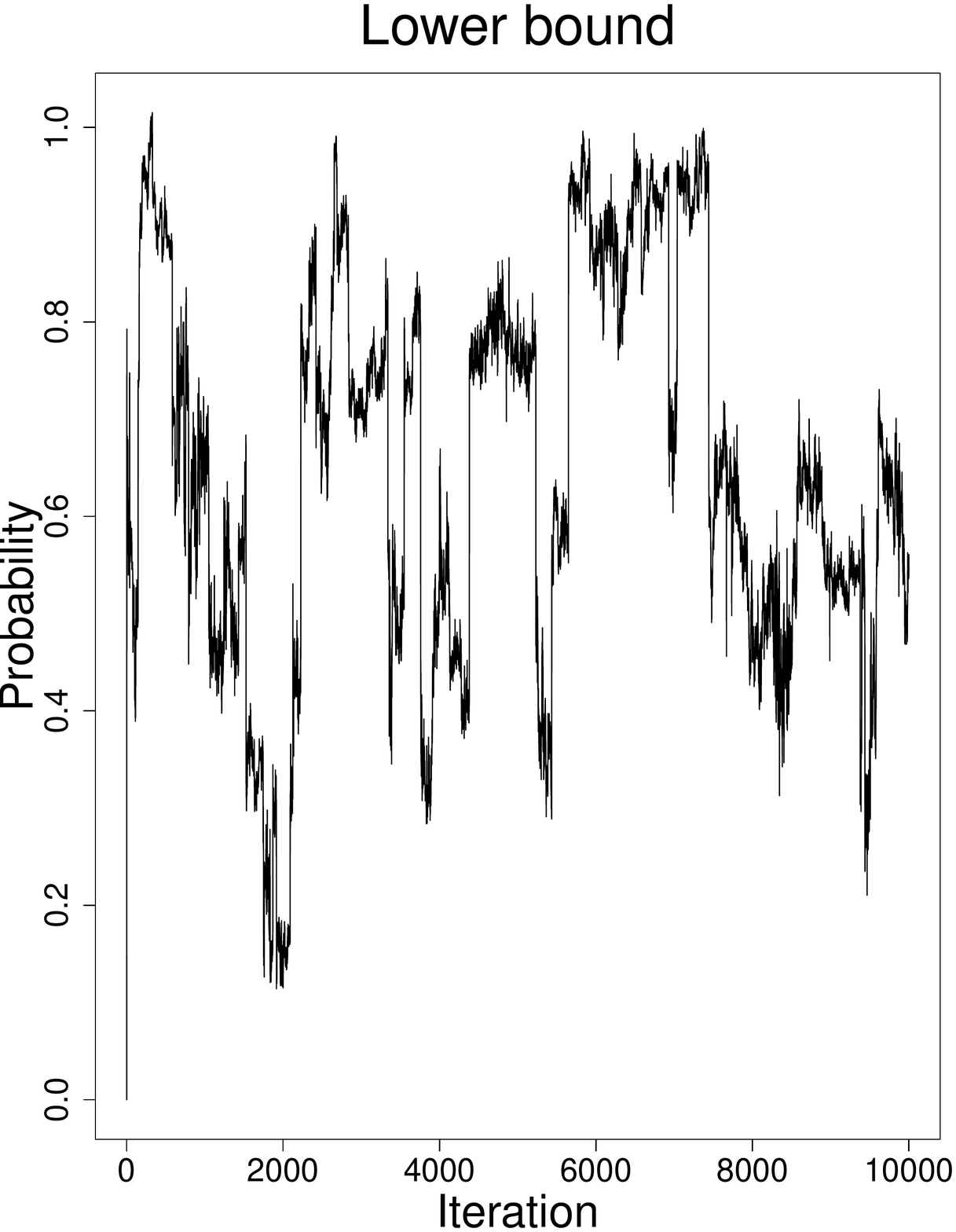}
\caption{Depiction of convergence of the solution path for estimating $F_{X=0.75}(0.75)$ for $\lambda = 100$ (left) and $\lambda=600$ (right) for a coarse approximation of length $5$ of the unit interval, i.e.~a decomposition $0,0.25, 0.5,0.75,1$. $16$ new paths are sampled at each iteration. The average of the values over the last 500 iterations in the left panel is nontrivial at $0.065$.}\label{convergenceplot}
\end{figure}
This solution path in general will be erratic due to the nature of the sampling approach, but has the added benefit over the ``static'' direct method, where all relevant paths are sampled immediately, that one can gauge if the solutions ``converge'' to some stable limit after a ``burn-in'' period. This convergence relies on the choice of the penalty parameter $\lambda$. The larger $\lambda$ is chosen, the more it forces the optimizer $\mu$ to replicate the observable $\hat{F}_{Y,X|Z;h_n}$. In fact, in the limit $\lambda\to+\infty$, the program forces to replicate the constraint perfectly. If $\lambda$ is too low, the program ignores the constraint, which will always result in trivial bounds. 

This implies that there exists a range of $\lambda$-values for which constraint and objective function are balanced. Figure \ref{convergenceplot} depicts the behavior of the solution paths of this ``sampling trick'' in a stylized setting of the household expenditure application in the next section, where the coarseness of the approximation of the unit interval makes it possible to sample \emph{all possible} paths; this is done in order to see whether convergence still occurs if one is able to sample the ``whole universe of paths''. The left panel depicts a solution path which converges, while the right panel depicts a case of non-convergence for estimating a lower bound on $F_{Y(X=0.75)}(0.75)$.

In this form, these solution paths are reminiscent of the solution paths of regularized linear programs such as LASSO. The difference, however, is that the paths induced by this program are for a \emph{fixed} $\lambda$, while the actual LASSO solution paths are traced out while varying $\lambda$. In order to have an analogue of the LASSO solution paths in the current method, one would have to solve the program for many different values of $\lambda$, which would generate a \emph{system of solution paths}. Then one could choose the largest lambda for which the corresponding solution path converges to a stable value.\footnote{It is an intriguing question how to determine an appropriate $\lambda$ by data-driven methods. Such a data-driven method might open up the way for solving other infinite dimensional programs on path spaces in statistics and mathematics via a ``sampling of paths''-approach.}

\section{Empirical demonstration of the method}\label{empiricalsec}
This section presents practical settings for testing the method. The first is a small Monte-Carlo simulation, where the focus lies on the sensitivity of the method to the choice of the number of basis functions $\kappa$ and the penalty term $\lambda$. The second is an application to real data, where the method manages to obtain informative bounds under minimal assumptions. 

\subsection{Simulation}\label{simulationsubsec}
The idea for generating the data in this simulation exercise is to (i) generate paths of stochastic processes $Y_x$ and $X_{z}$ for a given set $\{z_i\}_{i=1,\ldots,m}\in[0,1]$, (ii) combine the two processes to a joint process $[Y,X]^*_{z}$, (iii) randomly sample points $z$ and corresponding points $(y,x)$  induced by the paths of $[Y,X]^*_{z}$ to obtain the data $(Y,X)$ for the given set of $\{z_i\}$. The processes $Y_x$ and $X_{z}$ are Gaussian processes with mean $0.5$ and a squared exponential covariance kernel of the form $k_{SE}(z-z')\coloneqq \sigma^2\exp\left(-\frac{(z-z')^2}{\ell^2}\right)$, where the length parameter $\ell$ is $0.5$ for $Y_x$ and $0.2$ for $X_z$, and the variance parameter $\sigma^2$ is $0.2$ for $Y_x$ and $0.15$ for $X_z$. All paths of these processes are restricted to lie within $[0,1]$. $2500$ paths were sampled, generating $5000$ random data points from it on a relatively coarse grid based on a dyadic approximation of order $3$, i.e.~grid points at $\frac{k}{8}$ for $k=0,\ldots,8$. The values $z$ were drawn uniformly on the unit interval. This data-set has no additionally introduced randomness, to see if the method can actually obtain correct results in practice.

The goal is to estimate $E[Y(0.5)],$ which is equal to $0.5$ by construction. $\hat{F}_{Y,X|Z=z;h_n}$ is estimated by kernel density methods with a cross-validated bandwidth using the $np$-package in $R$. The only assumption on the paths made for estimation is continuity. The hat-functions defined in \eqref{waveleteq} form the basis for different levels of $\kappa$ and $\lambda$. Figure \ref{simul1} depicts the convergence of the bounds in the ``sampling trick'' approach.
\begin{figure}[htb!]
\centering
\includegraphics[width=7.5cm,height=7.5cm]{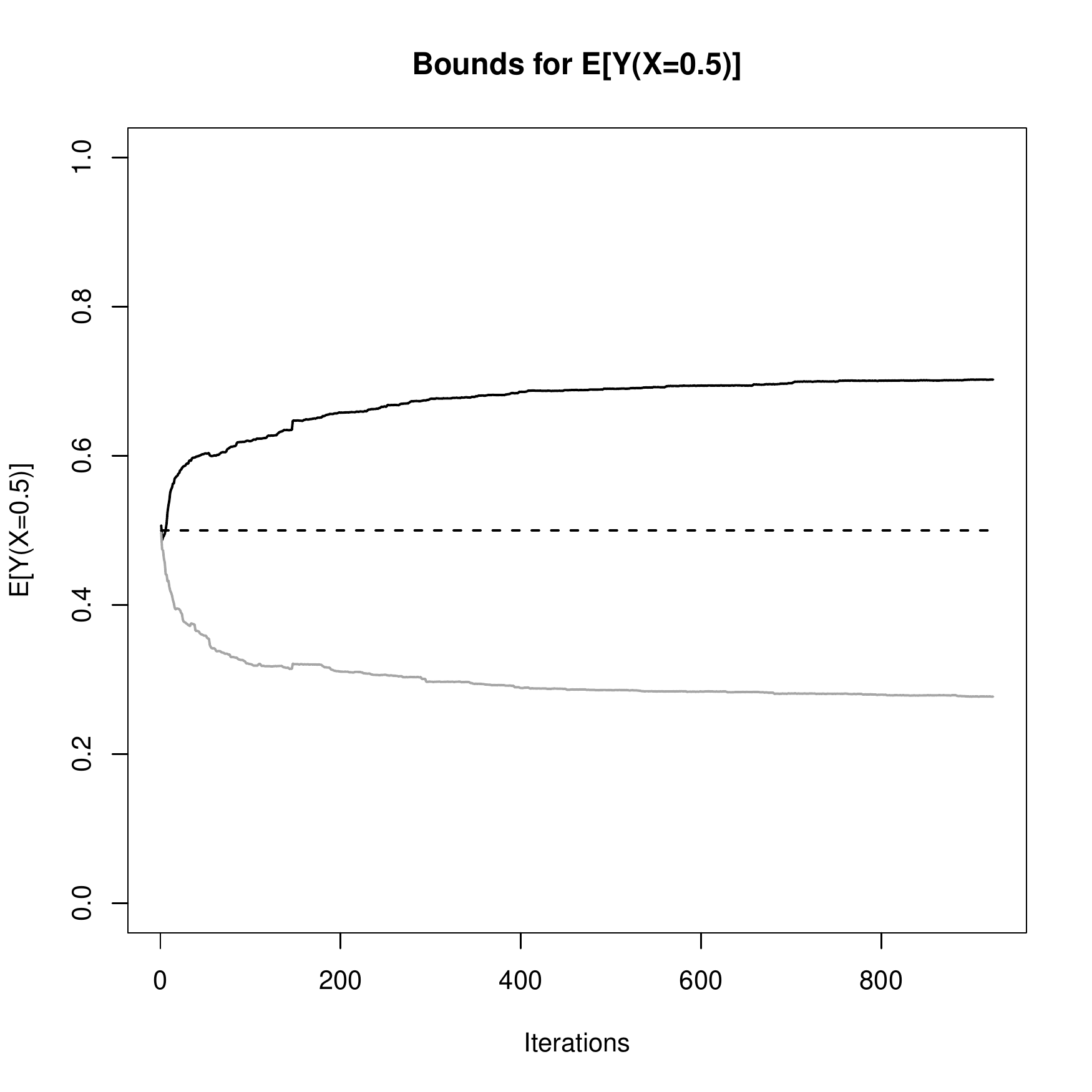}
\includegraphics[width=7.5cm,height=7.5cm]{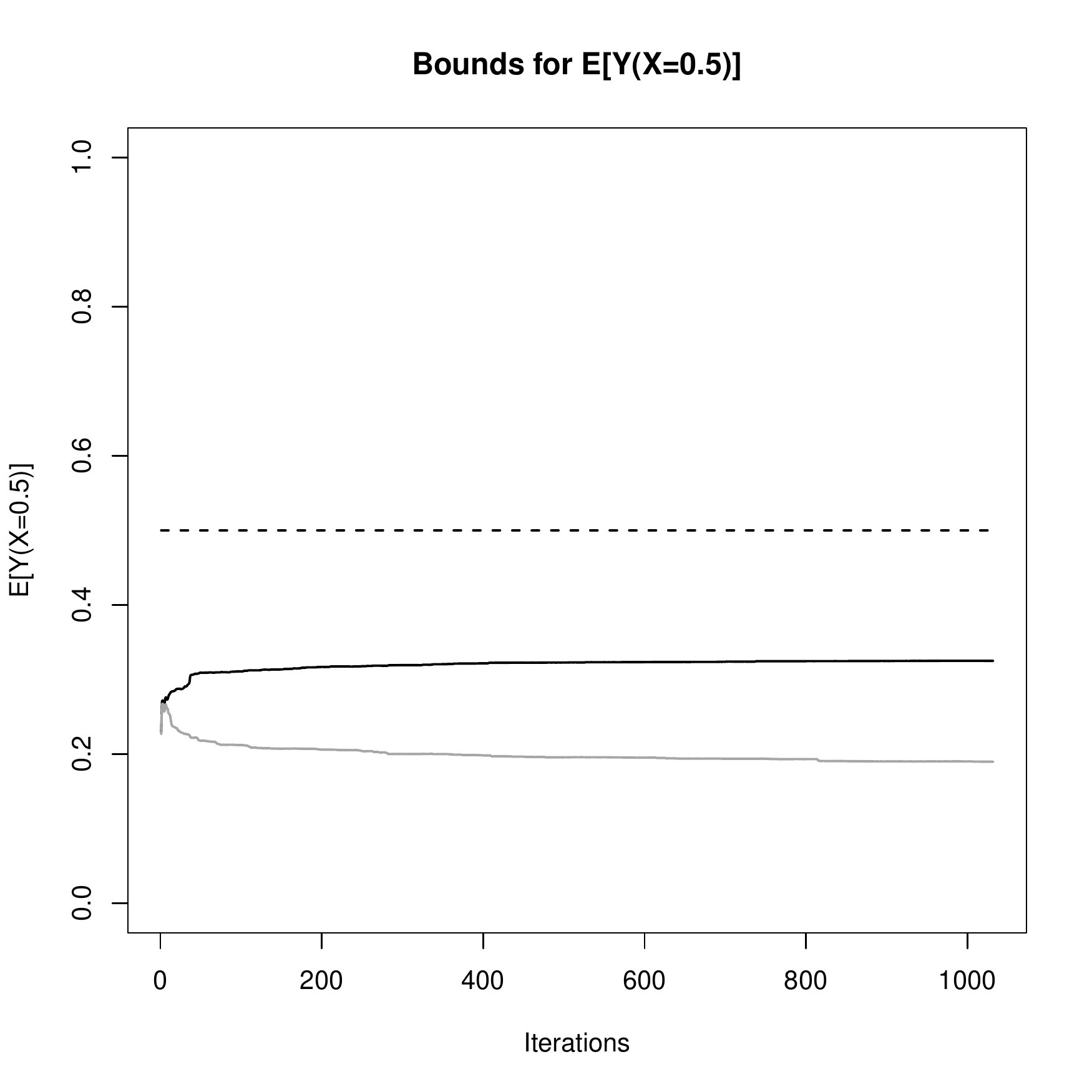}
\caption{Convergence of upper (black) and lower (gray) bounds on $E[Y(X=0.5)] = 0.5$. The penalty term is $\lambda = 5$ for an equidistant approximation of the unit interval by $9$ points. At each iteration the method samples $10$ new continuous paths. In the left panel the paths are constructed by summing over the wavelet basis with $\kappa=1,\ldots,5$. In the right panel the paths are constructed for fixed $\kappa=8$.}\label{simul1}
\end{figure}
The left panel depicts the convergence of the bounds on $E[Y(X=0.5)]$ for $\lambda=5$ and a sieve basis which consists of a sum of the basis functions for $\kappa=1,\ldots,5$. The method seems to converge nicely to bounds which contain the true value and are actually reasonably tight, especially for such a coarse grid and the fact that no functional form restrictions are imposed on the model except continuity. The right panel depicts the same thing, only there the paths $Y_x(w)$ and $X_z(w)$ are constructed for one form of basis functions corresponding to $\kappa=8$. Clearly, the method obtains biased results in this case, as the true value does not lie between the upper- and lower bound. Intuitively, this bias stems from the fact that the paths generated for fixed $\kappa=8$, despite being continuous, are not smooth and very erratic, while the paths generated in the simulation are exceptionally smooth due to the choice of the squared exponential covariance kernel and relatively large length parameters $\ell$. This bias can be understood as a ``nonparametric misspecification bias'' by sampling paths which are very different from the paths in the true data generating process. Fortunately, the left panel shows that already summing over only a few different basis functions removes this bias in this stylized example, which most likely also holds in more realistic settings \citep*{chen2007large}. 

\begin{figure}[htb!]
\centering
\includegraphics[width=7.5cm,height=7.5cm]{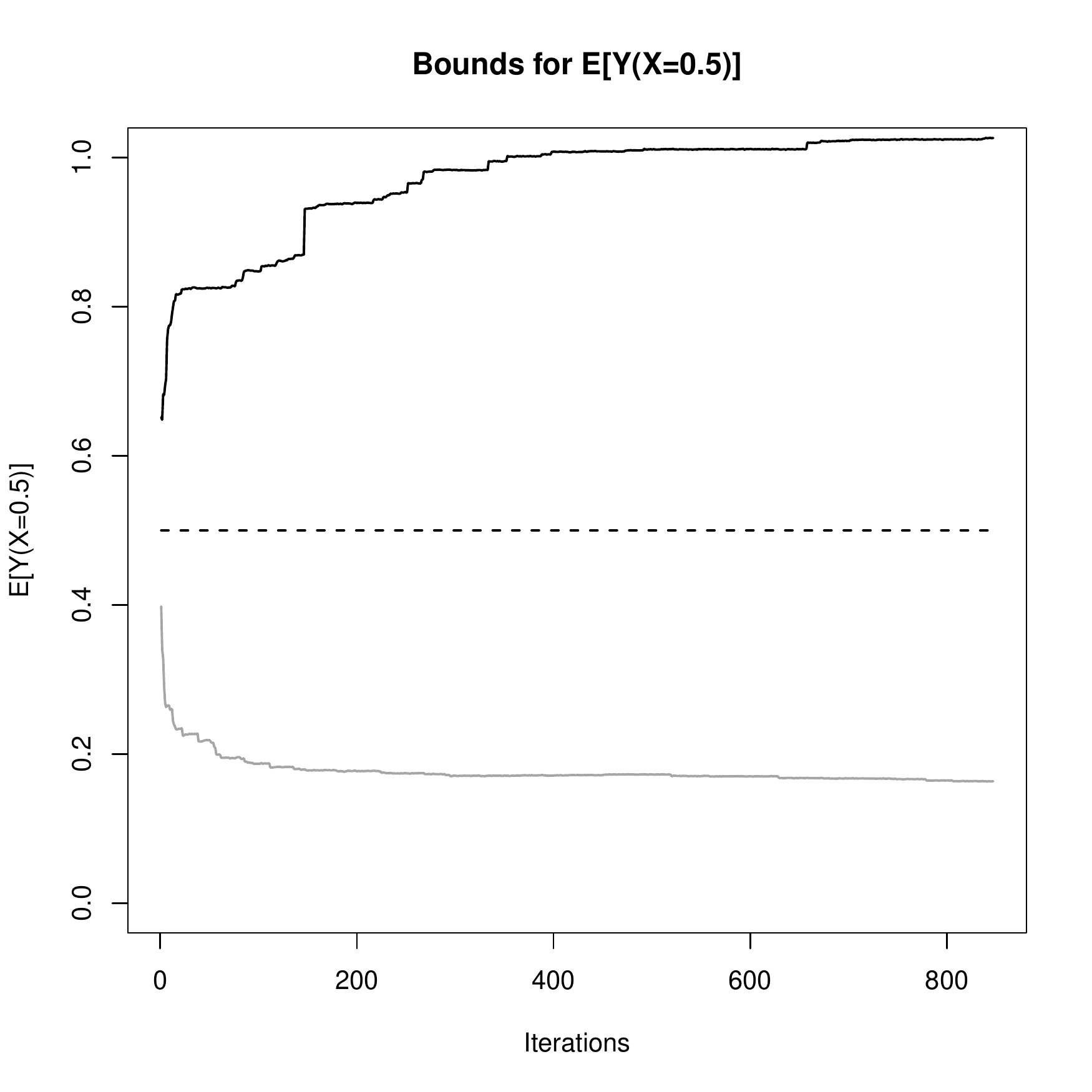}
\includegraphics[width=7.5cm,height=7.5cm]{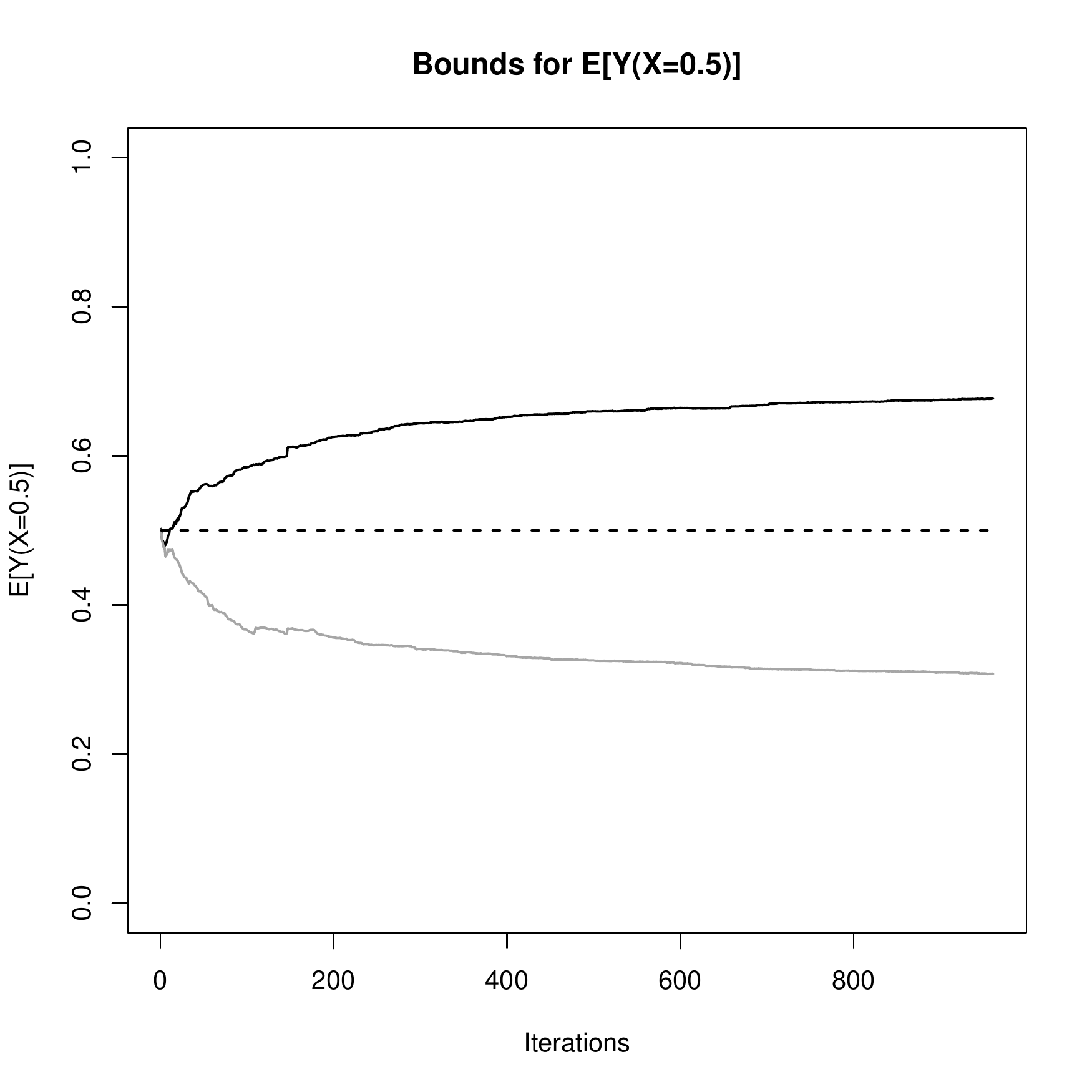}
\caption{Convergence of upper (black) and lower (gray) bounds on $E[Y(X=0.5)] = 0.5$. The penalty terms are $\lambda = 0.1$ (left panel) and $\lambda = 50$ (right panel) for an equidistant approximation of the unit interval by $9$ points. At each iteration the method samples $10$ new continuous paths. The paths are constructed by summing over the wavelet basis with $\kappa=1,\ldots,5$.}\label{simul2}
\end{figure}
Figure \ref{simul2} depicts the sensitivity of the method to choosing different penalty terms $\lambda$ in this setting. In the left panel $\lambda = 0.1$ while in the right $\lambda = 50$. The bounds in the left case are significantly wider than on the right, but both estimators are unbiased. Interestingly, and in contrast to Figure \ref{convergenceplot}, even larger $\lambda$ does not lead to erratic behavior of the paths. This is most likely a result of the smoothness of the simulated paths in combination with a very high signal-to-noise ratio of the simulated data, which implies that there always exists an optimal $\mu$ in the finite dimensional problems \eqref{minimizereqapprox2pen}. In contrast, when the signal-to-noise ratio is low, which is often the case in real-world applications, then the behavior of the solution paths will become too erratic as depicted in Figure \ref{convergenceplot}.

\subsection{Application}\label{applicationsubsec}
As a demonstration of its capabilities, the method estimates bounds on expenditure differences using the $1995/1996$ UK family expenditure survey. This problem is well-suited as it (i) is nonlinear with continuous variables (\citeauthor*{blundell2007semi} \citeyear{blundell2007semi}, \citeauthor*{imbens2009identification} \citeyear{imbens2009identification}), (ii) allows to gauge if the program actually obtains reasonable results, and (iii) provides a setting not directly related to causal inference, showing the scope of the proposed method. Therefore, in the following the focus will be on the outcomes food and leisure. 

Analogous to \citet*{blundell2007semi} and \citet*{imbens2009identification}, the outcome of interest $Y$ will be the share of expenditure on a commodity and $X$ will be the log of total expenditure, scaled to lie in the unit interval. The instrument used in this setting is gross earnings of the head of the household, which assumes that the way the head of the household earns the money is (sufficiently) independent of the household's expenditure allocation; this instrument is used in both \citet*{blundell2007semi} and \citet*{imbens2009identification}. All three variables are inherently continuous which makes this problem a nice setting for demonstrating the practical implementation of the method.

The sample is restricted to the subset of married and cohabiting couples where the head of the household is aged between 20 and 55, and couples with 3 or more children are excluded. Also excluded are households where the head of the household is unemployed in order to have the instrument available for each observation. The final sample comprises 1650 observations.\footnote{The data used in both \citet*{blundell2007semi} and \citet*{imbens2009identification} is based on the $1994/1995$ UK family expenditure survey. Under the same restriction they end up with $1655$ observations.}
\begin{table}[ht!]
\centering
\begin{tabular}{c|c|c}
&Leisure&Food\\\hline\hline
mean&$0.143$ & $0.182$\\
standard deviation& $0.110$&$0.0723$\\
skewness& $1.83$&$0.686$\\\hline
min&$0.00122$&$0.00946$\\
lower quartile&$0.0689$& $0.131$\\
median& $0.111$&$0.176$\\
upper quartile&$0.183$& $0.225$ \\
max&$0.831$& $0.657$\\\hline\hline
\end{tabular}
\caption{Summary statistics for the outcome distributions}\label{summarytable}
\end{table}
Table \ref{summarytable} gives a summary of the outcome distributions, showing that the relative expenditure on leisure is much more skewed towards zero but with a higher variance than the distribution for food.

The only shape restriction on the instrumental variable model is continuity, i.e.~$h$ and $g$ are continuous functions in $X$ and $Z$, respectively. This is a natural assumptions since Engel curves are usually believed to be continuous. It implies that almost all of the paths of $Y_x$ and $X_z$ lie in $C([0,1])$, the space of continuous functions. No other assumptions are upheld. The most general current approaches either require continuity and strict monotonicity of $g$ \citep*{imbens2009identification} or of $h$ \citep*{blundell2007semi} in the unobservable $W$. In contrast, the proposed method does not require any monotonicity assumptions and hence intuitively gives an indication of how much information is available in the data to solve this problem. Surprisingly, there seems to be a substantial amount of information, as the obtained bounds indicate that food is a necessity- and leisure is a luxury good without any assumptions on the model besides continuity. Furthermore, when introducing monotonicity assumptions \emph{in the observable variables}, the bounds become significantly tighter, showing the identificatory strength of these assumptions in this setting.

Figure \ref{leisure_low} depicts the solution paths for obtaining bounds on the counterfactual difference $F_{Y(X=0.75)}(0.15) - F_{Y(X=0.25)}(0.15)$ for a reasonably fine approximation of the unit interval into $17$ equidistant points (which corresponds to a dyadic approximation of order $4$). 
\begin{figure}[htb!]
\centering
\includegraphics[width=7.5cm,height=7.5cm]{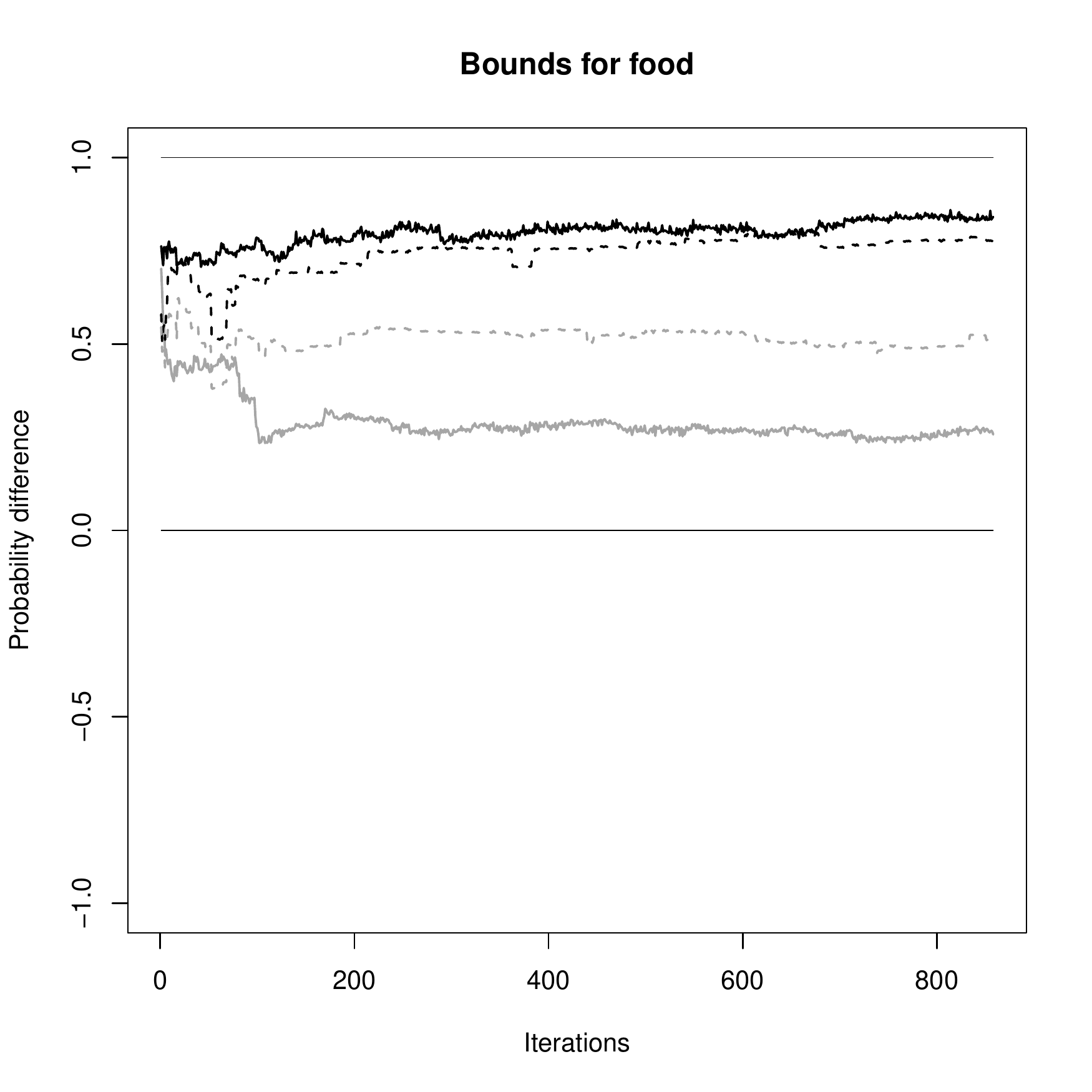}
\includegraphics[width=7.5cm,height=7.5cm]{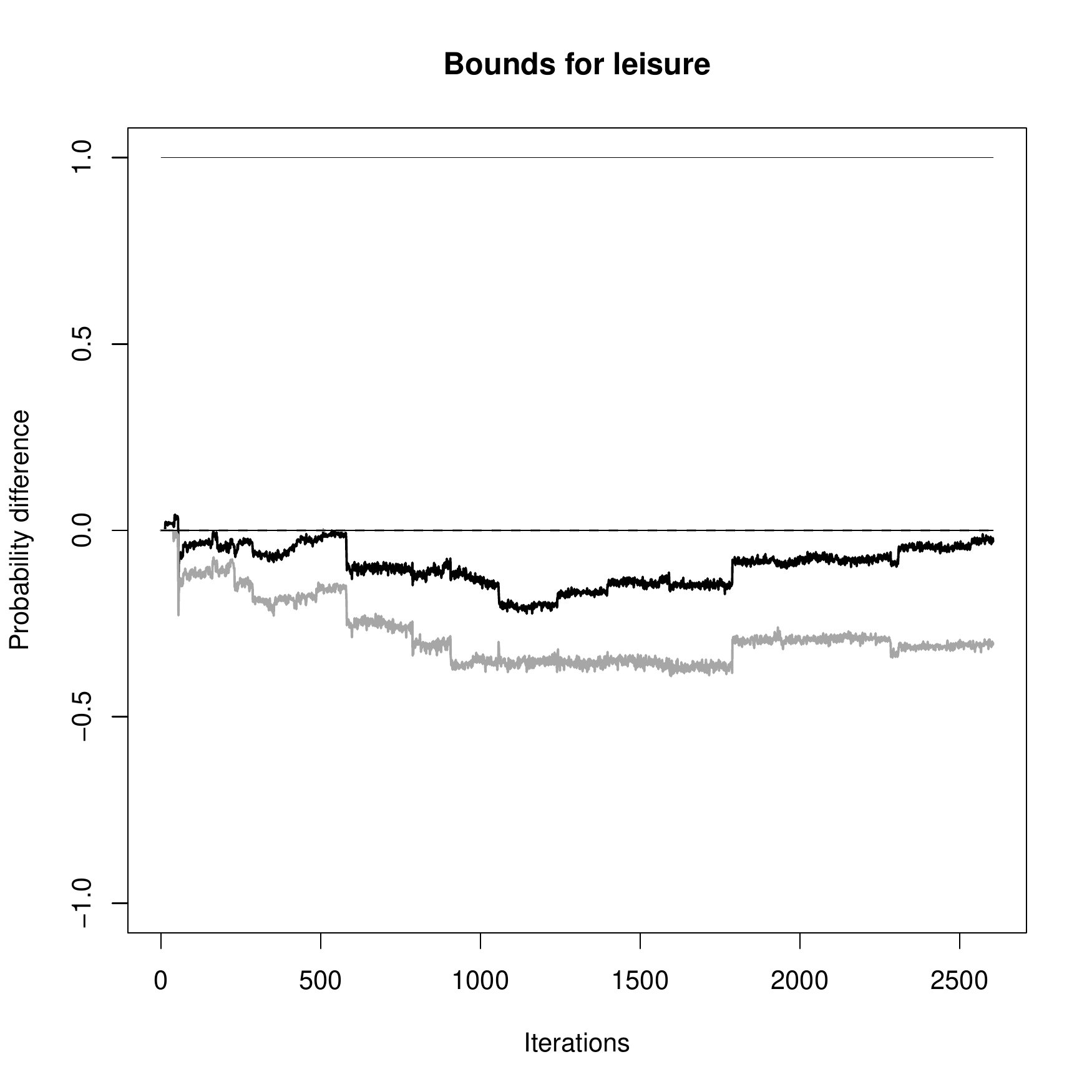}
\caption{Convergence of upper (black) and lower (gray) bounds on $F_{Y(X=0.75)}(0.15) - F_{Y(X=0.25)}(0.15)$ for $Y$ being the relative spending on food (left panel) and leisure (right panel). The solid lines indicate the bounds under no additional assumptions on the model, the dashed lines indicate the bounds under the additional assumption that $Y$ is decreasing (food) or increasing (leisure) in overall expenditure $X$ and that $X$ is increasing in income of the head of the household $Z$. The penalty term is $\lambda=1$ for an approximation of the unit interval by $17$ points. At each iteration the method samples $25$ new paths, of which $6$ satisfy the respective monotonicity requirement.}\label{leisure_low}
\end{figure}
  
Here the penalization parameter $\lambda = 1$, which seems to provide reasonable convergence to the solution, in particular for the food data. Remarkably, the estimated bounds in this setting are qualitatively informative for the problem. The left panel depicts the households' expenditures on food and the right depicts their expenditures on leisure. The solid lines are the upper- and the lower bound for a model without further assumptions, while the dashed lines are the upper- and lower bounds for a model with the additional assumption that $Y$ is increasing for leisure and decreasing for food in overall expenditure $X$ and that $X$ is increasing in income of the head of the household $Z$. 

Consider the left panel first, which depicts the expenditures on food. Here, the general bounds seem to converge and the average values of the bounds over the last $200$ iterations are $0.84$ and $0.26$ for the general upper- and lower bound, and $0.78$ and $0.50$ for the corresponding upper- and lower bounds for the monotone model. All four bounds are positive, which indicates that families that spend a lot in general ($X=0.75$) and spend up to $15\%$ on food ($Y\in[0,0.15]$) would spend much more on food relatively to overall expenditure if they spent much less overall ($X=0.25$). Put differently, families are much more likely to lie in higher quantiles for expenditure on food if they lie in the lower quartile in overall spending than families that lie in the upper quartile in overall spending, which is the defining characteristic of a necessity good at this given level $y^*=0.15$. 

It is rather striking that even the model without monotonicity assumptions produces bounds which reflect this fact via a positive lower bound. In this regard, note that the monotonicity assumptions do not only tighten the bounds, but also shift up the lower bound, indicating that monotonicity has a stronger identificatory content in this setting. In particular, they imply that more families (between $50\%$ and $78\%$) in the upper quartile of overall spending ($X=0.75$) spend only up to $15\%$ on food compared to families in the lower quartile of overall spending ($X=0.25$). As mentioned, without monotonicity assumptions, these differences can be as high as $84\%$ and as low as $26\%$, but all positive.

The results for expenditure on leisure for this given scenario are similar, but more erratic. For a clear indication of a luxury good at the given levels, one would expect both bounds to be negative. In fact, these would imply that families in the upper quartile on overall spending (i.e.~$X=0.75$) who spend up to $15\%$ of their overall expenditure on leisure ($Y\in[0,0.15]$) are very likely to spend even less on leisure, relatively, if they had a negative shock to overall spending ($X=0.25$). Put differently, families should be more likely to spend only up to $15\%$ of overall expenditure on leisure ($Y\in[0,0.15]$) if they lie in the lower quartile in overall spending ($X=0.25$) than families that lie in the upper quartile in overall spending ($X=0.75$). The obtained results do reflect this circumstance at this level. In particular, the averages of the last $500$ iterations of the general bounds are $-0.032$ and $-0.31$, implying that typically more families (up to $31\%$) in the lower quartile of overall spending ($X=0.25$) spend only up to $15\%$ of their overall on leisure compared to families in the upper quartile of overall spending. 

The overall convergence for the leisure data is more erratic compared to food, which is illuminating for the purpose of this article. In particular, the main culprit for the poor performance is the fact that the data is highly skewed towards zero in the leisure case, while the grid placed uniformly over the unit interval is too coarse to ``measure'' the behavior of the data around zero. Note that the method did not manage to sample enough monotone paths which correspond to the respective events of interest for this reason, so that no bounds for monotonicity exist in this setting. In order to circumvent this, one can simply put a non-uniform grid on $[0,1]$ which has more points close to zero and fewer points further away.

One can perform this exercise for different levels of $y^*$ in order to gauge the behavior of the bounds at different quantiles. Figure \ref{leisure_35} depicts the same exercise for $y=0.25$, i.e.~looking at families that spend up to a quarter of the overall expenditures on food and leisure. Here, the method managed to sample enough relevant monotonic paths in the leisure setting despite the coarser grid around zero. These bounds converge to some value in the more general bounds. One important fact to point out is the behavior of the solution path in the leisure case. At a certain point around the $1000$ iterations mark a path is sampled which forces the solution path to jump up dramatically, i.e.~a path which contains a lot of information for the given problem. However, the good news here is that over time the paths approach the previous level as more and more paths are drawn and add to the information from the one path. This is an important property to note, as it shows that the solution paths need not monotonically increase in the number of samples, and that even if paths are sampled which contain a lot of information for the respective program, will the method still converge in the long run. In particular, it is the choice of the penalty term which forces to solution path to approach zero again after the jump.
\begin{figure}[htb!]
\centering
\includegraphics[width=7.5cm,height=7.5cm]{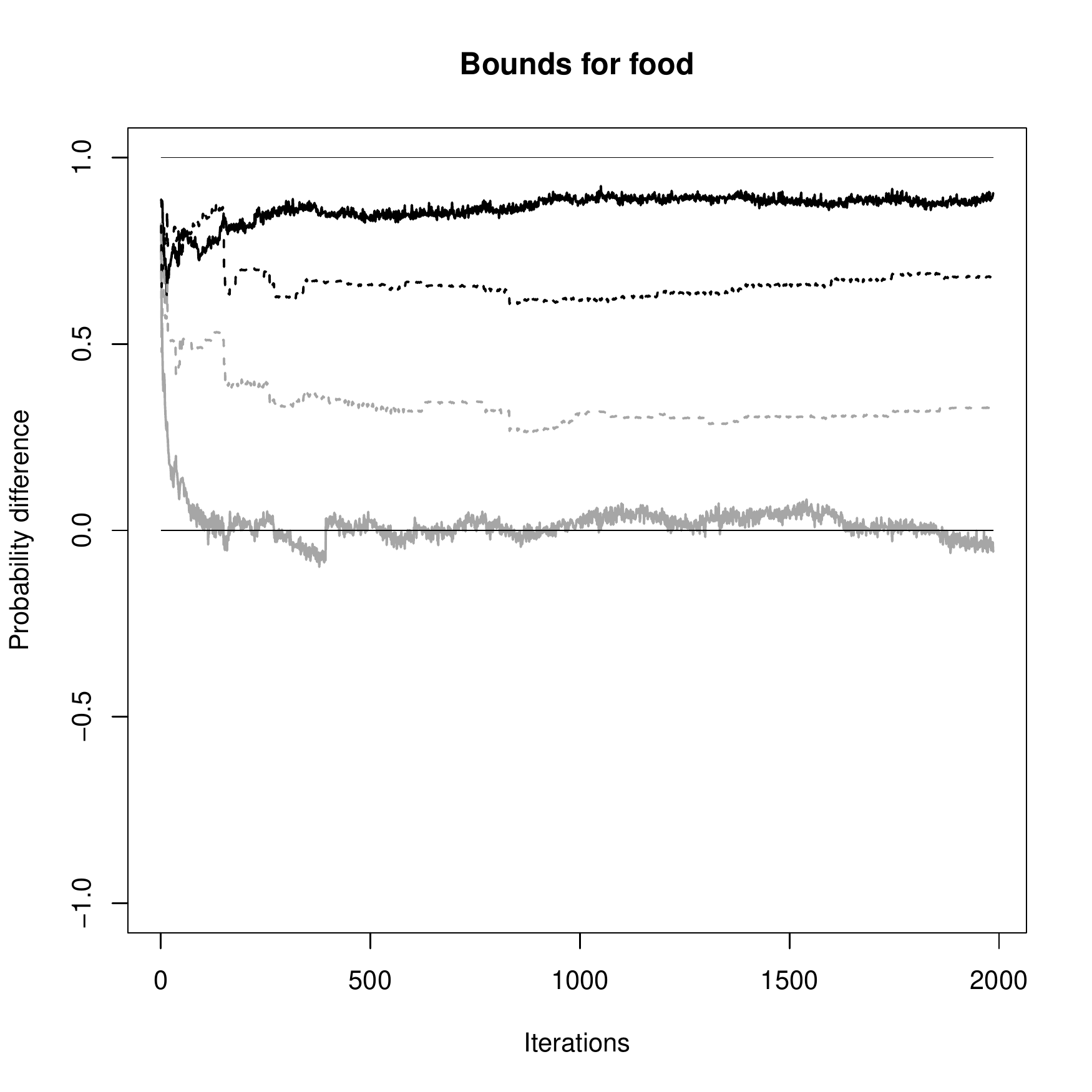}
\includegraphics[width=7.5cm,height=7.5cm]{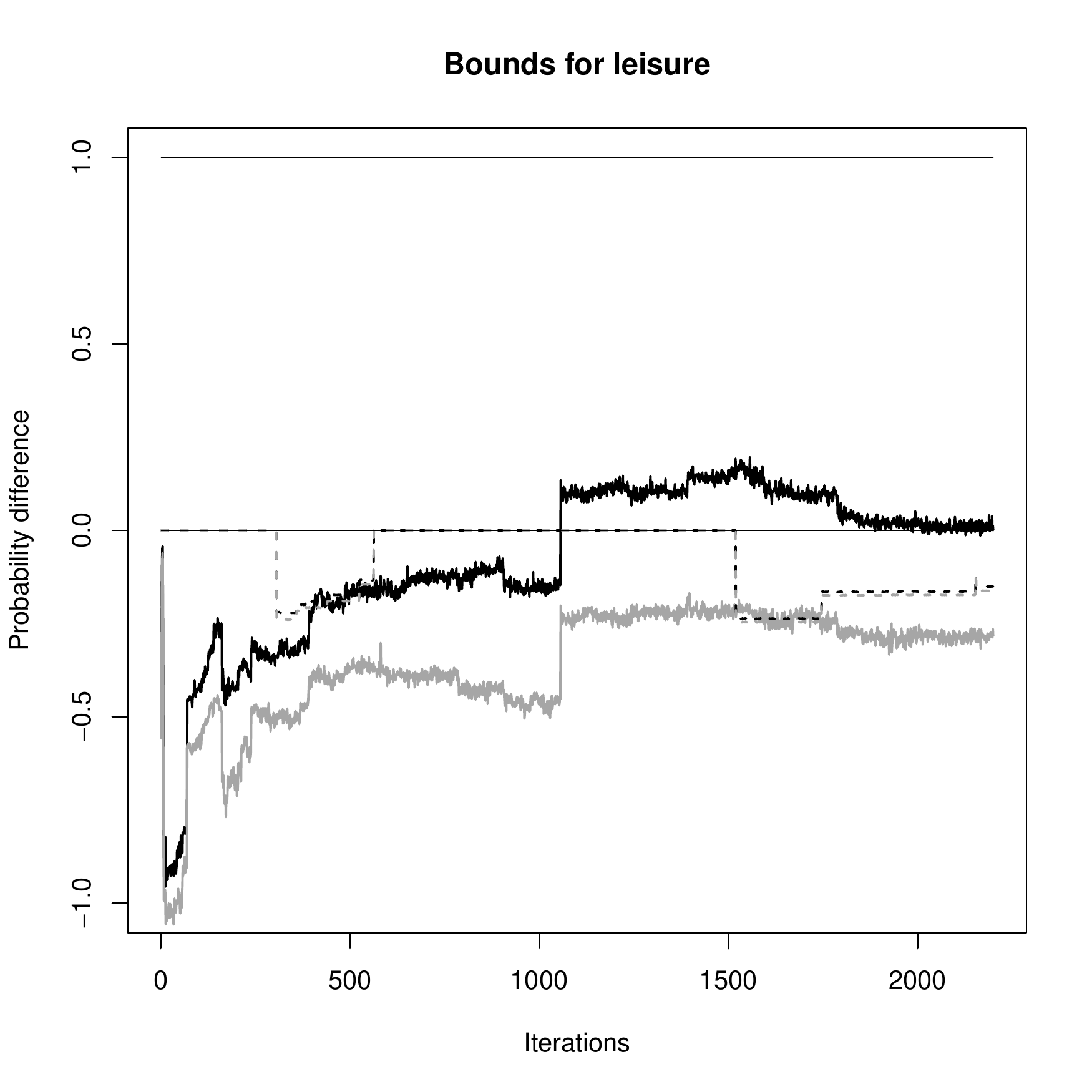}
\caption{Convergence of upper (black) and lower (gray) bounds on $F_{Y(X=0.75)}(0.25) - F_{Y(X=0.25)}(0.25)$ for $Y$ being the relative spending on food (left panel) and leisure (right panel). The solid lines indicate the bounds under no additional assumptions on the model, the dashed lines indicate the bounds under the additional assumption that $Y$ is decreasing (food) or increasing (leisure) in overall expenditure $X$ and that overall expenditure $X$ is increasing in income of the head of the household $Z$. The penalty term is $\lambda = 1$ for an approximation of the unit interval by $17$ points. At each iteration the method samples $25$ new paths, of which $6$ satisfy the respective monotonicity requirement.}\label{leisure_35}
\end{figure}

The results are qualitatively similar to the ones for $y^*=0.15$, which is not surprising as one is now comparing more families, i.e.~all families that spend up to a quarter on food/leisure. The more families one compares, the less pronounced the effects become. The following table provides an overview of the results at the levels $y^* = 0.15, 0.25, 0.5, 0.75$, which shows exactly this. At $y=0.75$ one is basically comparing all existing families, so that one does not obtain any effects, another sanity check for the method.
\begin{table}[ht!]
\centering
\begin{tabular}{c|c|c|c|c||c|c|c|c}
\multicolumn{5}{c}{Food}&\multicolumn{4}{c}{Leisure}\\\hline
$y^*$&Lower &Lower& Upper& Upper & Lower & Lower& Upper& Upper\\
&&monotone&monotone&&&monotone&monotone&\\\hline
$0.15$&$0.26$& $0.50$ & $0.78$ & $0.84$ & $-0.31$ & $0$& $0$&$-0.032$\\
$0.25$& $0.008$&$0.31$&$0.67$&$0.88$&$-0.28$&$-0.18$& $-0.17$ & $0.028$\\
$0.5$&$-0.33$& $0.0063$ &$0.072$&$0.34$&$-0.34$&$-0.22$&$-0.19$&$0.09$\\
$0.75$&$-0.021$&$0$&$0$&$0.0075$&$-0.072$&$-0.046$&$0$&$0.0029$\\\hline\hline
\end{tabular}
\caption{Upper- and lower bounds for $F_{Y(X=0.75)}(y^*)-F_{Y(X=0.25)}(y^*)$ for different values $y^*$.}\label{quantiletable}
\end{table}

Overall, these estimation results are remarkably informative from a qualitative perspective. Recall that the instrumental variable model allows for general unobserved heterogeneity, in particular measurement error in the treatment variable $X$, which indicates that the ratio of information to noise in the data for answering these questions is rather high. These qualitative results not only corroborate the theoretical predictions for expenditure, but also the previous results obtained in \citet*{blundell2007semi}, \citet*{imbens2009identification}, and \citet*{song2018nonseparable}. During their estimation process \citet*{imbens2009identification} and \citet*{song2018nonseparable} assume a univariate and strictly monotonic production function $g(z,W)$ between $X$ and $W$ for all $z$ and use a control variable approach to estimate the production function $h$; \citet*{blundell2007semi} estimate Engel curves semi-nonparametrically, imposing monotonicity in the second stage, and obtaining similar results. \citet*{de2016nonparametric} work with an additively separable first stage and allow for the outcome $Y$ to be measured with error. This is more general than what the proposed method can handle, which can only encompass measurement error in the dependent variable, but not the outcome. Nonetheless, their results are similar to the ones obtained here. In this sense these qualitative results are a ``robustness check'' for other non- or semiparametric approaches. 

Moreover, this method makes it possible to gauge the identificatory content of monotonicity assumptions in the current model. In all cases is this content rather high. Imposing monotonicity \emph{between the observables} makes the results much more clear-cut and in turn leads to rather strong implications in the cases considered. In this setting, monotonicity is a plausible assumption based on economic theory, but it is important to be aware of the strength of this assumption in other settings.

\section{Conclusion}\label{conclusion}
This article introduces a novel method for partially identifying and causal effects in instrumental variable models with general heterogeneity. It is the first practically applicable method for the most general models with continuous endogenous variables. The idea is to write the respective instrumental variable model as a system of counterfactual stochastic processes and to solve for an optimal probability measure on the paths of these processes subject to the constraint that the law of the joint processes induced by this probability measure replicates the observable distribution. The resulting optimization problem takes the form of an infinite dimensional linear program on path spaces. 

The main contribution of this article is to introduce the ``sampling-of-paths'' approach to solve these types of infinite dimensional programs. The underlying idea is to reduce the infinite dimensional program to a semi-infinite program \citep*{anderson1987linear} by sampling a subset of the paths over which the program optimizes. Then an approximation of the (finite dimensional) state-space of the random variables leads to a finite dimensional program which can be solved efficiently. 
The main idea for reducing the infinite dimensional program to a semi-infinite dimensional one is to explicitly introduce randomness by sampling paths. This, in conjunction with large deviation results (\citet*{vapnik1998statistical}, \citet*{wellner2013weak}) allows to obtain probabilistic approximation guarantees of the semi-infinite program to the infinite dimensional program. In particular, these guarantees imply a lower bound on the number of paths required for achieving a good approximation with high probability. 

A remaining challenge is to obtain an efficient data-driven method for choosing an appropriate penalty term $\lambda$ for the practical optimization routine. This is a similar challenge to finding good penalty terms in high-dimensional regularized regression estimators, but more general, as the setting here is infinite dimensional in a counterfactual path space. Some heuristic guidelines can be given: one should choose the largest $\lambda$ such that the ``solution paths'' of the sampling method converge to a fixed value after a ``burn-in'' period. If the solution path is ``too erratic'', then one should lower the value of $\lambda$. Formally establishing what ``convergence'', ``burn-in period'', and ``too erratic'' mean would not only solve this issue, but would open up potentially novel approaches for data-driven validation approaches in counterfactual settings. In particular, an analogue to the data-driven method for $\ell1$-regularization in high-dimensional regression models as put forward in \citet*{belloni2011l1} could be valuable.

The current practical implementation of the program works for univariate variables. Moreover, the only currently implemented additional nonparametric restriction which can be placed on the model is monotonicity. The program can straightforwardly be extended to higher dimensional settings, but runs into the curse of dimensionality as the stochastic processes become high-dimensional random fields. One standard way to circumvent the curse of dimensionality is to introduce sparsity- and factor assumptions on the stochastic processes in a higher-dimensional setting. Furthermore, it is also imperative to allow for a wide variety of additional (non-) parametric assumptions in the model, like convexity, bounds, reflection processes, Slutsky-type conditions, first-passage times, martingale properties, etc. Furthermore, the current ``sampling-and discarding'' approaches are not efficient. From a computational perspective, it would be interesting to obtain more efficient resampling methods. One promising approach is to use ideas from sequential Monte-Carlo approaches (see e.g.~\citeauthor*{schweizer2012thesis} \citeyear{schweizer2012thesis}), which need to be extended to the infinite dimensional path spaces considered here. Another idea is to directly use diffusion- or Levy-processes in the generation of the paths and use the Karhunen-Lo\`eve transform to generate a basis for the paths.

\bibliography{main_bib}

@inproceedings{balke1994counterfactual,
  title={Counterfactual probabilities: {C}omputational methods, bounds and applications},
  author={Balke, Alexander and Pearl, Judea},
  booktitle={Proceedings of the Tenth international conference on Uncertainty in artificial intelligence},
  pages={46--54},
  year={1994},
  organization={Morgan Kaufmann Publishers Inc.}
}

@article{balke1997bounds,
  title={Bounds on treatment effects from studies with imperfect compliance},
  author={Balke, Alexander and Pearl, Judea},
  journal={Journal of the American Statistical Association},
  volume={92},
  number={439},
  pages={1171--1176},
  year={1997},
  publisher={Taylor \& Francis}
}

@article{cheng2006bounds,
  title={Bounds on causal effects in three-arm trials with non-compliance},
  author={Cheng, Jing and Small, Dylan S},
  journal={Journal of the Royal Statistical Society: Series B (Statistical Methodology)},
  volume={68},
  number={5},
  pages={815--836},
  year={2006},
  publisher={Wiley Online Library}
}

@inproceedings{hess1982kuratowski,
	title = {A {K}uratowski approach to {W}iener measure},
	author = {Hess, Hans-Ulrich},
	booktitle = {{M}easure {T}heory {O}berwolfach 1981},
	pages = {336--346},
	year = {1982}
}

@article{kuratowski1934generalisation,
  title={Sur une g{\'e}n{\'e}ralisation de la notion d'hom{\'e}omorphie},
  author={Kuratowski, K},
  journal={Fundamenta Mathematicae},
  volume={22},
  pages={206--220},
  year={1934}
}

@book{bauer1996probability,
  title={Probability {T}heory},
  author={Bauer, Heinz},
  publisher={De Gruyter studies in {M}athematics},
  year={1996}
}

@book{bogachev2007measure2,
  title={{\GG{20020603}}Measure theory},
  author={Bogachev, Vladimir I},
  volume={2},
  year={2007},
  publisher={Springer Science \& Business Media}
}

@book{anderson1987linear,
  title={Linear programming in infinite dimensional spaces: {T}heory and applications},
  author={Anderson, EJ and Nash, P},
  publisher={Wiley},
  year={1987}
}

@book{billingsley1999convergence,
  title={Convergence of probability measures},
  author={Billingsley, Patrick},
  year={1999},
  publisher={John Wiley \& Sons}
}

@article{wu2001solving,
  title={Solving general capacity problem by relaxed cutting plane approach},
  author={Wu, Soon-Yi and Fang, Shu-Cherng and Lin, Chih-Jen},
  journal={Annals of Operations Research},
  volume={103},
  number={1},
  pages={193--211},
  year={2001},
  publisher={Springer}
}

@article{pearl1995causal,
  title={Causal diagrams for empirical research},
  author={Pearl, Judea},
  journal={Biometrika},
  volume={82},
  number={4},
  pages={669--688},
  year={1995},
  publisher={Oxford University Press}
}

@article{rubin1974estimating,
  title={Estimating causal effects of treatments in randomized and nonrandomized studies.},
  author={Rubin, Donald B},
  journal={Journal of Educational Psychology},
  volume={66},
  number={5},
  pages={688},
  year={1974},
  publisher={American Psychological Association}
}

@article{torgovitsky2015identification,
  title={Identification of nonseparable models using instruments with small support},
  author={Torgovitsky, Alexander},
  journal={Econometrica},
  volume={83},
  number={3},
  pages={1185--1197},
  year={2015},
  publisher={Wiley Online Library}
}

@article{angrist1996identification,
  title={Identification of causal effects using instrumental variables},
  author={Angrist, Joshua D and Imbens, Guido W and Rubin, Donald B},
  journal={Journal of the American Statistical Association},
  volume={91},
  number={434},
  pages={444--455},
  year={1996},
  publisher={Taylor \& Francis}
}

@article{d2015identification,
  title={Identification of nonseparable triangular models with discrete instruments},
  author={{d}'Haultf{\oe}uille, Xavier and F{\'e}vrier, Philippe},
  journal={Econometrica},
  volume={83},
  number={3},
  pages={1199--1210},
  year={2015},
  publisher={Wiley Online Library}
}

@article{imbens2009identification,
  title={Identification and estimation of triangular simultaneous equations models without additivity},
  author={Imbens, Guido W and Newey, Whitney K},
  journal={Econometrica},
  volume={77},
  number={5},
  pages={1481--1512},
  year={2009},
  publisher={Wiley Online Library}
}

@book{manski2003partial,
  title={Partial identification of probability distributions},
  author={Manski, Charles F},
  year={2003},
  publisher={Springer Science \& Business Media}
}

@article{chernozhukov2005iv,
  title={An {IV} model of quantile treatment effects},
  author={Chernozhukov, Victor and Hansen, Christian},
  journal={Econometrica},
  volume={73},
  number={1},
  pages={245--261},
  year={2005},
  publisher={Wiley Online Library}
}

@article{chesher2017generalized,
  title={Generalized instrumental variable models},
  author={Chesher, Andrew and Rosen, Adam M},
  journal={Econometrica},
  volume={85},
  number={3},
  pages={959--989},
  year={2017},
  publisher={Wiley Online Library}
}

@book{wellner2013weak,
  title={Weak convergence and empirical processes: with applications to statistics},
  author={Aad {van der Vaart} and Wellner, Jon},
  year={2013},
  publisher={Springer Science \& Business Media}
}

@article{shapiro1991asymptotic,
  title={Asymptotic analysis of stochastic programs},
  author={Shapiro, Alexander},
  journal={Annals of Operations Research},
  volume={30},
  number={1},
  pages={169--186},
  year={1991},
  publisher={Springer}
}

@book{bonnans2013perturbation,
  title={Perturbation analysis of optimization problems},
  author={Bonnans, J Fr{\'e}d{\'e}ric and Shapiro, Alexander},
  year={2013},
  publisher={Springer Science \& Business Media}
}

@article{dumbgen1993nondifferentiable,
  title={On nondifferentiable functions and the bootstrap},
  author={D{\"u}mbgen, Lutz},
  journal={Probability Theory and Related Fields},
  volume={95},
  number={1},
  pages={125--140},
  year={1993},
  publisher={Springer}
}

@article{fang2018inference,
  title={Inference on directionally differentiable functions},
  author={Fang, Zheng and Santos, Andres},
  journal={The Review of Economic Studies},
  volume={86},
  number={1},
  pages={377--412},
  year={2018},
  publisher={Oxford University Press}
}

@article{gine2008uniform,
  title={Uniform central limit theorems for kernel density estimators},
  author={Gin{\'e}, Evarist and Nickl, Richard},
  journal={Probability Theory and Related Fields},
  volume={141},
  number={3-4},
  pages={333--387},
  year={2008},
  publisher={Springer}
}

@article{imbens2004confidence,
  title={Confidence intervals for partially identified parameters},
  author={Imbens, Guido W and Manski, Charles F},
  journal={Econometrica},
  volume={72},
  number={6},
  pages={1845--1857},
  year={2004},
  publisher={Wiley Online Library}
}

@book{folland2013real,
  title={Real analysis: modern techniques and their applications},
  author={Folland, Gerald B},
  year={2013},
  publisher={John Wiley \& Sons}
}

@Article{hayfield2008nonparametric,
    title = {Nonparametric Econometrics: The np Package},
    author = {Tristen Hayfield and Jeffrey S. Racine},
    journal = {Journal of Statistical Software},
    year = {2008},
    volume = {27},
    number = {5}
  }

@article{de2004constraint,
  title={On constraint sampling in the linear programming approach to approximate dynamic programming},
  author={{Pucci de Farias}, Daniela and Van Roy, Benjamin},
  journal={Mathematics of {O}perations {R}esearch},
  volume={29},
  number={3},
  pages={462--478},
  year={2004},
  publisher={INFORMS}
}

@article{boyd2011distributed,
  title={Distributed optimization and statistical learning via the alternating direction method of multipliers},
  author={Boyd, Stephen and Parikh, Neal and Chu, Eric and Peleato, Borja and Eckstein, Jonathan},
  journal={Foundations and Trends{\textregistered} in Machine learning},
  volume={3},
  number={1},
  pages={1--122},
  year={2011},
  publisher={Now Publishers, Inc.}
}

@article{parikh2014proximal,
  title={Proximal algorithms},
  author={Parikh, Neal and Boyd, Stephen},
  journal={Foundations and Trends{\textregistered} in Optimization},
  volume={1},
  number={3},
  pages={127--239},
  year={2014},
  publisher={Now Publishers, Inc.}
}

@book{vapnik1998statistical,
  title={Statistical learning theory. 1998},
  author={Vapnik, Vladimir},
  year={1998},
  publisher={Wiley, New York}
}

@book{anthony1997computational,
  title={Computational learning theory},
  author={Anthony, M and Biggs, N},
  volume={30},
  year={1997},
  publisher={Cambridge University Press}
}

@article{blundell2007semi,
  title={Semi-Nonparametric {IV} Estimation of Shape-Invariant {E}ngel Curves},
  author={Blundell, Richard and Chen, Xiaohong and Kristensen, Dennis},
  journal={Econometrica},
  volume={75},
  number={6},
  pages={1613--1669},
  year={2007},
  publisher={Wiley Online Library}
}

@techreport{song2018nonseparable,
 	title = {Nonseparable Triangular Models with Errors in Endogenous Variables},
	author = {Suyong Song},
	year = {2018},
	institution = {University of Iowa}
}

@article{stoye2009more,
  title={More on confidence intervals for partially identified parameters},
  author={Stoye, J{\"o}rg},
  journal={Econometrica},
  volume={77},
  number={4},
  pages={1299--1315},
  year={2009},
  publisher={Wiley Online Library}
}

@article{beresteanu2012partial,
  title={Partial identification using random set theory},
  author={Beresteanu, Arie and Molchanov, Ilya and Molinari, Francesca},
  journal={Journal of Econometrics},
  volume={166},
  number={1},
  pages={17--32},
  year={2012},
  publisher={Elsevier}
}

@article{galichon2011set,
  title={Set identification in models with multiple equilibria},
  author={Galichon, Alfred and Henry, Marc},
  journal={The Review of Economic Studies},
  volume={78},
  number={4},
  pages={1264--1298},
  year={2011},
  publisher={Oxford University Press}
}

@article{mogstad2018using,
  title={Using instrumental variables for inference about policy relevant treatment effects},
  author={Mogstad, Magne and Santos, Andres and Torgovitsky, Alexander},
  year={2018},
  journal={Econometrica, forthcoming},
  publisher={Oxford University Press}
}

@article{kamat2017identification,
  title={Identification with Latent Choice Sets: The Case of the Head Start Impact Study},
  author={Kamat, Vishal},
  journal={arXiv:1711.02048},
  year={2017}
}

@article{torgovitsky2016nonparametric,
  title={Nonparametric inference on state dependence with applications to employment dynamics},
  author={Torgovitsky, Alexander},
  year={2016},
  note={University of Chicago working paper}
}

@article{kitamura2018nonparametric,
  title={Nonparametric analysis of random utility models},
  author={Kitamura, Yuichi and Stoye, J{\"o}rg},
  journal={Econometrica, forthcoming},
  year={2018}
}

@article{laffers2015bounding,
  title={Bounding average treatment effects using linear programming},
  author={Laff{\'e}rs, Luk{\'a}{\v{s}}},
  journal={Empirical Economics},
  pages={1--41},
  year={2015},
  publisher={Springer}
}

@article{demuynck2015bounding,
  title={Bounding average treatment effects: A linear programming approach},
  author={Demuynck, Thomas},
  journal={Economics Letters},
  volume={137},
  pages={75--77},
  year={2015},
  publisher={Elsevier}
}

@article{chiburis2010semiparametric,
  title={Semiparametric bounds on treatment effects},
  author={Chiburis, Richard C},
  journal={Journal of Econometrics},
  volume={159},
  number={2},
  pages={267--275},
  year={2010},
  publisher={Elsevier}
}

@article{honore2006bounds,
  title={Bounds on parameters in panel dynamic discrete choice models},
  author={Honor{\'e}, Bo E and Tamer, Elie},
  journal={Econometrica},
  volume={74},
  number={3},
  pages={611--629},
  year={2006},
  publisher={Wiley Online Library}
}

@article{manski2007partial,
  title={Partial identification of counterfactual choice probabilities},
  author={Manski, Charles F},
  journal={International Economic Review},
  volume={48},
  number={4},
  pages={1393--1410},
  year={2007},
  publisher={Wiley Online Library}
}

@article{molinari2008partial,
  title={Partial identification of probability distributions with misclassified data},
  author={Molinari, Francesca},
  journal={Journal of Econometrics},
  volume={144},
  number={1},
  pages={81--117},
  year={2008},
  publisher={Elsevier}
}

@article{hu1988generalization,
  title={A generalization of {K}olmogorov's extension theorem and an application to the construction of stochastic processes with random time domains},
  author={Hu, K Y},
  journal={The Annals of Probability},
  volume={16},
  number={1},
  pages={222--230},
  year={1988},
  publisher={Institute of Mathematical Statistics}
}

@book{karatzas1998brownian,
  title={Brownian motion and stochastic calculus},
  author={Karatzas, Ioannis and Shreve, Steven E},
  year={1998},
  publisher={Springer}
}

@article{anastassiou1992monotone,
  title={Monotone and probabilistic wavelet approximation},
  author={Anastassiou, GA and Yu, XM},
  journal={Stochastic Analysis and Applications},
  volume={10},
  number={3},
  pages={251--264},
  year={1992},
  publisher={Taylor \& Francis}
}

@article{anastassiou1992convex,
  title={Convex and coconvex-probabilistic wavelet approximation},
  author={Anastassiou, GA and Yu, XM},
  journal={Stochastic Analysis and Applications},
  volume={10},
  number={5},
  pages={507--521},
  year={1992},
  publisher={Taylor \& Francis}
}

@article{chen2007large,
  title={Large sample sieve estimation of semi-nonparametric models},
  author={Chen, Xiaohong},
  journal={Handbook of econometrics},
  volume={6},
  pages={5549--5632},
  year={2007},
  publisher={Elsevier}
}

@inproceedings{girosi1995approximation,
  title={Approximation error bounds that use {VC}-bounds},
  author={Girosi, Federico},
  booktitle={Proc. International Conference on Artificial Neural Networks, F. Fogelman-Soulie and P. Gallinari, editors},
  volume={1},
  pages={295--302},
  year={1995}
}

@book{van2000asymptotic,
  title={Asymptotic statistics},
  author={{van der Vaart}, Aad W},
  volume={3},
  year={2000},
  publisher={Cambridge university press}
}

@article{hong2018numerical,
  title={The numerical delta method},
  author={Hong, Han and Li, Jessie},
  journal={Journal of Econometrics},
  volume={206},
  number={2},
  pages={379--394},
  year={2018},
  publisher={Elsevier}
}

@article{advani2019mostly,
  title={Mostly harmless simulations? Using {M}onte {C}arlo studies for estimator selection},
  author={Advani, Arun and Kitagawa, Toru and S{\l}oczy{\'n}ski, Tymon},
  journal={Journal of Applied Econometrics},
  year={2019},
  publisher={Wiley Online Library},
  note = {forthcoming}
}

@article{hausman2016individual,
  title={Individual heterogeneity and average welfare},
  author={Hausman, Jerry A and Newey, Whitney K},
  journal={Econometrica},
  volume={84},
  number={3},
  pages={1225--1248},
  year={2016},
  publisher={Wiley Online Library}
}

@article{de2016nonparametric,
  title={Nonparametric errors in variables models with measurement errors on both sides of the equation},
  author={{de Nadai}, Michele and Lewbel, Arthur},
  journal={Journal of Econometrics},
  volume={191},
  number={1},
  pages={19--32},
  year={2016},
  publisher={Elsevier}
}

@article{honore2006bounds2,
  title={Bounds in competing risks models and the war on cancer},
  author={Honor{\'e}, Bo E and Lleras-Muney, Adriana},
  journal={Econometrica},
  volume={74},
  number={6},
  pages={1675--1698},
  year={2006},
  publisher={Wiley Online Library}
}

@article{belloni2011l1,
  title={$\ell1$-penalized quantile regression in high-dimensional sparse models},
  author={Belloni, Alexandre and Chernozhukov, Victor},
  journal={The Annals of Statistics},
  volume={39},
  number={1},
  pages={82--130},
  year={2011},
  publisher={Institute of Mathematical Statistics}
}

@article{norets2013semiparametric,
  title={Semiparametric inference in dynamic binary choice models},
  author={Norets, Andriy and Tang, Xun},
  journal={Review of Economic Studies},
  volume={81},
  number={3},
  pages={1229--1262},
  year={2013},
  publisher={Oxford University Press}
}

@article{hansen1995econometric,
  title={Econometric evaluation of asset pricing models},
  author={Hansen, Lars Peter and Heaton, John and Luttmer, Erzo GJ},
  journal={The Review of Financial Studies},
  volume={8},
  number={2},
  pages={237--274},
  year={1995},
  publisher={Oxford University Press}
}

@article{manski2014identification,
  title={Identification of income--leisure preferences and evaluation of income tax policy},
  author={Manski, Charles F},
  journal={Quantitative Economics},
  volume={5},
  number={1},
  pages={145--174},
  year={2014},
  publisher={Wiley Online Library}
}

@techreport{tebaldi2019nonparametric,
  title={Nonparametric Estimates of Demand in the California Health Insurance Exchange},
  author={Tebaldi, Pietro and Torgovitsky, Alexander and Yang, Hanbin},
  year={2019},
  institution={National Bureau of Economic Research}
}

@techreport{aguiar2019prices,
	title = {Prices, {P}rofits, and {P}roduction: Identification and Counterfactuals},
	author = {Victor Aguiar and Roy Allen and Nail Kashaev},
	year = {2019},
	institution = {University of Western Ontario}
}
\appendix
\section{Proofs}

\subsection{Proof of Proposition \ref{myprop}}
\begin{proof}
The construction of the counterfactual processes $Y_x(w)$ and $X_z(w)$ on $\mathbb{R}^{\mathbb{R}}$ with laws $P_{Y(x)}$ and $P_{X(z)}$ follows immediately from Kolmogorov's extension theorem \citep*[Theorem 2.2.2]{karatzas1998brownian}. Note that the measurable space $([0,1],\mathscr{B}_{[0,1]})$ is large enough to accommodate all paths $Y_x(w)$, $X_z(w)$ satisfying Assumption \ref{skorokhodpath}, as the Skorokhod space equipped with the Skorokhod metric is Polish \citep*[chapter 12]{billingsley1999convergence}, and there aways exist measure-preserving isomorphisms between the unit interval and a Polish space \citep[chapter 9]{bogachev2007measure2}. Hence, one can define $\mathcal{W}=[0,1]$.

Together, the two laws $P_{Y(x)}$ and $P_{X(z)}$ generate the joint law $P_{[Y,X]^*(z)}$ as
\begin{equation}\label{jointprocess}
P_{[Y,X]^*(z)}(A_y,A_x) = \int_{A_x}P_{Y(x)}(A_y)dP_{X(z)}(x),
\end{equation} 
which follows from the exclusion restriction: $P_{Y(x)}$ does not depend on $Z$. The fact that this joint law corresponds to a stochastic process $[Y,X]_z^*$ on $(\mathbb{R}^2)^{\mathbb{R}}$ again follows from Kolmogorov's extension theorem, as $[Y,X]_z^*=(Y_{X_z},X_z)$ is the Cartesian product of the composed process $Y_{X_z}$ and the process of the first stage $X_z$. 

The next thing to show is that under Assumption \ref{skorokhodpath}, the processes $Y_x(w)$, $X_z(w)$, and $[Y,X]_z^*(w)$ are measurable as stochastic processes on the smaller spaces $D(\mathcal{X})$ and $D(\mathcal{Z})$. Consider $X_z(w)$. Since $\mathcal{Z}$ is fixed, under Assumption \ref{skorokhodpath} this stochastic process is progressively measurable with respect to its natural filtration $\{\mathscr{F}_z^X\}$, which is the smallest $\sigma$-algebra with respect to which $X_z$ for non-random $z$ is measurable. Indeed, by definition of the natural filtration, $X_z$ is adapted to it. Moreover, under Assumption \ref{skorokhodpath} every sample path is right-continuous. Then under Assumption \ref{skorokhodpath} by Proposition 1.1.13 in \citet*{karatzas1998brownian}, it holds that $X_z$ is also progressively measurable with respect to its natural filtration $\{\mathscr{F}_z^X\}$. The same argument shows measurability of $Y_x(w)$ with respect to the $\sigma$-field $\{\mathscr{F}_x^Y\}$ under Assumption \ref{skorokhodpath}. 

Now focus on the joint process $[Y,X]_{z}^*$. It is right-continuous under Assumption \ref{skorokhodpath} on its codomain $\mathbb{R}^2$, as the Cartesian product of two continuous functions on the real line is continuous by a projection argument, so that its paths lie in the respective Skorokhod space $D(\mathcal{Z})$ with codomain $\mathbb{R}^2$. Therefore, and since $z$ is nonrandom, Proposition 1.1.13 in \citet*{karatzas1998brownian} implies that $[Y,X]^*_z$ is also progressively measurable with respect to its natural filtration $\{\mathscr{F}_z^{[Y,X]}\}$. Now the key for showing measurability of $[Y,X]_{z}^*$ is to show that the random process $X_z(w)$ is a stopping time for the process $Y_{X_z}(w)$, as this is how part of $[Y,X]_z^*$ is constructed under the exclusion restriction: $[Y,X]^*_z \equiv (Y_{X_z}, X_z)$. But this fact follows from the exclusion restriction: the path $Y_{X_z}(w)$ only depends on the current position $x$ of $X_z(w)$ for each fixed $z$ and not on the actual path $X_z(w)$, which means that $Y_x$ does not depend on future values of $x$ and $x$ does not depend on future values of $Y_x$, or more formally, it holds that $\{X_z\leq x\}\in\mathscr{F}_x^Y$. Therefore, it follows again from Proposition 1.2.18 in \citet*{karatzas1998brownian} that $Y_{X_z}(w)$ is measurable with respect to the $\sigma$-field $\mathscr{F}_{z}^{[Y,X]}$ of all events prior to $z$. 

Finally, the independence restriction $Z\independent W$ implies that one can compare the properties of the stochastic process $[Y,X]_z$ induced by the observable distribution $P_{Y,X|Z=z}$ to the stochastic process $[Y,X]_z^*(w)$ corresponding to $P_{[Y,X]^*(z)}$.
\end{proof}

\subsection{Proof of Lemma \ref{penalizedlemma}}
\begin{proof}
Focus on \eqref{minimizereqpen} first. Under Assumption \ref{radonnikodymass}, we can write $dP_W = \frac{dP_W}{dP_0}dP_0$, so that \eqref{minimizereq} coincides with
\begin{equation}\label{minimizereq2}
\begin{aligned}
& \underset{\substack{\frac{dP_w}{dP_0}\\P_0,P_W\in\mathscr{P}^*(\mathcal{W})}}{\text{minimize/maximize}}\quad  \int f(Y_{x}(w),x)\frac{dP_W}{dP_0}(w)P_0(dw)\\
&\text{s.t.}\thickspace\thickspace \left\| F_{Y,X|Z=z} - \int \mathds{1}_{[0,\cdot]\times[0,\cdot]}\left(Y_{X_z(w)}(w),X_z(w)\right)\frac{dP_W}{dP_0}(w)P_0(dw)\right\|^2_{L^2([0,1]^2)}\leq \varepsilon^*
\end{aligned}
\end{equation}
for all $\varepsilon^*\geq0$. In fact, since the equality constraints hold perfectly in the population, $\varepsilon^*=0$ is allowed. 

Now bound the constraint. 
\begin{align*}
0\leq\varepsilon^*=&\left\| F_{Y,X|Z=z} - \int \mathds{1}_{[0,\cdot]\times[0,\cdot]}\left(Y_{X_z(w)}(w),X_z(w)\right)\frac{dP_W}{dP_0}(w)P_0(dw)\right\|_{L^2([0,1]^2)}\\
=& \left(\int_0^1\int_0^1\left[F_{Y,X|Z=z}(y,x)-\int\mathds{1}_{[0,\cdot]\times[0,\cdot]}\left(Y_{X_z(w)}(w),X_z(w)\right)\frac{dP_W}{dP_0}(w)P_0(dw)\right]^2dydx\right)^{1/2}\\
=&  \left(\int_0^1\int_0^1\left[\int \left\{F_{Y,X|Z=z}(y,x)-\mathds{1}_{[0,\cdot]\times[0,\cdot]}\left(Y_{X_z(w)}(w),X_z(w)\right)\frac{dP_W}{dP_0}(w)\right\}P_0(dw)\right]^2dydx\right)^{1/2}\\
\leq& \int\left(\int_0^1\int_0^1\left[F_{Y,X|Z=z}(y,x) - \mathds{1}_{[0,\cdot]\times[0,\cdot]}\left(Y_{X_z(w)}(w),X_z(w)\right)\frac{dP_W}{dP_0}(w)\right]^2dydx\right)^{1/2}P_0(dw)\\
=&\int\left\|F_{Y,X|Z=z} - \mathds{1}_{[0,\cdot]\times[0,\cdot]}\left(Y_{X_z(w)}(w),X_z(w)\right)\frac{dP_W}{dP_0}(w)\right\|_{L^2([0,1]^2)}P_0(dw),
\end{align*}
where the second-to-last line follows by Minkowski's inequality for integrals \citep*[Theorem 6.19]{folland2013real}. 
Taking squares on both sides and applying Jensen's inequality using the fact that $P_0$ is a probability measure gives
\begin{align*}
0\leq\varepsilon^*=&\left\| F_{Y,X|Z=z} - \int \mathds{1}_{[0,\cdot]\times[0,\cdot]}\left(Y_{X_z(w)}(w),X_z(w)\right)\frac{dP_W}{dP_0}(w)P_0(dw)\right\|^2_{L^2([0,1]^2)}\\
\leq&\int\left\|F_{Y,X|Z=z} - \mathds{1}_{[0,\cdot]\times[0,\cdot]}\left(Y_{X_z(w)}(w),X_z(w)\right)\frac{dP_W}{dP_0}(w)\right\|^2_{L^2([0,1]^2)}P_0(dw)
\end{align*}

Therefore, there must exist some $\varepsilon_0>0$ such that if
\[\int\left\|F_{Y,X|Z=z} - \mathds{1}_{[0,\cdot]\times[0,\cdot]}\left(Y_{X_z(w)}(w),X_z(w)\right)\frac{dP_W}{dP_0}(w)\right\|^2_{L^2([0,1]^2)}P_0(dw)= \varepsilon_0,\] then the original constraint is equal to $\varepsilon^*$.

We can now rewrite the programs in penalized form as
\begin{multline*}
\underset{\substack{\frac{dP_w}{dP_0}\\P_0,P_W\in\mathscr{P}^*(\mathcal{W})}}{\min/\max} \int f(Y_{x}(w),x)\frac{dP_W}{dP_0}(w)P_0(dw)\\
+\lambda_0\int\left\|F_{Y,X|Z=z} - \mathds{1}_{[0,\cdot]\times[0,\cdot]}\left(Y_{X_z(w)}(w),X_z(w)\right)\frac{dP_W}{dP_0}(w)\right\|^2_{L^2([0,1]^2)}P_0(dw),
\end{multline*}
where $\lambda_0$ is the penalty term corresponding to $\varepsilon_0$. Putting the new constraint and the objective together gives
\begin{multline*}
\underset{\substack{\frac{dP_W}{dP_0}\\P_0,P_W\in\mathscr{P}^*(\mathcal{W})}}{\min/\max}\int\left[f(Y_{x}(w),x)\frac{dP_W}{dP_0}(w) \vphantom{\lambda_0\left\|F_{Y,X|Z=z} - \mathds{1}_{[0,\cdot]\times[0,\cdot]}\left(Y_{X_z(w)}(w),X_z(w)\right)\frac{dP_W}{dP_0}(w)\right\|^2_{L^2([0,1]^2)}}\right.\\
\left.+\lambda_0\left\|F_{Y,X|Z=z} - \mathds{1}_{[0,\cdot]\times[0,\cdot]}\left(Y_{X_z(w)}(w),X_z(w)\right)\frac{dP_W}{dP_0}(w)\right\|^2_{L^2([0,1]^2)}\right]P_0(dw),
\end{multline*}
which is \eqref{minimizereqpen}. Since the constraint of \eqref{minimizereq} holds for $\varepsilon=0$, the solution to \eqref{minimizereqpen} coincides with the solution to \eqref{minimizereq} as $\lambda\to\infty$.

The same reasoning holds for the programs \eqref{minimizereqapprox}. One can rewrite them in the same way as above in order to obtain the empirical counterpart of \eqref{minimizereqpen} as
\begin{multline*}
\underset{\substack{\frac{d\hat{P}_W}{d\hat{P}_0}\\\hat{P}_0,\hat{P}_W\in\hat{\mathscr{P}}^*(\mathcal{W})}}{\min/\max}\frac{1}{l}\sum_{i=1}^l\left[f(\tilde{Y}^\kappa_{x}(i),x)\frac{d\hat{P}_W}{d\hat{P}_0}(i) \vphantom{\lambda_0\left\|F_{Y,X|Z=z} - \mathds{1}_{[0,\cdot]\times[0,\cdot]}(\tilde{Y}_x^\kappa(i),\tilde{X}_z^\kappa(i))\frac{d\hat{P}_W}{d\hat{P}_0}(i)\right\|_{L^2([0,1]^2)}}\right.\\
\left.+\lambda_0\left\|F_{Y,X|Z=z} - \mathds{1}_{[0,\cdot]\times[0,\cdot]}(\tilde{Y}_{\tilde{X}^\kappa_z(i)}^\kappa(i),\tilde{X}_z^\kappa(i))\frac{d\hat{P}_W}{d\hat{P}_0}(i)\right\|^2_{L^2([0,1]^2)}\right],
\end{multline*}
which is \eqref{minimizereqapproxpen}.
\end{proof}

\subsection{Proof of Theorem \ref{maintheorem2}}
The proof of the theorem requires the following lemma, which bounds the approximation of the paths $Y_x$ and $X_z$ by the wavelet basis. 
\begin{lemma}\label{waveletlemma}
The wavelet operators $W_\kappa$ acting on a c\`adl\`ag function $f\in L^2(\mathbb{R})$ through
\[W_\kappa(f)(x)\coloneqq \sum_{j=-\infty}^\infty \langle f,\varphi_{\kappa j}\rangle\varphi_{\kappa j}(x),\]
where 
\[\langle f,\varphi_{\kappa j}\rangle\coloneqq \int_{-\infty}^\infty f(t)\varphi_{\kappa j}(t)dt\] and $\varphi$ is the hat-function wavelet basis defined in the main text, satisfy 
\[|W_\kappa(f)(x) - f(x)|\leq \omega'_{f}(2^{-\kappa+1})\qquad\text{for all $x\in\mathbb{R}$ and $k\in\mathbb{Z}$},\] where $\omega'_f$ is the extended modulus of continuity. 
\end{lemma}
\begin{proof}
The proof is analogous to the proof of Theorem 1 in \citet*{anastassiou1992monotone}, but for c\`adl\`ag functions instead of continuous functions.
The hat-function wavelet basis satisfies $\sum_{j=-\infty}^\infty \varphi(x-j) = 1$ on $\mathbb{R}$ \citep*{anastassiou1992monotone}, so that for a square integrable function $f(x)$ we have 
\[W_\kappa(f)(x) - f(x) = 2^{\kappa/2}\sum_{j=-\infty}^\infty \left[\langle f,\varphi_{\kappa j}\rangle - 2^{-\kappa/2}f(x)\right] \varphi(2^\kappa x-j).\]
The hat-function wavelet basis also satisfies $\int_{-\infty}^\infty\varphi(u-j)du=1$ , $j\in\mathbb{Z}$ \citep*{anastassiou1992monotone}. Based on this, and by a change of variables, we have
\begin{align*}
\langle f,\varphi_{\kappa j}\rangle - 2^{-\kappa/2}f(x) &= 2^{\kappa/2}\int_{-\infty}^\infty f(t) \varphi(2^\kappa t-j)dt - 2^{-\kappa/2}f(x)\\
&=2^{-\kappa/2}\int_{-\infty}^\infty f(2^{-\kappa}u) \varphi(u-j)du - 2^{-\kappa/2}f(x)\\
&=2^{-\kappa/2}\int_{-\infty}^\infty [f(2^{-\kappa}u)-f(x)] \varphi(u-j)du.
\end{align*}
Since the support of $\varphi$ is $[-1,1]$ and $\varphi\geq0$, we have for $2^{-\kappa}(-1+j)\leq x \leq 2^{-\kappa}(1+j)$
\[|\langle f,\varphi_{\kappa j}\rangle - 2^{-\kappa/2}f(x)| = 2^{-\kappa/2}\left\lvert\int_{-1+j}^{1+j}[f(2^{-\kappa}u)-f(x)] \varphi(u-j)du\right\rvert \leq 2^{-\kappa/2}\int_{-1+j}^{1+j}|f(2^{-\kappa}u)-f(x)| \varphi(u-j)du.\]

The integral on the right is a Riemann integral, so that we can write it in terms of the upper Darboux integral of the function as
\[2^{-\kappa/2}\int_{-1+j}^{1+j}[f(2^{-\kappa}u)-f(x)] \varphi(u-j)du = 2^{-\kappa/2}\inf_{\Pi[2^{-\kappa+1}]}\sum_{i=1}^\eta(u_i-u_{i-1})\sup_{s\in[u_{i-1},u_i)}|f(2^{-\kappa}s)-f(x)|\varphi(s-j),\]
where $\Pi[2^{k-1}]$ denotes a partition $2^{-\kappa}(-1+j) = u_0< u_1 <\ldots <u_{\eta-1}<u_\eta = 2^{-\kappa}(1+j)$ of the interval $[2^{-\kappa}(-1+j),2^{-\kappa}(1+j)]$ of length $2^{\kappa-1}$, and where the infimum is taken over all partitions $\Pi[2^{k-1}]$ of arbitrary $\eta$.

Then we can bound the upper Darboux integral by
\begin{align*}
&2^{-\kappa/2}\inf_{\Pi[2^{-\kappa+1}]}\sum_{i=1}^\eta(u_i-u_{i-1})\sup_{s\in[u_{i-1},u_i)}|f(2^{-\kappa}s)-f(x)|\varphi(s-j)\\
\leq & 2^{-\kappa/2}\inf_{\Pi[2^{-\kappa+1}]}\sum_{i=1}^\eta(u_i-u_{i-1})\sup_{s\in[u_{i-1},u_i)}|f(2^{-\kappa}s)-f(x)|\sup_{s\in[u_{i-1},u_i)}\varphi(s-j)\\
\leq &  2^{-\kappa/2}\inf_{\Pi[2^{-\kappa+1}]}\sum_{i=1}^\eta(u_i-u_{i-1})\left\{\max_{1\leq i\leq\eta}\sup_{s\in[u_{i-1},u_i)}|f(2^{-\kappa}s)-f(x)|\right\}\sup_{s\in[u_{i-1},u_i)}\varphi(s-j)\\
\leq & 2^{-\kappa/2}\inf_{\Pi[2^{-\kappa+1}]}\max_{1\leq i\leq\eta}\sup_{s\in[u_{i-1},u_i)}|f(2^{-\kappa}s)-f(x)|\sum_{i=1}^\eta(u_i-u_{i-1})\sup_{s\in[u_{i-1},u_i)}\varphi(s-j).
\end{align*}
The first inequality follows from the fact that $\varphi\geq 0$, the second inequality follows by choosing the largest difference over all intervals for a given partition, and the third by the fact that $\max_{1\leq i\leq\eta}\sup_{s\in[u_{i-1},u_i)}|f(2^{-\kappa}s)-f(x)|$ does not depend on $i$ anymore. 

The following is the crucial argument for bounding the Darboux sum. By the fact $f$ is c\`adl\`ag, it only has finitely many jumps that exceed any $\sigma>0$, so that in order to minimize the expression, any partition $\Pi[2^{-\kappa+1}]$ has to partition the interval in such a way that the intervals $[u_{i-1},u_i)$ line up with finitely many points $x$ where $f(x)$ jumps by more than some fixed $\sigma>0$, i.e.~is only continuous on the right with limit on the left. That is, if $x$ is a point of discontinuity of $f$ in the considered interval, then any partition that minimizes $\max_{1\leq i\leq\eta}\sup_{s\in[u_{i-1},u_i)}|f(2^{-\kappa}s)-f(x)|$, needs to have $x = u_i$ for some $i$. Now since $\varphi(\cdot-j)$ is continuous on the given interval, the upper Darboux integral converges to the Riemann integral as the partition becomes finer. This still holds if we only consider partitions which have $x=u_i$ for points of discontinuities of $f$, since the points of discontinuity are countable and hence of (Lebesgue-) measure zero---note that Lebesgue-and Riemann integral coincide. Therefore, the infimum over all partitions $\Pi[2^{\kappa-1}]$ for the upper Darboux integral 
\[\sum_{i=1}^\eta(u_i-u_{i-1})\sup_{s\in[u_{i-1},u_i)}\varphi(s-j)\] coincides with the infimum over all partitions $\Pi[2^{-\kappa}]$ which minimize
\[\max_{1\leq i\leq\eta}\sup_{s\in[u_{i-1},u_i)}|f(2^{-\kappa}s)-f(x)|,\] so that we can write
\begin{align*}
& 2^{-\kappa/2}\inf_{\Pi[2^{-\kappa+1}]}\max_{1\leq i\leq\eta}\sup_{s\in[u_{i-1},u_i)}|f(2^{-\kappa}s)-f(x)|\sum_{i=1}^\eta(u_i-u_{i-1})\sup_{s\in[u_{i-1},u_i)}\varphi(s-j)\\
=& 2^{-\kappa/2}\inf_{\Pi[2^{-\kappa+1}]}\max_{1\leq i\leq\eta}\sup_{s\in[u_{i-1},u_i)}|f(2^{-\kappa}s)-f(x)|\inf_{\Pi[2^{-\kappa+1}]}\sum_{i=1}^\eta(u_i-u_{i-1})\sup_{s\in[u_{i-1},u_i)}\varphi(s-j)\\
=& 2^{-\kappa/2}\omega'_f(2^{-\kappa+1})\int_{-1+j}^{1+j}(u_i-u_{i-1})\varphi(u-j)\\
=& 2^{-\kappa/2}\omega'_f(2^{-\kappa+1})\int_{-\infty}^{\infty}(u_i-u_{i-1})\varphi(u-j)\\
=& 2^{-\kappa/2}\omega'_f(2^{-\kappa+1})
\end{align*}
by the definition of $\omega'_f(2^{-\kappa+1})$, the fact that $\varphi(\cdot-j)$ is Riemann and hence Darboux integrable, and the fact that $\varphi$ integrates to $1$ as argued above.

By the fact that the support of $\varphi$ is $[-1,1]$, it holds that
\begin{align*}
|W_\kappa(f)(x) - f(x)|&\leq 2^{\kappa/2}\sum_{j:2^{\kappa}x-j\in[-1,1]} \left\lvert \langle f,\varphi_{\kappa j}\rangle-2^{-\kappa/2}\right\rvert \varphi(2^\kappa x-j)\\
&=\omega'_f(2^{-\kappa+1})\sum_{j=-\infty}^\infty\varphi(2^\kappa x-j) \\
&= \omega'_f(2^{-\kappa+1}).
\end{align*}
\end{proof}

\begin{proof}[Proof of Theorem \ref{maintheorem2}]
The proof is split into two parts. The first part uses Lemma \ref{waveletlemma} to reduce the complexity of the paths $Y_x$ and $X_z$ via their wavelet approximations $\tilde{Y}_x$ and $\tilde{X}_z$, which will make the second part easier to handle. The second part consists of considering an approximated version of the programs \eqref{minimizereq} as an $M$-estimator in the sense of \citet*[chapter 2.14]{wellner2013weak} and derives a lower bound on the number of paths $l$ to be sampled by using the associated concentration results from empirical process theory.\\

\noindent \emph{Part 1:}
We want to bound the approximation of \eqref{minimizereqpen} by
\begin{multline*}
\underset{\substack{\frac{dP_W}{dP_0}\\P_0,P_W\in\mathscr{P}^*(\mathcal{W})}}{\min/\max}\int\left[f(\tilde{Y}^\kappa_{x}(w),x)\frac{dP_W}{dP_0}(w) \vphantom{+\lambda_0\left\|F_{Y,X|Z=z} - \mathds{1}_{[0,\cdot]\times[0,\cdot]}(\tilde{Y}^\kappa_x(w),\tilde{X}^\kappa_z(w))\frac{dP_W}{dP_0}(w)\right\|_{L^2([0,1]^2)}}\right.\\
\left.+\lambda_0\left\|F_{Y,X|Z=z} - \mathds{1}_{[0,\cdot]\times[0,\cdot]}(\tilde{Y}^\kappa_{\tilde{X}^\kappa_z(w)})(w),\tilde{X}^\kappa_z(w))\frac{dP_W}{dP_0}(w)\right\|^2_{L^2([0,1]^2)}\right]P_0(dw),
\end{multline*}
where the notation $\tilde{Y}_x(w)$ means the application of the wavelet operator from Lemma \ref{waveletlemma} to the path $Y_x(w)$, and analogously for $\tilde{X}_z(w)$. That is, we want to bound
\begin{equation}\label{needtobound}
\begin{aligned}
&\left\lvert\underset{\substack{\frac{dP_W}{dP_0}\\P_0,P_W\in\mathscr{P}^*(\mathcal{W})}}{\min/\max}\int\left[f(Y_{x}(w),x_0)\frac{dP_W}{dP_0}(w) \vphantom{+\lambda_0\left\|F_{Y,X|Z=z} - \mathds{1}_{[0,\cdot]\times[0,\cdot]}(\tilde{Y}^\kappa_x(w),\tilde{X}^\kappa_z(w))\frac{dP_W}{dP_0}(w)\right\|_{L^2([0,1]^2)}}\right.\right.\\
&\qquad\left.+\lambda_0\left\|F_{Y,X|Z=z} - \mathds{1}_{[0,\cdot]\times[0,\cdot]}(Y_{X_z(w)})(w),X_z(w))\frac{dP_W}{dP_0}(w)\right\|^2_{L^2([0,1]^2)}\right]P_0(dw)\\
&-\underset{\substack{\frac{dP_W}{dP_0}\\P_0,P_W\in\mathscr{P}^*(\mathcal{W})}}{\min/\max}\int\left[f(\tilde{Y}^\kappa_{x}(w),x)\frac{dP_W}{dP_0}(w)\vphantom{+\lambda_0\left\|F_{Y,X|Z=z} - \mathds{1}_{[0,\cdot]\times[0,\cdot]}(\tilde{Y}^\kappa_x(w),\tilde{X}^\kappa_z(w))\frac{dP_W}{dP_0}(w)\right\|^2_{L^2([0,1]^2)}}\right.\\
&\qquad\left.\vphantom{\left\lvert\underset{\substack{\frac{dP_W}{dP_0}\\P_0,P_W\in\mathscr{P}^*(\mathcal{W})}}{\min/\max}\int\left[f(Y_{x}(w),x_0)\frac{dP_W}{dP_0}(w) \vphantom{+\lambda_0\left\|F_{Y,X|Z=z} - \mathds{1}_{[0,\cdot]\times[0,\cdot]}(\tilde{Y}^\kappa_x(w),\tilde{X}^\kappa_z(w))\frac{dP_W}{dP_0}(w)\right\|^2_{L^2([0,1]^2)}}\right.\right.}
\left.+\lambda_0\left\|F_{Y,X|Z=z} - \mathds{1}_{[0,\cdot]\times[0,\cdot]}(\tilde{Y}^\kappa_{\tilde{X}^\kappa_z(w)})(w),\tilde{X}^\kappa_z(w))\frac{dP_W}{dP_0}(w)\right\|^2_{L^2([0,1]^2)}\right]P_0(dw)\right\rvert,
\end{aligned}
\end{equation}
By the rules of the maximum and the minimum, both the maximum and the minimum variant can be bounded above by
\begin{equation}\label{wantobound}
\begin{aligned}
\underset{\substack{\frac{dP_W}{dP_0}\\P_0,P_W\in\mathscr{P}^*(\mathcal{W})}}{\max}&\left\lvert\int\left[f(Y_{x}(w),x_0)\frac{dP_W}{dP_0}(w) \vphantom{+\lambda_0\left\|F_{Y,X|Z=z} - \mathds{1}_{[0,\cdot]\times[0,\cdot]}(\tilde{Y}^\kappa_x(w),\tilde{X}^\kappa_z(w))\frac{dP_W}{dP_0}(w)\right\|_{L^2([0,1]^2)}}\right.\right.\\
&\qquad\left.+\lambda_0\left\|F_{Y,X|Z=z} - \mathds{1}_{[0,\cdot]\times[0,\cdot]}(Y_{X_z(w)})(w),X_z(w))\frac{dP_W}{dP_0}(w)\right\|^2_{L^2([0,1]^2)}\right]P_0(dw)\\
&-\int\left[f(\tilde{Y}^\kappa_{x}(w),x)\frac{dP_W}{dP_0}(w)\vphantom{+\lambda_0\left\|F_{Y,X|Z=z} - \mathds{1}_{[0,\cdot]\times[0,\cdot]}(\tilde{Y}^\kappa_x(w),\tilde{X}^\kappa_z(w))\frac{dP_W}{dP_0}(w)\right\|^2_{L^2([0,1]^2)}}\right.\\
&\qquad\left.
\left.+\lambda_0\left\|F_{Y,X|Z=z} - \mathds{1}_{[0,\cdot]\times[0,\cdot]}(\tilde{Y}^\kappa_{\tilde{X}^\kappa_z(w)})(w),\tilde{X}^\kappa_z(w))\frac{dP_W}{dP_0}(w)\right\|^2_{L^2([0,1]^2)}\right]P_0(dw)\right\rvert,
\end{aligned}
\end{equation}

Focus on the terms in brackets first, and in particular the second term; that is, we first want to bound
\begin{equation}\label{firstbound}
\begin{aligned}
&\int\left\lvert \left\|F_{Y,X|Z=z}-\mathds{1}_{[0,\cdot]\times[0,\cdot]}(Y_{X_z(w)}(w),X_z(w))\frac{dP_W}{dP_0}(w) \right\|^2_{L^2([0,1]^2)} \right.\\
&- \left.\left\|F_{Y,X|Z=z}-\mathds{1}_{[0,\cdot]\times[0,\cdot]}(\tilde{Y}^\kappa_{\tilde{X}^\kappa_z(w)}(w),\tilde{X}^\kappa_z(w))\frac{dP_W}{dP_0}(w) \right\|^2_{L^2([0,1]^2)}\right\rvert dP_0(w).
\end{aligned} 
\end{equation}

To do so, we first focus on 
\begin{equation}\label{bothparts}
\begin{aligned}
&\left\lvert \left\|F_{Y,X|Z=z}-\mathds{1}_{[0,\cdot]\times[0,\cdot]}(Y_{X_z(w)}(w),X_z(w))\frac{dP_W}{dP_0}(w) \right\|^2_{L^2([0,1]^2)} \right.\\
&- \left.\left\|F_{Y,X|Z=z}-\mathds{1}_{[0,\cdot]\times[0,\cdot]}(\tilde{Y}^\kappa_{\tilde{X}^\kappa_z(w)}(w),\tilde{X}^\kappa_z(w))\frac{dP_W}{dP_0}(w) \right\|^2_{L^2([0,1]^2)}\right\rvert.
\end{aligned} 
\end{equation}
for all $w\in[0,1]$. The Parallelogram identity of Hilbert spaces \citep*[p.~8]{conway1990course} and the fact that the zero element is orthogonal to all elements in a Hilbert space gives 
\begin{align*}
&\left\lvert \left\|F_{Y,X|Z=z}-\mathds{1}_{[0,\cdot]\times[0,\cdot]}(Y_{X_z(w)}(w),X_z(w))\frac{dP_W}{dP_0}(w) \right\|^2_{L^2([0,1]^2)} \right.\\
&- \left.\left\|F_{Y,X|Z=z}-\mathds{1}_{[0,\cdot]\times[0,\cdot]}(\tilde{Y}^\kappa_{\tilde{X}^\kappa_z(w)}(w),\tilde{X}^\kappa_z(w))\frac{dP_W}{dP_0}(w) \right\|^2_{L^2([0,1]^2)}\right\rvert\\
=&\left\lvert\left\|\mathds{1}_{[0,\cdot]\times[0,\cdot]}(Y_{X_z(w)}(w),X_z(w))\frac{dP_W}{dP_0}(w)\right\|^2_{L^2([0,1]^2)}-\left\|\mathds{1}_{[0,\cdot]\times[0,\cdot]}(\tilde{Y}^\kappa_{\tilde{X}^\kappa_z(w)}(w),\tilde{X}^\kappa_z(w))\frac{dP_W}{dP_0}(w) \right\|^2_{L^2([0,1]^2)}\right\rvert.
\end{align*} 
Writing out and simplifying the integrals, while noting that indicator functions are idempotent, we can bound this once more to obtain
\begin{align*}
&\left\lvert \left\|F_{Y,X|Z=z}-\mathds{1}_{[0,\cdot]\times[0,\cdot]}(Y_{X_z(w)}(w),X_z(w))\frac{dP_W}{dP_0}(w) \right\|^2_{L^2([0,1]^2)} \right.\\
&\left.- \left\|F_{Y,X|Z=z}-\mathds{1}_{[0,\cdot]\times[0,\cdot]}(\tilde{Y}^\kappa_{\tilde{X}^\kappa_z(w)}(w),\tilde{X}^\kappa_z(w))\frac{dP_W}{dP_0}(w) \right\|^2_{L^2([0,1]^2)}\right\rvert\\
=&\left\lvert\int_{[0,1]^2}\left[\mathds{1}_{[0,\cdot]\times[0,\cdot]}(Y_x(w),X_z(w))-\mathds{1}_{[0,\cdot]\times[0,\cdot]}(\tilde{Y}^\kappa_{\tilde{X}^\kappa_z(w)}(w),\tilde{X}^\kappa_z(w))\right]dydx\right\rvert\left(\frac{dP_W}{dP_0}\right)^2(w).
\end{align*}

We can write the last term as
\begin{align*}
&\left\lvert\int_{[0,1]^2}\left[\mathds{1}_{[0,\cdot]\times[0,\cdot]}(Y_x(w),X_z(w))-\mathds{1}_{[0,\cdot]\times[0,\cdot]}(\tilde{Y}^\kappa_{\tilde{X}^\kappa_z(w)}(w),\tilde{X}^\kappa_z(w))\right]dydx\right\rvert\left(\frac{dP_W}{dP_0}\right)^2(w)\\
=&\left\lvert\int_{[0,1]^2}\left[\mathds{1}_{[Y_{X_z(w)}(w),1]\times[X_z(w),1]}(\cdot,\cdot)-\mathds{1}_{[\tilde{Y}^\kappa_{\tilde{X}^\kappa_z(w)}(w),1]\times[\tilde{X}^\kappa_z(w),1]}(\cdot,\cdot)dydx\right]\right\rvert\left(\frac{dP_W}{dP_0}\right)^2(w)\\
=&\left\lvert\int_{Y_{X_z(w)}(w)\wedge \tilde{Y}^\kappa_{\tilde{X}^\kappa_z(w)}(w)}^{Y_{X_z(w)}(w)\vee \tilde{Y}^\kappa_{\tilde{X}^\kappa_z(w)}(w)}\int_{X_z(w)\wedge \tilde{X}^\kappa_z(w)}^{X_z(w)\vee \tilde{X}^\kappa_z(w)}dydx\right\rvert\left(\frac{dP_W}{dP_0}\right)^2(w)\\
=&\left(\frac{dP_W}{dP_0}\right)^2(w)\left\lvert Y_{X_z(w)}(w)-\tilde{Y}^\kappa_{\tilde{X}^\kappa_z(w)}(w)\right\rvert\left\lvert X_z(w)-\tilde{X}^\kappa_z(w)\right\rvert.
\end{align*}

By Lemma \ref{waveletlemma} it holds
\[\left\lvert Y_x(w) - \tilde{Y}_x^\kappa(w)\right\rvert\leq \omega'_{Y_x(w)}(2^{-\kappa+1}),\qquad \left\lvert X_z(w) - \tilde{X}_z^\kappa(w)\right\rvert\leq \omega'_{X_z(w)}(2^{-\kappa+1})\qquad\text{for all $w\in[0,1]$,}\]
so that
\begin{align*}
&\left\lvert \left\|F_{Y,X|Z=z}-\mathds{1}_{[0,\cdot]\times[0,\cdot]}(Y_{X_z(w)}(w),X_z(w))\frac{dP_W}{dP_0}(w) \right\|^2_{L^2([0,1]^2)} \right.\\
&\left.- \left\|F_{Y,X|Z=z}-\mathds{1}_{[0,\cdot]\times[0,\cdot]}(\tilde{Y}^\kappa_{\tilde{X}^\kappa_z(w)}(w),\tilde{X}^\kappa_z(w))\frac{dP_W}{dP_0}(w) \right\|^2_{L^2([0,1]^2)}\right\rvert\\
\leq&\omega'_{Y_x(w)}(2^{-\kappa+1})\omega'_{X_z(w)}(2^{-\kappa+1})\left(\frac{dP_W}{dP_0}\right)^2(w),
\end{align*}
and therefore
\begin{equation}\label{firstbound}
\begin{aligned}
&\int\left\lvert \left\|F_{Y,X|Z=z}-\mathds{1}_{[0,\cdot]\times[0,\cdot]}(Y_{X_z(w)}(w),X_z(w))\frac{dP_W}{dP_0}(w) \right\|^2_{L^2([0,1]^2)} \right.\\
&- \left.\left\|F_{Y,X|Z=z}-\mathds{1}_{[0,\cdot]\times[0,\cdot]}(\tilde{Y}^\kappa_{\tilde{X}^\kappa_z(w)}(w),\tilde{X}^\kappa_z(w))\frac{dP_W}{dP_0}(w) \right\|^2_{L^2([0,1]^2)}\right\rvert dP_0(w)\\
\leq& \int\omega'_{Y_x(w)}(2^{-\kappa+1})\omega'_{X_z(w)}(2^{-\kappa+1}) \left(\frac{dP_W}{dP_0}\right)^2(w)dP_0(w)\\
\leq& \sup_{w\in[0,1]}\omega'_{Y_x(w)}(2^{-\kappa+1})\omega'_{X_z(w)}(2^{-\kappa+1}) C_{RN},
\end{aligned} 
\end{equation}
where the last line follows from H\"older's inequality in combination with the fact that the integral with respect to the Radon-Nikodym density is a probability integral and that $\sup_{w\in[0,1]}\frac{dP_W}{dP_0}(w)\leq C_{RN}$ by Assumption \ref{radonnikodymass}.

The same decomposition can be achieved for the objective function, i.e.~the first term in \eqref{wantobound}. Indeed, if $f(Y_x(w),x_0)$ for some $x_0\in[0,1]$ is H\"older-continuous with coefficient $\alpha$ and constant $K<+\infty$, one obtains 
\begin{align*}
&\left\lvert \int \left[f(Y_x(w),x_0) - f(\tilde{Y}^\kappa_{x}(w), x_0)\right]\frac{dP_W}{dP_0}(w)dP_0(w)\right\rvert\\
\leq& \int \left\vert f(Y_x(w),x_0) - f(\tilde{Y}^\kappa_{x}(w), x_0)\right\rvert\frac{dP_W}{dP_0}(w)dP_0(w)\\
=& \int K\left\vert Y_x(w) - \tilde{Y}^\kappa_{x}(w)\right\rvert^\alpha\frac{dP_W}{dP_0}(w)dP_0(w)\\
\leq&\sup_{w\in[0,1]}\left(\omega'_{Y_x(w)}(2^{-\kappa+1})\right)^{\alpha}
\end{align*} by the fact that the integration is with respect to a probability measure.

If $f(Y_x,x_0)$ takes the form $f(Y_x,A_x)=\mathds{1}_{A_y}(Y_{A_x})$, we need to approximate them by a logistic function $\mathcal{S}(Y_x(w)),y,\eta)$, and use the Lipschitz continuity of logistic functions. For this, we use Theorem 7 in \citet*{anguelov2015hausdorff}, which bounds the Hausdorff distance\footnote{The Hausdorff distance between two non-empty subsets $\mathcal{X}_0$ and $\mathcal{X}_1$ of $\mathbb{R}^d$ is defined as $d_H(\mathcal{X}_0,\mathcal{X}_1) \coloneqq \max\left\{\sup_{x\in\mathcal{X}_0}\inf_{x'\in\mathcal{X}_1} \|x-x'\|, \sup_{x'\in\mathcal{X}_1}\inf_{x\in\mathcal{X}_0} \|x-x'\|\right\}$.} between the graphs of the logit function and the indicator function. This approximation in our notation reads
\[d_H\left(\mathds{1}_{[0,y]}(Y_{x_0}),\mathcal{S}(Y_x(w)),y,\eta)\right)\leq \frac{\log(\eta+1)}{\eta+1}\left(1+O\left(\frac{\log\log(\eta+1)}{\log(\eta+1)}\right)\right).\]
The actual approximation then follows exactly as in the Lipschitz case, as the logistic function is Lipschitz continuous 

Putting both terms together bounds and noting that the bound hold for the supremum over all $w\in[0,1]$, so that it also holds uniformly for all probability measures $\frac{dP_W}{dP_0}(w)dP_0(w)$ on $[0,1]$ by H\"older's inequality, we can bound \eqref{wantobound} and therefore \eqref{needtobound} by
\begin{equation}
\eqref{needtobound}\leq \sup_{w\in[0,1]}\omega'_{Y_x(w)}(2^{-\kappa+1})\omega'_{X_z(w)}(2^{-\kappa+1})C_{RN}+\sup_{w\in[0,1]}K\left(\omega'_{Y_x(w)}(2^{-\kappa+1})\right)^\alpha.
\end{equation} 
When $f(Y_x,x_0)\coloneqq \mathds{1}_{A_y}(Y_{A_x}(w))$ in Assumption \ref{objectivefunctionstronger}, the bound is
\begin{multline}
\eqref{needtobound}\leq \sup_{w\in[0,1]}\omega'_{Y_x(w)}(2^{-\kappa+1})\omega'_{X_z(w)}(2^{-\kappa+1})C_{RN}+\sup_{w\in[0,1]}K\omega'_{Y_x(w)}(2^{-\kappa+1})\\
+\frac{\log(\eta+1)}{\eta+1}\left(1+O\left(\frac{\log\log(\eta+1)}{\log(\eta+1)}\right)\right),
\end{multline} 
where $K=\frac{\eta\cdot e^{\eta(y+Y_x(w))}}{\left(e^{\eta Y_x(w)}+e^{\eta y}\right)^2}$.

This gives the first part of the bound of Theorem \ref{maintheorem2}.\\

\noindent\emph{Part 2:}
The idea is to consider $\frac{dP_W}{dP_0}$ as an M-estimator of the programs \eqref{minimizereqapproxpen}. In the following we write $h(w)\coloneqq \frac{dP_W}{dP_0}(w)$ to save on notation. We also write
\begin{align*}
&\mathbb{P}_{0,l}m_{h}\coloneqq \frac{1}{l}\sum_{i=1}^l m_{h}(i) \\
\equiv&\frac{1}{l}\sum_{i=1}^l\left[f(\tilde{Y}^\kappa_{x}(i),x)\hat{h}(i)+\lambda_0\left\|F_{Y,X|Z=z} - \mathds{1}_{[0,\cdot]\times[0,\cdot]}(\tilde{Y}_x^\kappa(i),\tilde{X}_z^\kappa(i))\hat{h}(i)\right\|^2_{L^2([0,1]^2)}\right],
\end{align*}
where $\mathbb{P}_{0,l}$ denotes the empirical law of the representative law $P_0$ for $l$ samples of paths.
Analogously, we define
\begin{align*}
&P_{0}m_{h}\coloneqq \int m_{h}(w)dP_0(w) \\
\equiv&\int\left[f(\tilde{Y}^\kappa_{x}(w),x)h(w)+\lambda\left\|F_{Y,X|Z=z} - \mathds{1}_{[0,\cdot]\times[0,\cdot]}(\tilde{Y}_x^\kappa(w),\tilde{X}_z^\kappa(w))h(w)\right\|^2_{L^2([0,1]^2)}\right]dP_0(w).
\end{align*}

The idea now is to bound the uniform entropy numbers\footnote{The uniform entropy number relative to the $L^2$-norm of a class $\mathcal{F}$ of functions is defined as $\sup_{Q}\log N(\varepsilon\|F\|_{L^2(Q)},\mathcal{F},L^2(Q))$, where the supremum runs over all probability measures with finite support, and where $N(\varepsilon,\mathcal{F},L^2(Q))$ denotes the $\epsilon$-covering number of the class $\mathcal{F}$ with respect to the $L^2$-norm of $Q$, see \citet*[p.~84]{wellner2013weak}.} of the set of functions $\mathcal{M}_h$, i.e.~the set of all admissible functions $m_h$. We can bound their uniform entropy in terms of the uniform entropy of the Radon-Nikodym derivative $h(w)$. To do so, we show that $m_h$ is Lipschitz-continuous $h_n$. Note in this respect that we do not need to include the complexity of the paths $\tilde{Y}_x(w)$ and $\tilde{X}_z(w)$, because they are fixed for each $w$ by the theoretical construction in Proposition \ref{myprop}. 

So consider 
\[m_h(w)\coloneqq f(\tilde{Y}^\kappa_x(w),x)h(w)+\lambda\left\|F_{Y,X|Z=z} - \mathds{1}_{[0,\cdot]\times[0,\cdot]}(\tilde{Y}_x^\kappa(w),\tilde{X}_z^\kappa(w))h(w)\right\|^2_{L^2([0,1]^2)}.\] We can bound
\begin{align*}
&\left\lvert f(\tilde{Y}^\kappa_x(w),x)h(w)-f(\tilde{Y}^\kappa_x(w),x)g(w)+\left\|F_{Y,X|Z=z}-\mathds{1}_{[0,\cdot]\times[0,\cdot]}(\tilde{Y}^\kappa_{\tilde{X}^\kappa_z(w)}(w),\tilde{X}^\kappa_z(w))h(w) \right\|^2_{L^2([0,1]^2)} \right.\\
&\left.- \left\|F_{Y,X|Z=z}-\mathds{1}_{[0,\cdot]\times[0,\cdot]}(\tilde{Y}^\kappa_{\tilde{X}^\kappa_z(w)}(w),\tilde{X}^\kappa_z(w))g(w) \right\|^2_{L^2([0,1]^2)}\right\rvert\\
\leq&M\left\lvert h(w)-g(w)\right\rvert\\
&+\left\lvert\left\|\mathds{1}_{[0,\cdot]\times[0,\cdot]}(\tilde{Y}^\kappa_{\tilde{X}^\kappa_z(w)}(w),\tilde{X}^\kappa_z(w))h(w)\right\|^2_{L^2([0,1]^2)}-\left\|\mathds{1}_{[0,\cdot]\times[0,\cdot]}(\tilde{Y}^\kappa_{\tilde{X}^\kappa_z(w)}(w),\tilde{X}^\kappa_z(w))g(w) \right\|^2_{L^2([0,1]^2)}\right\rvert\\
\leq &M\left\lvert h(w)-g(w)\right\rvert\\
&+2C_{RN}\left\|\mathds{1}_{[0,\cdot]\times[0,\cdot]}(\tilde{Y}^\kappa_{\tilde{X}^\kappa_z(w)}(w),\tilde{X}^\kappa_z(w))h(w)-\mathds{1}_{[0,\cdot]\times[0,\cdot]}(\tilde{Y}^\kappa_{\tilde{X}^\kappa_z(w)}(w),\tilde{X}^\kappa_z(w))g(w)) \right\|_{L^2([0,1]^2)}\\
=&M\left\lvert h(w)-g(w)\right\rvert+2C_{RN}\left(\int_{[0,1]^2}\mathds{1}_{[0,\cdot]\times[0,\cdot]}(\tilde{Y}^\kappa_{\tilde{X}^\kappa_z(w)}(w),\tilde{X}^\kappa_z(w))\left(h(w)-g(w)\right)^2dydx\right)^{1/2}\\
=&M\left\lvert h(w)-g(w)\right\rvert+2C_{RN}\left(\int_{[0,1]^2}\mathds{1}_{[\tilde{Y}_{\tilde{X}_z(w)}(w),1]\times[\tilde{X}_z(w),1]}(\cdot,\cdot)\left(h(w)-g(w)\right)^2dydx\right)^{1/2}\\
=&M\left\lvert h(w)-g(w)\right\rvert+2C_{RN}\left(\int_{\tilde{Y}^\kappa_{\tilde{X}^\kappa_z(w)}(w)}^{1}\int_{\tilde{X}_z(w)}^{1}dydx\right)^{1/2}\left\lvert h(w)-g(w)\right\rvert\\
=&M\left\lvert h(w)-g(w)\right\rvert+2C_{RN}\left\lvert 1-\tilde{Y}^\kappa_{\tilde{X}^\kappa_z(w)}(w)\right\rvert^{1/2}\left\lvert 1-\tilde{X}^\kappa_z(w)\right\rvert^{1/2}\left\lvert h(w)-g(w)\right\rvert\\
\leq&2C_{RN}M\left\lvert h(w)-g(w)\right\rvert,
\end{align*}
where $h(w)$ and $g(w)$ are two different Radon-Nikodym derivatives satisfying Assumption \ref{radonnikodymass} with bound $C_{RN}$, and where $M<+\infty$ is the bound on the objective function under Assumption \ref{objectivefunctionstronger}. The first inequality follows from the Parallelogram identity in Hilbert spaces \citep*[p.~8]{conway1990course} in connection with the fact that the zero element is orthogonal to every other element. The second inequality follows from 
\[|\|x\|^2-\|y\|^2|=|\|x\|-\|y\|||\|x\|+\|y\||,\] the triangle inequality, and the fact that the Radon Nikodym derivatives are bounded above by $C_{RN}$. The first equality follows from the idempotency of the indicator functions.

Due to the Lipschitz property of $m_h$, every $\varepsilon$-covering of the set of Radon-Nikodym derivatives is an $(C_{RN}\cdot M\cdot\varepsilon)$-covering of $\mathcal{M}_h$ (e.g.~section 2.7.4 in \citet*{wellner2013weak} and the fact that bracketing numbers bound covering numbers with respect to half the radius), so that the uniform entropy number of $\mathcal{M}_h$ can be bounded by the uniform entropy number of the set of all Radon-Nikodym derivatives. By Assumption \ref{radonnikodymass}, every $h$ is $\beta$-H\"older continuous, and the entropy number of this set can be bounded by \citep*[Theorem 2.7.1]{wellner2013weak}
\[\sup_Q \log N\left(\varepsilon C_{RN},C^\beta_{C_{RN}}([0,1]),L^2(Q)\right)\leq \sup_Q \log N\left(\varepsilon C_{RN},C^\beta_{C_{RN}}([0,1]),\|\cdot\|_\infty\right)\leq C\left(\frac{1}{\varepsilon}\right)^{1/\beta},\] as $\mathds{1}_{[0,1]}C_{RN}$ is an envelope function. Here the constant $C<+\infty$ depends on $\beta$ and $C_{RN}$. Therefore, we can bound the entropy number of $\mathcal{M}_h$ by
\[\sup_Q \log N\left(\varepsilon MC_{RN},\mathcal{M}_h,L^2(Q)\right)\leq \sup_Q \log N\left(\varepsilon C_{RN},C^\beta_{C_{RN}}([0,1]),L^2(Q)\right)\leq C\left(\frac{1}{\varepsilon}\right)^{1/\beta}.\]

Now based on the bound \eqref{wantobound} the goal is to bound its analogue in the form of 
\[P^*\left(\sup_{m_h\in\mathcal{M}_h}\left\lvert \mathbb{P}_{l,0}m_h - P_0m_h\right\rvert>t\right)\] for small $t>0$, where $P^*$ denotes outer probability to avoid measurability issues.
It holds
\begin{align*}
&P^*\left(\sup_{m_h\in\mathcal{M}_h}\left\lvert \mathbb{P}_{l,0}m_h - P_0m_h\right\rvert>t\right)\\
=&P^*\left(\sqrt{l}\sup_{m_h\in\mathcal{M}_h}\left\lvert \mathbb{P}_{l,0}m_h - P_0m_h\right\rvert>\sqrt{l}t\right).
\end{align*}
The tail bound derived in Theorem 2.14.10 of \citet*{wellner2013weak} implies that for every $\delta>0$ and $t>0$ 
\begin{equation}\label{tailbound}
P^*\left(\sqrt{l}\sup_{m_h\in\mathcal{M}_h}\left\lvert \mathbb{P}_{l,0}m_h - P_0m_h\right\rvert>\sqrt{l}t\right)\leq L \exp\left(D(\sqrt{l}t)^{U+\delta}\right)\exp(-2lt^2),
\end{equation}
where $U\coloneqq \frac{6\beta-1}{\beta(2\beta+1)}$ and the constants $L$ and $D$ depends on the constant $C$, $\beta$, and $\delta$. 
Since $\beta>\frac{1}{2}$, there exists a small enough $\delta>0$ such that $U+\delta <2$. Therefore, one can bound the term further by 
\[L \exp\left(D(\sqrt{l}t)^{U+\delta}\right)\exp(-2lt^2)\leq \bar{C}\exp(-\bar{D}lt^2),\]
for constants $0<\bar{C}<+\infty$ and (potentially small) $0<\bar{D}<2$. 

Now apply inversion. It follows that 
\[P^*\left(\sup_{m_h\in\mathcal{M}_h}\left\lvert \mathbb{P}_{l,0}m_h - P_0m_h\right\rvert>\sqrt{\frac{\log\left(\frac{\bar{C}}{\rho}\right)}{\bar{D}l}}\right)\leq \rho\qquad\text{for $\rho\in(0,1)$,}\] and so
\[P^*\left(\sup_{m_h\in\mathcal{M}_h}\left\lvert \mathbb{P}_{l,0}m_h - P_0m_h\right\rvert\leq\sqrt{\frac{\log\left(\frac{\bar{C}}{\rho}\right)}{\bar{D}l}}\right)\geq 1-\rho\qquad\text{for $\rho\in(0,1)$.}\] Therefore, with probability at least $1-\rho$ it holds that 
\[\sup_{m_h\in\mathcal{M}_h}\left\lvert \mathbb{P}_{l,0}m_h - P_0m_h\right\rvert\leq\sqrt{\frac{\log\left(\frac{\bar{C}}{\rho}\right)}{\bar{D}l}}.\] 

Now we put everything together.
Denote by $V^*$ and $V_*$ the value function of the maximization and minimization of \eqref{minimizereqpen}, respectively, and by $\tilde{V}^*_{l,\kappa}$ and $\tilde{V}_{*,l,\kappa}$ the value functions of \eqref{minimizereqapproxpen}. Then by part 1 and part 2 it holds in combination with the triangle inequality that in the case where $f(Y_x,x_0)$ is Lipschitz continuous
\begin{multline*}
\max\{|V^* - \tilde{V}^*_{l,\kappa}|,|V_*-\tilde{V}_{*,l,\kappa}|\}\\
\leq \left[\sup_{w\in[0,1]}\omega'_{Y_x(w)}(2^{-\kappa+1})\omega'_{X_z(w)}(2^{-\kappa+1})C_{RN}+\sup_{w\in[0,1]}K\left(\omega'_{Y_x(w)}(2^{-\kappa+1})\right)^\alpha \right]+\sqrt{\frac{\log\left(\frac{\bar{C}}{\rho}\right)}{\bar{D}l}}
\end{multline*}
with probability of at least $1-\rho$, and analogous when $f(Y_x,x_0)=\mathds{1}_{[0,y]}(Y_{x_0}(w))$.
\end{proof}

\subsection{Proof of Proposition \ref{consistencyprop}}
\begin{proof}
From the triangle inequality it follows that
\begin{multline}\label{toshowbound}
P^*\left(\left\lvert\hat{\tilde{V}}_{*,l,\kappa}(\hat{F}_{Y,X|Z}) - V_{*}(F_{Y,X|Z})\right\rvert\right)\leq \\
P^*\left(\left\lvert\hat{\tilde{V}}_{*,l,\kappa}(\hat{F}_{Y,X|Z;n}) - V_{*}(\hat{F}_{Y,X|Z;n})\right\rvert\right)+P^*\left(\left\lvert V_{*}(\hat{F}_{Y,X|Z;n}) - V_{*}(F_{Y,X|Z})\right\rvert\right),
\end{multline}
where the probability measure on the left hand side is taken with respect to the joint distribution of iid data $(Y_k,X_k)_{k=1}^n$ and the sampled paths $(Y_x(i),X_z(i))_{i=1}^l$. The same decomposition is valid for the maximization $V^*$. The randomness of the first term on the right hand sidewith respect to both data and paths, while the randomness of the second term on the right hand side is only with respect to the data. 
We now bound each term separately. 

The proof that the first term goes to zero is exactly the same as the proof of Theorem \ref{maintheorem2} up to the inversion argument. Consider the case where $f(Y_x,x_0)$ is Lipschitz continuous, as the other case is perfectly analogous up to the logistic term. Based on Theorem \ref{maintheorem2}, we can bound
\begin{align*}
&P^*\left(\left\lvert\hat{\tilde{V}}_{*,l,\kappa}(\hat{F}_{Y,X|Z;n}) - V_{*}(\hat{F}_{Y,X|Z;n})\right\rvert\right) \\
\leq&\sup_{w\in[0,1]}\omega'_{Y_x(w)}(2^{-\kappa+1})\omega'_{X_z(w)}(2^{-\kappa+1})C_{RN}+\sup_{w\in[0,1]}K\left(\omega'_{Y_x(w)}(2^{-\kappa+1})\right)^\alpha+\bar{C}\exp(-\bar{D}lt^2),
\end{align*}
for constants $0<\bar{C}<+\infty$ and (potentially small) $0<\bar{D}<2$. Letting $l,\kappa\to\infty$ shows that the first term of \eqref{toshowbound} goes to zero. 

For the second term of \eqref{toshowbound}, note that $\hat{F}_{Y,X|Z;n}$ is in $L^2([0,1]^2)$ for all $n$ and as $n\to\infty$. Moreover, $V_{*}(\hat{F}_{Y,X|Z;n})$ is continuous in its argument under Assumption \ref{nonemptinessass} the $L^2$-norm is continuous in its argument, and $V_{*}$ is a composition of the $L^2$-norm and other continuous functions. Now by the Glivenko-Cantelli Theorem \citep*[Theorem 19.1]{van2000asymptotic}, it holds that 
\[P^*\left(\|\hat{F}_{Y,X|Z;n}-F_{Y,X|Z}\|_{L^\infty([0,1]^2)}\right)\to0,\] so that by the Continuous Mapping Theorem \citep*[Theorem 1.11.1]{wellner2013weak}
\[P^*\left(\left\lvert(V_{*}(\hat{F}_{Y,X|Z;n})-V_{*}(F_{Y,X|Z})\right\rvert\right)\to0.\]
This shows that the second term of the right-hand side of \eqref{toshowbound} goes to zero as $n\to\infty$. Finally, this holds for all $\lambda\in\mathbb{R}$. 
\end{proof}

\subsection{Proof of Proposition \ref{largesampleprop}}
\begin{proof}
Focus on the minimization, as the maximization is completely analogous. We want to obtain the large-sample distribution 
\[\sqrt{n}(\hat{\tilde{V}}_{*,l,\kappa}(\hat{F}_{Y,X|Z=z;n})-V_{*}(F_{Y,X|Z=z})),\]
which analogously to the proofs before we split into two parts:
\begin{align*}
&\sqrt{n}(\hat{\tilde{V}}_{*,l,\kappa}(\hat{F}_{Y,X|Z=z;n})-V_{*}(F_{Y,X|Z=z}))\\
=&\sqrt{n}(\hat{\tilde{V}}_{*,l,\kappa}(\hat{F}_{Y,X|Z=z;n})-V_{*}(\hat{F}_{Y,X|Z=z;n}))+\sqrt{n}(V_{*}(\hat{F}_{Y,X|Z=z;n})-V_{*}(F_{Y,X|Z=z})).
\end{align*}
We show that the first part goes to zero in probability and derive the large sample distribution of the second part using the functional delta method \citep*[Theorem 2.1]{shapiro1991asymptotic}.

The bound on the first part is perfectly analogous to the bound in the proof Proposition \ref{consistencyprop}, using Theorem \ref{maintheorem2}. The only difference is that we have to account for the term $\sqrt{n}$, which affects the probability term as follows:
\begin{align*}
&P^*\left(\sqrt{n}\sup_{m_{h,n}\in\mathcal{M}_{h,n}}\left\lvert \mathbb{P}_{l,0}m_{h,n} - P_0m_{h,n}\right\rvert>t\right)\\
=&P^*\left(\sqrt{l}\sup_{m_{h,n}\in\mathcal{M}_{h,n}}\left\lvert \mathbb{P}_{l,0}m_{h,n} - P_0m_{h,n}\right\rvert>t\sqrt{\frac{l}{n}}\right).
\end{align*}
The tail bound derived in Theorem 2.14.10 of \citet*{wellner2013weak} implies that for every $\delta>0$ and $t>0$ 
\begin{equation}\label{tailbound}
P^*\left(\sqrt{l}\sup_{m_{h,n}\in\mathcal{M}_{h,n}}\left\lvert \mathbb{P}_{l,0}m_{h,n} - P_0m_{h,n}\right\rvert>\sqrt{\frac{l}{n}}t\right)\leq L \exp\left(D\left(\sqrt{\frac{l}{n}}t\right)^{U+\delta}\right)\exp\left(-\frac{2lt^2}{n}\right),
\end{equation}
where $U\coloneqq \frac{6\beta-1}{\beta(2\beta+1)}$ and the constants $L$ and $D$ depends on the constant $C$, $\beta$, and $\delta$. 
Since $\beta>\frac{1}{2}$, there exists a small enough $\delta>0$ such that $U+\delta <2$. Therefore, one can bound the term further by 
\[L \exp\left(D\left(\sqrt{\frac{l}{n}}t\right)^{U+\delta}\right)\exp\left(-\frac{2lt^2}{n}\right)\leq \bar{C}\exp\left(-\bar{D}\frac{lt^2}{n}\right),\]
for constants $0<\bar{C}<+\infty$ and (potentially small) $0<\bar{D}<2$.

Therefore, just as in the proof of Proposition \ref{consistencyprop}, we can bound 
\begin{multline*}
P^*\left(\sqrt{n}\left(\hat{\tilde{V}}_{*,l,\kappa}(\hat{F}_{Y,X|Z;n}) - V_{*}(\hat{F}_{Y,X|Z;n})\right)\right) \\
\leq\sqrt{n}\left[\sup_{w\in[0,1]}\omega'_{Y_x(w)}(2^{-\kappa+1})\omega'_{X_z(w)}(2^{-\kappa+1})C_{RN}+\sup_{w\in[0,1]}K\left(\omega'_{Y_x(w)}(2^{-\kappa+1})\right)^\alpha\right]+\bar{C}\exp(-\bar{D}\frac{lt^2}{n}),
\end{multline*}
for constants $0<\bar{C}<+\infty$ and (potentially small) $0<\bar{D}<2$.

Thus, the requirements for the first term to go to zero are
\begin{align*}
&\sqrt{n}\left[\sup_{w\in[0,1]}\omega'_{Y_x(w)}(2^{-\kappa+1})\omega'_{X_z(w)}(2^{-\kappa+1})C_{RN}+\sup_{w\in[0,1]}K\left(\omega'_{Y_x(w)}(2^{-\kappa+1})\right)^\alpha\right]\to0\qquad\text{and}\\
&\frac{n}{l}\to 0 
\end{align*}
as $n\to\infty$.

Now we derive the large sample distribution of the second term 
\[\sqrt{n}(V_{*}(\hat{F}_{Y,X|Z=z;n})-V_{*}(F_{Y,X|Z=z})).\]
We want to apply Theorem 4.13 in \citet*{bonnans2013perturbation}. For this, we need to compute the directional derivative $\delta_{F_{Y,X|Z}} m_h(w;F)$ of
\[m_h(w;F_{Y,X|Z})\coloneqq f(\tilde{Y}^\kappa_x(w),x)h(w)+\lambda\left\|F_{Y,X|Z=z} - \mathds{1}_{[0,\cdot]\times[0,\cdot]}(\tilde{Y}_x^\kappa(w),\tilde{X}_z^\kappa(w))h(w)\right\|^2_{L^2([0,1]^2)}\] at $F_{Y,X|Z}(y,x)$ in direction $F(y,x)\in\mathcal{F}_{Y,X|Z}$:
\begin{align*}
\delta_{F_{Y,X|Z}} m_h(w;F) \coloneqq& \left.\frac{d}{dt}\right\rvert_{t=0} m_h(w;F_{Y,X|Z}-F)\\
=&\int2\lambda\left\langle F, F_{Y,X|Z}-\mathds{1}_{[0,\cdot]\times[0,\cdot]}(Y_{X_z(w)}(w),X_z(w))h(w)\right\rangle dP_0(w),
\end{align*}
where we used the dominated convergence theorem in combination with the fact that $h$ is bounded and the integrals are over $[0,1]$, and where $\left\langle f,g\right\rangle$ is the inner product in $L^2([0,1])$, i.e.
\[\left\langle f,g\right\rangle\coloneqq\int_{[0,1]^2}f(y,x)g(y,x)dydx.\]

Both $\delta_{F_{Y,X|Z}}m_h(w;F)$ and $m_h(w;\cdot)$ are clearly continuous in $F$ and $F_{Y,X|Z}$. Therefore, under Assumption \ref{nonemptinessass}, which implies the inf-compactness assumption of \citet*[p.~272]{bonnans2013perturbation} in our case, it follows from Theorem 4.13 in \citet*{bonnans2013perturbation} that the directional derivative $\delta_{F_{Y,X|Z}}V_{*}(F)$ satisfies
\[\delta_{F_{Y,X|Z}}V_{*}(F) = \min_{h\in\mathcal{S}(F_{Y,X|Z})}\int2\lambda\left\langle F, F_{Y,X|Z}-\mathds{1}_{[0,\cdot]\times[0,\cdot]}(Y_{X_z(w)}(w),X_z(w))h(w)\right\rangle dP_0(w),\]
where $\mathcal{S}(F_{Y,X|Z})$ is the solution set of $V_{*}(F_{Y,X|Z})$.

Now by Donsker's theorem \citep*[Theorem 19.3]{van2000asymptotic}, it holds that $\sqrt{n}(\hat{F}_{Y,X|Z=z;n}-F_{Y,X|Z=z})\rightsquigarrow \mathbb{G}_{F_{Y,X|Z=z}}$, where $\mathbb{G}_{F_{Y,X|Z=z}}$ is a Brownian bridge with covariance function 
\[\text{Cov}_{\mathbb{G}_{F_{Y,X|Z=z}}} = F_{Y,X|Z=z}(\min\{y,y'\},\min\{x,x'\})-F_{Y,X|Z=z}(y,x)F_{Y,X|Z=z'}(y',x')\] for all $(y,x),(y',x')\in[0,1]^2$ and $z\in[0,1]$.
Therefore, applying the functional delta method \citep*[Theorem 2.1]{shapiro1991asymptotic} directly yields that
\[\sqrt{n}(V_*(\hat{F}_{Y,X|Z=z;n})-V_{*}(F_{Y,X|Z=z}))\rightsquigarrow \delta_{F_{Y,X|Z}}V_{*}(\mathbb{G}_{F_{Y,X|Z=z}}).\] 
Putting both terms together gives
\[\sqrt{n}(\hat{\tilde{V}}_{*,l,\kappa}(\hat{F}_{Y,X|Z=z;n})-V_{*}(F_{Y,X|Z=z}))\rightsquigarrow \delta_{F_{Y,X|Z}}V_{*}(\mathbb{G}_{F_{Y,X|Z=z}}).\] The analogous result holds for $V^*$.
\end{proof}

\end{document}